\documentclass[a4paper,11pt]{article}
\pdfoutput=1 

\usepackage{jheppub} 

\usepackage[T1]{fontenc} 

\usepackage{amsthm} 

\newtheorem{mytheorem}{Property}
\newtheorem{lemma}{Lemma}

\newcommand{\g}{g}
\newcommand{\dd}{\gamma}
\newcommand{\ip}{\Gamma}
\newcommand{\ai}{a_{\infty}}

\usepackage{tikz}
\usetikzlibrary{math}
\usetikzlibrary{arrows,shapes,positioning}
\usetikzlibrary{decorations.markings}
\tikzstyle arrowstyle=[scale=1]
\tikzstyle directed=[postaction={decorate,decoration={markings,
    mark=at position .65 with {\arrow[arrowstyle]{stealth}}}}]
\tikzstyle reverse directed=[postaction={decorate,decoration={markings,
    mark=at position .65 with {\arrowreversed[arrowstyle]{stealth};}}}]

\newlength{\mywidth}
\usepackage[colorlinks=true
,urlcolor=blue
,anchorcolor=blue
,citecolor=blue
,filecolor=blue
,linkcolor=blue
,menucolor=blue
,pagecolor=blue
,linktocpage=true
,pdfproducer=medialab
,pdfa=true
]{hyperref}

\title{\boldmath One-loop elastic amplitudes from tree-level elasticity in 2d}

\author[a,b]{Matheus Fabri,}
\author[a,b]{Davide Polvara}

\affiliation[a]{Dipartimento di Fisica e Astronomia,
Universita degli Studi di Padova, via Marzolo 8, 35131 Padova, Italy.}
\affiliation[b]{INFN,
Sezione di Padova, via Marzolo 8, 35131 Padova, Italy.}

\emailAdd{matheusaugusto.fabri@unipd.it}
\emailAdd{davide.polvara@gmail.com}

\abstract{In this paper we extend the study initiated in~\cite{Polvara:2023vnx} to the computation of one-loop elastic amplitudes. We consider 1+1 dimensional massive bosonic Lagrangians with polynomial-like potentials and absence of inelastic processes at the tree level; starting from these assumptions we show how to write sums of one-loop diagrams as products and integrals of tree-level amplitudes. We derive in this way a universal formula for the one-loop two-to-two S-matrices in terms of tree S-matrices. 
We test our results on different integrable theories, such as sinh-Gordon, Bullough–Dodd and the full class of simply-laced affine Toda theories, finding perfect agreement with the bootstrapped S-matrices known in the literature. We show how Landau singularities in amplitudes are naturally captured by our universal formula while they are lost in results based on unitarity-cut methods implemented in the past~\cite{Bianchi:2013nra,Bianchi:2014rfa}.
}

\begin{document} 
\maketitle
\flushbottom

\section{Introduction}
\label{Introductory_section_on_conventions}

Integrable quantum field theories (IQFTs) in 1+1 dimensions are the hydrogen atoms of scattering amplitudes: they are simple enough to be completely solved~\cite{Zamolodchikov:1978xm} and provide a fascinating link between exact and perturbative methods for the computation of S-matrices. Even so, despite the tremendous progress made in the past years in the implementation of on-shell methods to compute amplitudes in four-dimensional gauge theories (see, e.g., \cite{Britto:2004ap, Britto:2004nc, Britto:2005fq, Bern:1994zx, Bern:1994cg}), very little is known about the application of these methods to two-dimensional theories\footnote{We should mention that partial results in this direction where obtained in~\cite{Bianchi:2013nra,Bianchi:2014rfa, Engelund:2013fja}. We will discuss these results later in the paper.}. Due to this fact, the bootstrap approach proposed in~\cite{Zamolodchikov:1978xm} is still the main way to attack the problem of deriving all-loop S-matrices of integrable theories in two dimensions. While this axiomatic program allowed in the past decades for the determination of the S-matrices of several integrable theories, 
there is now growing evidence that many S-matrices are hardly obtained through the bootstrap, as is the case for many non-linear sigma models defined on the worldsheet of superstrings (we remand to~\cite{Demulder:2023bux} for a recent review on this topic). Moreover, the S-matrices found through the bootstrap are only conjectured and the bootstrap axioms themselves are based on classical considerations which may be violated at the quantum level. It is then natural to ask whether there is an alternative way to derive all-loop S-matrices of integrable models unambiguously. While in this paper we do not provide a complete answer to this question we find partial results in this direction for all massive bosonic integrable theories with polynomial-like interactions.

More precisely, we consider Lagrangians of the following type
\begin{equation}
\label{eq0_1}
\mathcal{L}_0=\sum_{a=1}^r \biggl( \frac{\partial_\mu \phi_a  \partial^\mu \phi_{\bar{a}}}{2} - \frac{m_a^2}{2}  \phi_a \phi_{\bar{a}}\biggr) - \sum_{n=3}^{+\infty} \frac{1}{n!}\sum_{a_1,\dots, a_n=1}^r C^{(n)}_{a_1 \ldots a_n} \phi_{a_1} \ldots \phi_{a_n}
\end{equation}
and ask ourselves whether is possible to find closed expressions for the one-loop S-matrices of these models using perturbation theory.
Here $\{\phi_j \}_{j=1}^r$ are bosonic massive scalar fields associated with the asymptotic states of the theory. We follow the same convention of \cite{Dorey:2021hub, Polvara:2023vnx} and define $\phi_{\bar{a}} \equiv \phi^*_a$; this convention allows for the description of both real and complex scalars; indeed if $\phi_a$ is real we assume $a$ and $\bar{a}$ to be the same index $\in \{1, \dots, r\}$ so that the kinetic term is counted a single time. On the contrary, if $\phi_a$ is complex then we assume $a$ and  $\bar{a}$ to be different indices in $\{1, \dots, r\}$ and the kinetic term comes with the correct normalization for complex scalars. 
As in~\cite{Polvara:2023vnx}, we will assume models that satisfy the following property:
\begin{mytheorem}\label{Condition_tree_level_elasticity_introduction}
The theory is purely elastic at the tree level, i.e., the only tree-level amplitudes different from zero are those
in which the incoming and outgoing particles are of the same type and carry the same set of momenta.
\end{mytheorem}
Examples of IQFTs satisfying this property are the sinh-Gordon, Bullough-Dodd and affine Toda models, as universally proven in~\cite{Gabai:2018tmm, Dorey:2021hub}. An example of integrable theories violating this property is instead provided by the class of $O(N)$ sigma models since they have reflection and annihilation processes~\cite{Zamolodchikov:1978xm}.

For models satisfying Property~\ref{Condition_tree_level_elasticity_introduction} and presenting Lagrangians of type~\eqref{eq0_1} we derive the elastic one-loop two-to-two amplitudes purely in terms of the tree-level amplitudes in a closed form.
We remark that the Lagrangian \eqref{eq0_1} is the bare Lagrangian and that (as usual in perturbation theory) it is necessary to add counterterms to regularize the model. 
After fixing the renormalization conditions by requiring the absence of ultraviolet (UV) divergences and the vanishing of all one-loop inelastic amplitudes we show that the one-loop expressions we find for the elastic S-matrices
satisfy unitarity and correctly reproduce the Coleman-Thun singularities~\cite{Coleman:1978kk} expected from the bootstrap.
We also universally prove for the first time that the S-matrices of simply-laced affine Toda models conjectured in~\cite{Dorey:1990xa,Dorey:1991zp} are correct to one-loop order in perturbation theory.

The paper is structured as follows. In Section~\ref{sec:main_results_conventions} we set up our notation, review some facts about the tree-level perturbative integrability of these models, detail the renormalization procedure used and present our main result: the aforementioned formula for two-to-two one-loop S-matrices in terms of tree-level S-matrices for purely elastic theories. 
In Section~\ref{sec:one loop derivation} we provide the derivation of our formula for one-loop S-matrices giving details of the counterterm contributions.
In Section \ref{sec:integr_counterterms} we show how to fix the counterterms to ensure the vanishing of all one-loop production amplitudes in all models preserving the mass ratios at one loop; in the same section we show that our formula satisfies unitarity and encodes the correct Landau poles expected by Feynman diagrams of Coleman-Thun type. 
To conclude, in Section \ref{sec:SL-affine-toda} we derive one-loop S-matrices for the entire class of simply-laced affine Toda models and show that they match with the formula bootstraped in~\cite{Dorey:1990xa,Dorey:1991zp}.  We regard some technical details of the computations carried out in Section~\ref{sec:one loop derivation} to appendices~\ref{appendix_on_off_shell_limit_of_tree_level_amplitudes} and~\ref{app:on_shell_limit}. Finally, in Appendix~\ref{app:orbit_relation} we prove some root system relations necessary to derive the one-loop S-matrices of simply-laced affine Toda theories.

\section{Main results and overall methodology}
\label{sec:main_results_conventions}

In this section we set the notation, describe the renormalization procedure used and show our universal formula for one-loop amplitudes, whose derivation is provided in Section \ref{sec:one loop derivation}.

\subsection{Conventions}
In this paper, we will mostly write the momenta of the particles using light-cone components
\begin{equation}
p \equiv p_0 + p_1 \quad \text{and} \quad \bar{p} \equiv p_0 - p_1 \,,
\end{equation}
where $p_0$ and $p_1$ are the temporal and spatial components of the momenta, respectively. A particle of generic mass $m_j$ is on-shell if
\begin{equation}
    \label{light_cone_components_PhD_thesis}
    p=m_j e^{\theta_p} \hspace{5mm}, \hspace{5mm} \bar{p}=m_j e^{-\theta_p} \, ,
\end{equation}
and therefore $p\bar{p}=m_j^2$. The variable $\theta_p$ used to parameterise the momentum is called the `rapidity' and is real if the momentum is physical. For a pair of momenta $p$ and $q$ we will use the convention of labelling the difference between their rapidities by
\begin{equation}
\label{convention_on_difference_between_rapidities}
\theta_{pq} \equiv \theta_p - \theta_q \,.
\end{equation}

Scattering processes can be described in terms of amplitudes, by which we mean the sums over all connected on-shell Feynman diagrams and eventual counterterms, or equivalently in terms of S-matrices. We will refer to the S-matrix as the amplitude properly normalized and multiplied by the Dirac delta function of the overall energy-momentum conservation. Since the models under discussion are relativistic both the amplitudes and the S-matrices  are of difference form in the rapidities of the external particles and are connected by the following normalization factor
\begin{equation}
\label{S_matrix_amplitude_connection}
S_{ab}(\theta_{p p'})=\frac{M_{ab}(\theta_{p p'})}{4 m_a m_b \sinh{\theta_{p p'}}} \,.
\end{equation}
This proportionality factor comes from expressing the Dirac delta function of the overall energy-momentum conservation in terms of the rapidities. The masses $m_a$ and $m_b$ appearing in~\eqref{S_matrix_amplitude_connection} are the physical masses of the scattered particles; they correspond to the renormalised masses in the Lagrangian and are independent of the coupling with respect to the loop expansion is performed. 
We label by $M_{ab \to cd}$ an amplitude associated with a process having $a$ and $b$ as incoming particles and $c$ and $d$ as outgoing particles. If the initial particles are equal to the final ones we represent the amplitude as $M_{ab} \equiv M_{ab \to ab}$. The same convention is used for the S-matrices. Here we are interested in the perturbative computation of the S-matrix, whose expansion is given by 
\begin{equation}
\label{eq:1-loop-def}
S_{ab}(\theta)=1+S^{(0)}_{ab} (\theta)+ S^{(1)}_{ab} (\theta) + \cdots \,.
\end{equation}
Where $S^{(0)}_{ab} (\theta)$ and $S^{(1)}_{ab} (\theta)$ are the tree-level and one-loop S-matrices, respectively. We use the same notation for the perturbative expansion of the amplitude.

\subsection{Tree-level integrability}
\label{sec:tree_level_int}

Inspired by the review~\cite{Dorey:1996gd}, the tree-level perturbative integrability for Lagrangian theories of type~\eqref{eq0_1} was studied in different works~\cite{Khastgir:2003au,Gabai:2018tmm,Dorey:2021hub}. 
Here we review the results of~\cite{Dorey:2021hub,Gabai:2018tmm}, that establish the conditions on the masses and Lagrangian couplings for a theory of type~\eqref{eq0_1} to satisfy Property~\ref{Condition_tree_level_elasticity_introduction}; relying on these properties we will later investigate the perturbative integrability at one loop. 

Let us consider for simplicity a two-to-two inelastic process of the following type
\begin{equation}
\label{eq:inel_scattering}
a(p)+b(p') \to c(q)+d(q') \,,
\end{equation}
where $\{a, b\} \ne \{c, d\}$ and $p$, $p'$, $q$, and $q'$ represent the two-component momenta of the scattered particles. If we define the Mandelstam variables as follows 
\begin{equation}
\label{eq:Mandelstam}
\begin{split}
&s \equiv (p+p')^2= m^2_a + m_b^2 + 2 m_a m_b \cosh{\theta_{p p'}}\,,\\
&t \equiv (p-q')^2= m^2_a + m_d^2 - 2 m_a m_d \cosh{\theta_{p q'}}\,,\\
&u \equiv (p-q)^2= m^2_a + m_c^2 - 2 m_a m_c \cosh{\theta_{p q}}\,,
\end{split}
\end{equation}
then the tree-level amplitude associated with this process is given by
\begin{multline}
\label{eq:tree-level-scattering}
M_{ab \to cd}^{(0)}=-i \sum_{e \in s}  \ C_{ab\bar{e}}^{(3)}\ \frac{1}{s-m_e^2}\ C_{e\bar{c}\bar{d}}^{(3)} -i \sum_{j \in t} C_{a \bar{d} \bar{j}}^{(3)} \ \frac{1}{t-m_j^2} \ C_{j b \bar{c}}^{(3)}\\
-i \sum_{l \in u} C_{a \bar{c} \bar{l}}^{(3)} \ \frac{1}{u-m_l^2} \ C_{l b\bar{d}}^{(3)}- i C_{ab \bar{c} \bar{d}}^{(4)} \,.
\end{multline}
This is simply the sum of Feynman diagrams with particles propagating in the $s$-, $t$- and $u$-channel, and the quartic vertex, respectively. For the theory to be purely elastic at the tree level this amplitude needs to vanish for all the on-shell values of the momenta of the external particles. Let us consider  
the limit $s \to + \infty$, where the kinematics simplifies drastically; in this limit there are two branches of solutions satisfying the overall energy-momentum conservation: $\{s \to + \infty,t \to -\infty, u \to 0\}$ and  $\{s \to + \infty,t \to 0, u \to -\infty\}$. The requirement that the amplitude must vanish in both branches leads to the following constraint
\begin{equation}
\label{eq0_4point}
C^{(4)}_{ab\bar{c}\bar{d}}= \sum_{j \in t} C_{a \bar{d} \bar{j}}^{(3)} \ \frac{1}{m_j^2} \ C_{j b \bar{c}}^{(3)} = \sum_{l \in u} C_{a \bar{c} \bar{l}}^{(3)} \ \frac{1}{m_l^2} \ C_{l b\bar{d}}^{(3)} \, ,
\end{equation}
linking the values of masses, cubic- and quartic-couplings.
This simple argument was recently generalised 
to inelastic processes with arbitrary numbers of external particles by using a particular high-energy limit of the kinematics~\cite{Bercini:2018ysh,Gabai:2018tmm}, through
which it was noted that a necessary condition for a tree-level $n$-point production amplitude to vanish is
\begin{equation}
\label{eq0_6}
\begin{split}
C^{(n)}_{b_1 \ldots b_n}&- \sum_l C^{(n-1)}_{b_1 \ldots b_{n-2} \bar{l}} \frac{1}{m^2_l}C^{(3)}_{l  b_{n-1}b_{n}}- \sum_e C^{(n-2)}_{b_1 \ldots b_{n-3} \bar{e}} \frac{1}{m^2_e} C_{e b_{n-2} b_{n-1} b_n}^{(4)}\\
&+ \sum_{l, s} C^{(n-2)}_{b_1 \ldots b_{n-3} \bar{s}} \frac{1}{m^2_s} C_{s b_{n-2} \bar{l}}^{(3)} \frac{1}{m_l^2} C_{l b_{n-1} b_n}^{(3)}=0 \, \ \ \ \textrm{for} \ \ \ n\geq5.
\end{split}
\end{equation}
With these relations, the space of Lagrangians of type~\eqref{eq0_1} which yields purely elastic S-matrices at the tree level is completely determined by the masses and $3$-point couplings. It is important to stress that only particular sets of masses and $3$-point couplings allow for equations~\eqref{eq0_4point} and \eqref{eq0_6} to be solved. The space of masses and $3$-point couplings from which the theory is generated by applying recursively~\eqref{eq0_4point} and \eqref{eq0_6} can be obtained by requiring the cancellation of all singularities in inelastic $4$- and $5$-point processes. We remand to~\cite{Dorey:2021hub} for a detailed explanation of all the necessary constraints for the theory to be purely elastic at the tree level. 
Early studies on the absence of production for similar models date back to~\cite{Goebel:1986na, Arefeva:1974bk}.

Proving the vanishing of production amplitudes in perturbation theory is an interesting problem, whose understanding would lead to a deeper comprehension of why these models are integrable and would help classify integrable models. Moving from the tree-level to one-loop different difficulties arise, mainly due to the sheer complexity involved in enumerating and computing the loop diagrams.
One-loop production processes for the class of theories defined by Lagrangians of the type in~\eqref{eq0_1} have been recently studied in~\cite{Polvara:2023vnx} and it was shown that the absence of production at the tree-level can be used to drastically simplify one-loop computations. In the following, we will use the same renormalization procedure adopted in~\cite{Polvara:2023vnx} and we will show how the method applied in that paper can be extended to the computation of one-loop elastic amplitudes.

\subsection{Renormalization conditions and counterterms}
\label{sec:ren_condition_123}

To compute one-loop amplitudes we work with renormalized perturbation theory. We should stress that the fields, masses and couplings appearing in the Lagrangian~\eqref{eq0_1} are the bare quantities. Thus in the rewriting of~\eqref{eq0_1} in terms of renormalized fields and physical parameters, we introduce a series of counterterms that we fix by imposing the following set of renormalization conditions:
\begin{itemize}
    \item[(1)] All amplitudes need to be free of UV divergences at one-loop;
    \item[(2)] For all pairs of indices $a, b \in \{1, \dots, r\}$ we require the propagator $G_{ab}(p^2)$, having as incoming leg $a$ and outgoing leg $b$, to satisfy
    \begin{equation}
    \label{eq:res_prop_rencond}
\text{Res} \, G_{ab}(p^2)\Bigl|_{p^2=m^2_a}=\text{Res} \, G_{ab}(p^2)\Bigl|_{p^2=m^2_b}= i \delta_{ab} \,,
    \end{equation}
    where $m_a$ and $m_b$ are the renormalised masses;
    \item[(3)] All inelastic and production amplitudes need to vanish at one-loop for all values of on-shell momenta\footnote{By on-shell we mean that the momenta squared are equal to the renormalized masses.}.
\end{itemize}
The first two conditions listed above are standard requirements from renormalized perturbation theory; the latter is instead because we want our model to be integrable and possess elastic scattering to one-loop. Note that the divergent parts of all counterterms are unambiguously fixed by condition (1) and remove all diagrams with self-contracted vertices, as detailed in \cite{Polvara:2023vnx}. Due to this fact, we keep into account these counterterms by simply neglecting all diagrams with vertices containing selfcontracted propagators and
from now on we refer to the finite part of the counterterms as counterterms. The expressions we will write are therefore always finite and free of UV divergences.

As aforementioned, it is necessary to rewrite the bare Lagrangian~\eqref{eq0_1} in terms of the renormalized parameters. Here we follow the renormalization scheme used in~\cite{Braden:1991vz,Polvara:2023vnx} and replace the bare quantities with the physical quantities, summed with the associated counterterms
\begin{subequations}
\label{renormalization_m_C_phi}
\begin{equation}
\label{renormalization_m}
m^2_a \to m^2_a+ \delta m_a^2 \, ,
\end{equation}
\begin{equation}
\label{renormalization_C}
C_{a_1 \dots a_n}^{(n)} \to C_{a_1 \dots a_n}^{(n)}+ \delta C_{a_1 \dots a_n}^{(n)} \, ,
\end{equation}
\begin{equation}
\label{renormalization_phi}
\phi_a \to (1+\delta Z_a) \phi_a - \sum_{\substack{b=1 \\ b\neq a}}^r t_{ba} \phi_b \,.
\end{equation}
\end{subequations}
Since the original theory contains an infinite number of couplings determined by applying recursively relation~\eqref{eq0_6} then we need infinitely many counterterms. This justifies the introduction of an infinite number of renormalization conditions in point (3).
After substituting~\eqref{renormalization_m_C_phi} into~\eqref{eq0_1} we obtain the renormalized Lagrangian 
\begin{equation}
\label{eq:renormalised_Lagrangian}
\begin{split}
    \mathcal{L}&=\sum_{a=1}^r \biggl( \frac{\partial_\mu \phi_a  \partial^\mu \phi_{\bar{a}} }{2} - \frac{m_a^2}{2}  \phi_a \phi_{\bar{a}}\biggr) - \sum_{n=3}^{+\infty} \frac{1}{n!}\sum_{a_1,\dots, a_n=1}^r C^{(n)}_{a_1 \ldots a_n} \phi_{a_1} \ldots \phi_{a_n}\\
    & \begingroup\color{blue} +\sum_{a=1}^r\delta Z_{a} \Bigl(\partial_\mu \phi_a  \partial^\mu \phi_{\bar{a}}-  m^2_a \phi_a \phi_{\bar{a}}\Bigr) +\sum_{\substack{a,b=1 \\ b\neq a}}^r t_{ba} \Bigl(-\partial_\mu \phi_{\bar{a}} \partial^{\mu} \phi_b+m^2_a  \phi_{\bar{a}} \phi_b \Bigr) \endgroup \\
    &\begingroup\color{blue}- \sum_{n=3}^{+\infty} \frac{1}{n!}\sum_{a_1,\dots, a_n=1}^r C^{(n)}_{a_1 \ldots a_n} \Bigl( \delta Z_{a_1}+\delta Z_{a_2} + \ldots \delta Z_{a_n}\Bigr) \phi_{a_1} \ldots \phi_{a_n} \endgroup \\
    &\begingroup\color{blue}+\sum_{n=3}^{+\infty} \frac{1}{(n-1)!}\sum_{a_1,\dots, a_n=1}^r C^{(n)}_{a_1 \ldots a_n} \sum_{\substack{b=1 \\ b\neq a_1}}^r t_{b a_1} \phi_{b} \phi_{a_2} \ldots \phi_{a_n} -\sum_{a=1}^r\frac{1}{2} \delta m^2_a \phi_a \phi_{\bar{a}} \endgroup\\
    &{\color{red}- \sum_{n=3}^{+\infty} \frac{1}{n!}\sum_{a_1,\dots, a_n=1}^r  \delta C^{(n)}_{a_1 \ldots a_n} \phi_{a_1} \ldots \phi_{a_n} } \,.\\
\end{split}
\end{equation}
This will be the Lagrangian used in the computation of one-loop amplitudes in the next section.

The counterterms $\delta m_a^2$, $\delta Z_a$ and $t_{ab}$ can be fixed by imposing the renormalization condition~\eqref{eq:res_prop_rencond}. Indeed let $\Sigma_{jj}(p^2)$ be the sum over all one-loop bubble diagrams having two external legs of type $j$, which is shown in Figure~\ref{Im:Bubble_diagrams_j_particle}, and
\begin{equation}
\label{coefficients_of_the_Sigma_expansion}
\Sigma^{(0)}_{jj} \equiv \Sigma_{jj}(m_j^2) \hspace{5mm} \text{,} \hspace{5mm} \Sigma^{(1)}_{jj} \equiv \frac{\partial}{\partial p^2}\Sigma_{jj}(p^2)\Bigl|_{p^2=m_j^2} \,.
\end{equation}
Then condition (2) requires
\begin{equation}
\label{definition_of_mass_renormalization_and_tab_to_diagonalise_the_mass_matrix}
\delta m^2_a = - \Sigma^{(0)}_{aa} \hspace{3mm},\hspace{3mm} \delta Z_a = \frac{1}{2} \Sigma^{(1)}_{aa} \hspace{3mm} \text{and} \hspace{3mm}
    t_{ab} = \begin{cases}
\frac{\Sigma^{(0)}_{ab}}{m^2_b - m^2_a} \ \ \text{if} \ \ m_a \ne m_b\\
0 \hspace{12mm} \text{if} \ \ m_a = m_b \,.
\end{cases}
\end{equation}
A detailed derivation of these results is reported in~\cite{Braden:1991vz,Polvara:2023vnx}. 

While condition (2) is enough to fix the counterterms in~\eqref{renormalization_m} and \eqref{renormalization_phi}, to determine the counterterms~\eqref{renormalization_C} arising from the renormalization of the couplings we need to impose condition (3).
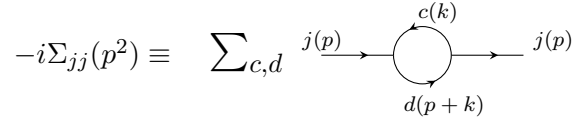
\begin{figure}[t]
\begin{center}
\begin{tikzpicture}
\tikzmath{\y=1.9;}


\filldraw[black] (-2.8*\y,0.1*\y)  node[anchor=west] {$-i \Sigma_{jj}(p^2) \equiv $};
\filldraw[black] (-1.45*\y,0.1*\y)  node[anchor=west] {\Large{$\sum_{c, d}$}};
\draw[directed] (-0.6*\y,0.1*\y) -- (-0.1*\y,0.1*\y);
\draw[directed] (0.3*\y,0.1*\y) arc(0:180:0.2*\y);
\draw[directed] (-0.1*\y,0.1*\y) arc(180:360:0.2*\y);
\draw[directed] (0.3*\y,0.1*\y) -- (0.8*\y,0.1*\y);
\filldraw[black] (-0.8*\y,0.2*\y)  node[anchor=west] {\scriptsize{$j(p)$}};
\filldraw[black] (0.8*\y,0.2*\y)  node[anchor=west] {\scriptsize{$j(p)$}};
\filldraw[black] (0*\y,0.4*\y)  node[anchor=west] {\scriptsize{$c(k)$}};
\filldraw[black] (-0.1*\y,-0.25*\y)  node[anchor=west] {\scriptsize{$d(p+k)$}};

\end{tikzpicture}
\caption{Sum of bubble diagrams with external leg $j$. Here we integrate over inner momentum $k$ and sum over exchanged particles $c$ and $d$.}
\label{Im:Bubble_diagrams_j_particle}
\end{center}
\end{figure}
One might expect this condition to impose multiple relations between distinct counterterms $\delta C_{a_1 \dots a_n}^{(n)}$.
However, it is important to recall that the couplings are not all independent due to the constraints~\eqref{eq0_4point} and~\eqref{eq0_6} determined by the tree-level integrability of the theory.
It is then reasonable to expect all counterterms to be related to each other and only a finite number of parameters to enter the renormalization conditions.
This is indeed the case of simply-laced affine Toda models~\cite{Braden:1989bu,Christe:1989my} where all masses and coupling counterterms scale with a common multiplicative factor
\begin{subequations}
\label{eq:equal_gamma_convention}
\begin{equation}
\label{eq:equal_gamma_convention1}
\delta m^2_a = \gamma\ m^2_a \, ,
\end{equation}
\begin{equation}
\label{eq:equal_gamma_convention2}
\delta C_{a_1 \dots a_n}^{(n)} = \gamma\ C_{a_1 \dots a_n}^{(n)}\, .
\end{equation}
\end{subequations}
Here the constant $\gamma$ depends neither on $n$ nor on the types of particles entering the vertices. Thus in these models, there is a single independent parameter and the counterterms in \eqref{eq:renormalised_Lagrangian} can be introduced altogether by modifying the mass scale in front of the integrable potential whose expansion generates the Lagrangian~\eqref{eq0_1}. Even if it is not proven that all integrable Lagrangians of type~\eqref{eq0_1} are of affine Toda type it is reasonable to expect that in all theories satisfying the tree-level elasticity condition defined by Property~\ref{Condition_tree_level_elasticity_introduction} the counterterms are introduced by renormalizing a finite number of parameters. This would be a consequence of all couplings being connected at the tree level. We will return to this point in Section~\ref{sec:integr_counterterms}.

\subsection{One-loop amplitudes}

Given the renormalized Lagrangian~\eqref{eq:renormalised_Lagrangian}, we want to analyze on-shell processes of the form
\begin{equation}
\label{two_to_two_elastic_process}
a(p)+b(p') \to a(q)+ b(q')
\end{equation}
on the elastic branch of the kinematics $\{q=p , \,q'=p'\}$. The one-loop amplitude associated with the process~\eqref{two_to_two_elastic_process} can be written as
\begin{equation}
\label{eq:definition_one_loop_amplitude}
M^{(1)}_{ab} = M^{(1\text{-loop})}_{ab}+{\color{blue} M^{(\text{ct.I})}_{a b} }+ {\color{red}M^{(\text{ct.II})}_{a b} } \,.
\end{equation}
Here $M^{(1\text{-loop})}_{ab}$ is the sum of all one-loop diagrams obtained by applying standard Feynman rules to the first-row of~\eqref{eq:renormalised_Lagrangian}, apart from the diagrams with self contracted vertices which are cancelled by the divergent part of the counterterms (which we do not write as they play no role apart from cancelling diagrams with self contracted vertices). We stress that this term contains both one-particle reducible and one-particle irreducible diagrams (and therefore also diagrams with loops in the external legs).
The remaining blue and red terms in~\eqref{eq:definition_one_loop_amplitude} contain instead tree-level diagrams with counterterms; they are obtained by applying Feynman rules to the red and blue interaction terms in the Lagrangian~\eqref{eq:renormalised_Lagrangian}, respectively. 

To describe our closed formula for one-loop amplitudes, we need the following result which is proven in Appendix~\ref{app:on_shell_limit}:
\begin{mytheorem}\label{Property_4_point_couplings}
If a theory has a Lagrangian of type~\eqref{eq0_1} and satisfies Property~\ref{Condition_tree_level_elasticity_introduction} then for any pair of particles of types $i$ and $j \in \{1,\dots, r\}$ 
the value of $M^{(0)}_{i j}$ at infinite rapidity is given by
\begin{equation}
\label{values_of_tree_level_amplitudes_at_infty}
\frac{M^{(0)}_{i j} (\infty)}{m_i^2 m_j^2}= \ai \,,
\end{equation}
where $\ai$ is a constant independent of $i$ and $j$. 
\end{mytheorem}

In Section~\ref{sec:one loop derivation} we will prove that elastic scattering processes of type~\eqref{two_to_two_elastic_process}, with rapidities $\theta_p>\theta_{p'}$, must satisfy the following relation
\begin{equation}
\label{eq:final_result_amplitude}
\begin{split}
&M^{(1\text{-loop})}_{ab}(\theta_{pp'})+ M^{(\text{ct.I})}_{a b}(\theta_{pp'})=\\
& \sum_{j \in \{\text{prop, ext} \}} \delta m_j^2 \frac{\partial}{\partial m_j^2} M^{(0)}_{ab}(\theta_{p p'})-\frac{1}{2} \Bigl( \frac{\delta m_a^2}{m_a^2} +\frac{\delta m_b^2}{m_b^2} \Bigl) M^{(0)}_{ab}(\theta_{p p'} )\\
&+\frac{\bigl(M^{(0)}_{ab}(\theta_{p p'}) \bigl)^2}{8 m_a m_b \sinh{\theta_{p p'}}}+\frac{i \ai}{2 \pi} m_a m_b \sinh{\theta_{p p'}} \theta_{p p'} \frac{\partial}{\partial \theta_{p p'}} S^{(0)}_{ab}(\theta_{p p'}) \sum^r_{e=1} m_e^2 \\
&+\frac{i}{\pi} m_a m_b \sinh{\theta_{pp'}}  \sum_{e=1}^r \frac{\partial}{\partial \theta_{p p'}} \ \text{p.v.} \int^{+\infty}_{-\infty} d\theta_k \,  S^{(0)}_{ae} (\theta_{pk} ) S^{(0)}_{b e} (\theta_{p' k}) \,.
\end{split}
\end{equation}
By ``p.v.'' we denote the Cauchy principal value prescription, which is necessary to avoid collinear singularities at $\theta_k=\theta_p$ and $\theta_k=\theta_{p'}$, as we should expect from the normalization factor in~\eqref{S_matrix_amplitude_connection}. This prescription means that we are integrating on the domain
\begin{equation}
(-\infty, \theta_{p'}-\epsilon) \cup (\theta_{p'}+\epsilon, \theta_{p}-\epsilon) \cup (\theta_{p}+\epsilon, +\infty) \quad \text{with} \ \epsilon > 0 \,.
\end{equation}
In the limit $\epsilon\rightarrow 0$, that we take after having performed the integral, the result is convergent because the tree-level S-matrices are odd functions in the rapidity. The constant $\ai$ in the third line of~\eqref{eq:final_result_amplitude} was defined in~\eqref{values_of_tree_level_amplitudes_at_infty}; it depends on the theory under discussion but not on the types of scattered particles $a$ and $b$. The mass counterterms in the second line of~\eqref{eq:final_result_amplitude} are given in~\eqref{definition_of_mass_renormalization_and_tab_to_diagonalise_the_mass_matrix} and are completely fixed in terms of one-loop bubble diagrams of the form in Figure~\ref{Im:Bubble_diagrams_j_particle}.  The sum in the first term on the r.h.s. of~\eqref{eq:final_result_amplitude} is performed over the masses of the virtual particles propagating inside the Feynman diagrams contributing to $M^{(0)}_{ab}$ and over the masses $m_a$ and $m_b$ of the external particles\footnote{This term must be computed by considering the masses of propagators and external particles appearing in $M^{(0)}_{ab}$ as free parameters, which need to be fixed to be equal to the renormalized masses appearing in the Lagrangian~\eqref{eq0_1} only after having performed the derivatives.}.

We stress that formula~\eqref{eq:final_result_amplitude} is written in terms of tree amplitudes and can be used to generate one-loop S-matrices of integrable theories in terms of tree-level data. As stressed in~\cite{Polvara:2023vnx}, the production amplitudes obtained from the black and blue terms in the Lagrangian~\eqref{eq:renormalised_Lagrangian}, do not vanish at one-loop in general. Then to enforce no particle production (our third renormalization condition), the additional red counterterms $\delta C^{(n)}_{a_1 \cdots a_n}$ in~\eqref{eq:renormalised_Lagrangian} are necessary. Due to this fact, from now on we will often refer to these counterterms as `integrable counterterms'. The contribution of these counterterms to the amplitude is 
\begin{equation}
\label{eq:countertermsII}
\begin{split}
M_{ab}^{(\text{ct. II})}=&-i \sum_{i \in s}  \frac{\delta C_{ab\bar{i}}^{(3)} C_{i \bar{a}\bar{b}}^{(3)}+ C_{ab\bar{i}}^{(3)} \delta C_{i \bar{a}\bar{b}}^{(3)}}{s-m_i^2}-i \sum_{j \in t} \frac{\delta C_{a \bar{b} \bar{j}}^{(3)} C_{j b \bar{a}}^{(3)}+ C_{a \bar{b} \bar{j}}^{(3)} \delta C_{j b \bar{a}}^{(3)}}{t-m_j^2}\\
&+i \sum_{l \in u}  \frac{\delta C_{a \bar{a} \bar{l}}^{(3)} C_{l b\bar{b}}^{(3)}+C_{a \bar{a} \bar{l}}^{(3)} \delta C_{l b\bar{b}}^{(3)}}{m_l^2} - i \delta C_{ab \bar{a} \bar{b}}^{(4)} \,.
\end{split}
\end{equation}
While the problem of fixing these counterterms for general theories remains open, in Section~\ref{sec:integr_counterterms} we show how these counterterms can be fixed in all models having mass ratios that do not renormalise at one loop (i.e. they satisfy~\eqref{eq:equal_gamma_convention1}). For these models condition~\eqref{eq:equal_gamma_convention2} is enough to remove all one-loop inelastic amplitudes and the one-loop elastic S-matrices are given by
\begin{equation}
\label{eq:result_S_mat_eq_m_ren_2Int}
\begin{split}
S^{(1)}_{ab}(\theta_{pp'})&=\frac{\bigl(S^{(0)}_{ab}(\theta_{p p'}) \bigl)^2}{2}+\frac{i \ai}{8 \pi}  \theta_{p p'} \frac{\partial}{\partial \theta_{p p'}} S^{(0)}_{ab}(\theta_{p p'}) \sum^r_{e=1} m_e^2 \\
&+\frac{i}{4 \pi}  \sum_{e=1}^r \frac{\partial}{\partial \theta_{p p'}} \ \text{p.v} \int^{+\infty}_{-\infty} d\theta_k \   S^{(0)}_{ea} (\theta_{k p} ) S^{(0)}_{e b} (\theta_{k p'}) \,.
\end{split}
\end{equation}

\section{Deriving closed formulas for one-loop amplitudes}
\label{sec:one loop derivation}

In this section, and the associated appendices~\ref{appendix_on_off_shell_limit_of_tree_level_amplitudes} and~\ref{app:on_shell_limit}, we report the derivation of formula~\eqref{eq:final_result_amplitude} for one-loop amplitudes. To this end, we follow the same technique adopted in~\cite{Polvara:2023vnx} and described in some detail in chapter 24 of~\cite{schwartz2014quantum}. We will often refer to this technique with the name `cutting method'. This consists in splitting each Feynman propagator into a retarded propagator and a Dirac delta function as follows
\begin{equation}
\label{splitting_propagators_into_retarded_part_and_delta_function}
\frac{i}{k^2-m^2_a+i \epsilon} =\Pi_a^{(R)} (k) + \frac{\pi}{\omega_a(k)} \delta(k_0 - \omega_a(k)) \,.
\end{equation}
In the expression above
\begin{equation}
\label{definition_of_the_retarded_propagator}
\Pi_a^{(R)} (k) \equiv \frac{i}{k^2-m_a^2- i k_0 \ \epsilon}\ ,
\end{equation}
is the retarded propagator associated with a particle of type $a \in\{1, \dots, r\}$ and
$$
\omega_a(k) \equiv \sqrt{k_1^2 + m_a^2}
$$
is the on-shell energy of the particle. Using this split, any loop diagram can be decomposed into the sum of different terms corresponding to combinations of Dirac-delta functions and retarded propagators. Among them, the nonzero terms contributing to the amplitude must contain at least one Dirac-delta function. Indeed, in contrast to Feynman propagators, retarded propagators have pairs of poles located in the same half of the energy complex plane; therefore, an integral containing just retarded propagators is zero by the existence of a contour in the complex plane that does not enclose any poles. Due to this fact, each loop amplitude can be completely written in terms of tree-level quantities, which can be products or integrals of on-shell tree-level amplitudes. We refer the reader to~\cite{Polvara:2023vnx} for a detailed example of how this procedure works. 

With these preliminaries, we now move to the derivation of formula~\eqref{eq:final_result_amplitude}.

\subsection{Standard counterterm contributions}
\label{section_to_explain_the_mass_and_field_renormalization}

We start with the derivation of $M^{(\text{ct.I})}_{a b}(\theta_{pp'})$. In~\cite{Polvara:2023vnx} it was shown that the sum of all the bubble diagrams contributing to a one-loop process of type $a(p) \to a(p)$ can be written as follows
\begin{equation}
\label{definition_necessary_for_mass_renormalization_first_time_ga_appears}
-i\Sigma_{aa}(p^2)= \frac{1}{8 \pi} \sum_{c=1}^r \int_{-\infty}^{+\infty} d\theta_k \ \hat{M}^{(0)}_{a c}(\theta_{pk}) \, ,
\end{equation}
where $\Sigma_{aa}(p^2)$ is the quantity depicted in Figure~\ref{Im:Bubble_diagrams_j_particle} and $\hat{M}^{(0)}_{a c}$ is the two-to-two tree-level amplitude $M^{(0)}_{a c}$ associated with the elastic process
$$
a(p)+c(k) \to a(p)+c(k)
$$
from which we subtract the contribution at $\theta_{k}=\infty$ (being $\theta_k$ the rapidity of the particle $c$ associated with the cut propagator):
\begin{equation}
\hat{M}^{(0)}_{a c}(\theta_{p k}) \equiv M^{(0)}_{a c}(\theta_{p k}) - M^{(0)}_{a c}(\infty) \,.
\end{equation}

Using these definitions for the sum of the bubble diagrams we can write $M^{(\text{ct.I})}_{a b}$ as follows
\begin{equation}
\label{total_counterterms_contribution_production_process}
\begin{split}
    M^{(\text{ct.I})}_{a b} &=\Bigl( \frac{2\delta m^2_a}{p^2 - m^2_a} + \frac{2\delta m^2_b}{p'^2 - m^2_b}\Bigl) M^{(0)}_{a b}(p, p')\\
    &- (\Sigma^{(1)}_{a a}+\Sigma^{(1)}_{b b}) M^{(0)}_{a b}(p, p')+\sum_{j \in \text{prop.}} \delta m^2_j \frac{\partial}{\partial m^2_j} M^{(0)}_{a b}(p, p').
    \end{split}
\end{equation}
The derivation of this result was performed in~\cite{Polvara:2023vnx} and we omit to report it again here.
The result is identical to the one in Section 2.4 of~\cite{Polvara:2023vnx} but for the additional contribution
$$
(\Sigma^{(1)}_{a a}+\Sigma^{(1)}_{b b}) M^{(0)}_{a b}(p, p'),
$$
which was zero in~\cite{Polvara:2023vnx} since the processes considered in that work were inelastic. Note that the first row in~\eqref{total_counterterms_contribution_production_process} is singular when the external momenta are on-shell. This singular contribution will be necessary to cancel specific ill-defined one-loop Feynman diagrams containing bubbles in external legs. It is indeed associated to one-loop external leg corrections.
Note that now the sum in the last term on the r.h.s. of~\eqref{total_counterterms_contribution_production_process} is performed \textit{only} over the masses appearing in the internal propagators in the tree-level process $M^{(0)}_{a b}$ and not on the external masses like in \eqref{eq:final_result_amplitude}.

\subsection{One-loop contributions}
\label{section_to_explain_two_to_two_inelastic_processes}

We now focus on one-loop diagrams obtained by applying standard Feynman rules to the first row of the Lagrangian~\eqref{eq:renormalised_Lagrangian}. As shown in \cite{Polvara:2023vnx}, using the cutting method previously described, the sum of all one-loop diagrams contributing to the process~\eqref{two_to_two_elastic_process} can be written as
\begin{equation}
\label{sum_of_single_and_double_cut_contributions}
M_{ab}^{(1\text{-loop})}=M_{ab}^{(2\text{-cut})}+M_{ab}^{(1\text{-cut})} \,,
\end{equation}
where
\begin{equation}
\label{4_point_one_loop_from_tree_single_cut_contribution}
M_{ab}^{(1\text{-cut})} =\frac{1}{8\pi} \sum_{e=1}^r \int_{-\infty}^{+\infty} d\theta_k \hat{M}_{abe}^{(0)}(p,p',k),
\end{equation}
and
\begin{equation}
\label{4_point_one_loop_from_tree_double_cut_contribution}
M_{ab}^{(2\text{-cut})} =M_{ab}^{(2\text{-cut},t)}+M_{ab}^{(2\text{-cut},u)}= \sum_{e, f=1}^r \frac{M_{ae\to bf}^{(0)}M_{bf\to ae}^{(0)}}{8 m_e m_f |\sinh{\theta_{k k'}}|}+\sum_{e, f=1}^r \frac{M_{ae\to af}^{(0)}M_{bf\to be}^{(0)}}{8 m_e m_f  |\sinh{\theta_{k k'}}|} \,.
\end{equation}
We refer to equations~\eqref{4_point_one_loop_from_tree_single_cut_contribution} and~\eqref{4_point_one_loop_from_tree_double_cut_contribution} as the single- and double-cut contributions, respectively. As before, the hat on the integrands of~\eqref{4_point_one_loop_from_tree_single_cut_contribution} means that we subtract from the tree-level amplitudes their values at $\theta_k=\infty$ ($\theta_k$ is the rapidity associated with the particle $e$). This subtraction corresponds to not considering diagrams with propagators that connect back to their originating vertices, which are removed thanks to the counterterms introduced by requiring point (1) in the renormalization conditions in Section~\ref{sec:main_results_conventions}. Note that in~\eqref{4_point_one_loop_from_tree_double_cut_contribution} $k$ and $k'$ are the momenta associated with the particles $e$ and $f$. 
These momenta are associated with the propagators that have been cut and are fixed by the requirement that the energy and momentum are conserved in each blob of Figure~\ref{All_possible_double_cuts_ab_into_cd}, where the double-cut contributions for the $s$-, $t$- and $u$-channels are shown. We start computing~\eqref{4_point_one_loop_from_tree_double_cut_contribution} and then we focus on the more complicated single-cut contribution in~\eqref{4_point_one_loop_from_tree_single_cut_contribution}.

\subsection{Double-cuts in elastic amplitudes}
\label{sec:double-cut-elastic}

The $s$-, $t$-, and $u$-channel double-cuts contributing to equation~\eqref{4_point_one_loop_from_tree_double_cut_contribution} are shown in Figure~\ref{All_possible_double_cuts_ab_into_cd}.  As explained in~\cite{Polvara:2023vnx} the double cut in the $s$-channel, depicted on the left of the figure, is null by kinematics, otherwise the process $f\to a+b+e$ would be allowed on-shell and Property~\ref{Condition_tree_level_elasticity_introduction} would be violated. For this reason, in~\eqref{4_point_one_loop_from_tree_double_cut_contribution} only double-cuts in the $u$ and $t$ channels are considered. 

We evaluate the double cut in the $u$-channel first. 
By Property~\ref{Condition_tree_level_elasticity_introduction} the only nonzero $u$-channel double-cuts are those in which $f=e$. Indeed this is the only possibility for the two blobs in the third picture in Figure~\ref{All_possible_double_cuts_ab_into_cd} to be elastic amplitudes. To correctly define these cuts
a regulator
\begin{equation}
\label{definition_of_my_regulator_x}
x=\mu(e^{\theta_x},e^{-\theta_x}) \, ,
\end{equation}
with $\mu>0$ and $\theta_x \in \mathbb{R}$, needs to be introduced. The vector in~\eqref{definition_of_my_regulator_x} is expressed in light-cone components. We assume $p$ and $p'$ fixed and on-shell, and $q$ and $q'$ defined by
\begin{equation}
\label{regularization_through_x_introducing_outgoing_mass_deformations}
q-p=p'-q'=x.
\end{equation}
From~\eqref{regularization_through_x_introducing_outgoing_mass_deformations} we see that $q$ and $q'$ are off-shell and become on-shell only in the limit $\mu \to 0$ where $q = p$ and $q' = p'$. With this regulator the momenta of the cut propagators need to satisfy
\begin{equation}
\label{difference_between_hatk_and_tildek_returning_x}
k- k'=x.
\end{equation}
Equation~\eqref{difference_between_hatk_and_tildek_returning_x} admits no physical solutions for $k$ and $k'$ when $\mu>0$ and $\theta_x \in \mathbb{R}$; the $u$-channel cuts are therefore all null due to kinematics. Since the double cuts in the $u$ channel are all null when $\mu>0$ then they need to be null also in the limit $\mu \to 0$ (when all the external particles become on-shell). Indeed, no discontinuity in $M_{ab}^{(1)}$ is expected in this limit. This is consistent as far as the same regulator is used to compute also the expression in~\eqref{4_point_one_loop_from_tree_single_cut_contribution}. 

Therefore, the only nonzero term in~\eqref{4_point_one_loop_from_tree_double_cut_contribution} is the $t$-channel contribution, which is depicted at the centre of Figure~\ref{All_possible_double_cuts_ab_into_cd} and corresponds to the first sum in~\eqref{4_point_one_loop_from_tree_double_cut_contribution}. The non-vanishing contributions in this sum are those in which $e=b$ and $f=a$: in this way both the tree-level amplitudes in the numerator of the first sum in~\eqref{4_point_one_loop_from_tree_double_cut_contribution} are elastic. Solving the energy-momentum conservation constraints for these values we obtain
\begin{equation}
\label{eq_double_cut_contribution_elastic}
M_{ab}^{(2\text{-cut})} =M_{ab}^{(2\text{-cut},t)}= \frac{(M_{ab}^{(0)}(\theta_{p p'}))^2}{8 m_a m_b  |\sinh{\theta_{p p'}}|}.
\end{equation}
While double-cut contributions are relatively simple to compute, diagrams containing single cuts (i.e. a single Dirac delta function and all remaining propagators of retarded type) are quite convoluted. We will spend the remaining part of this section to compute these diagrams.

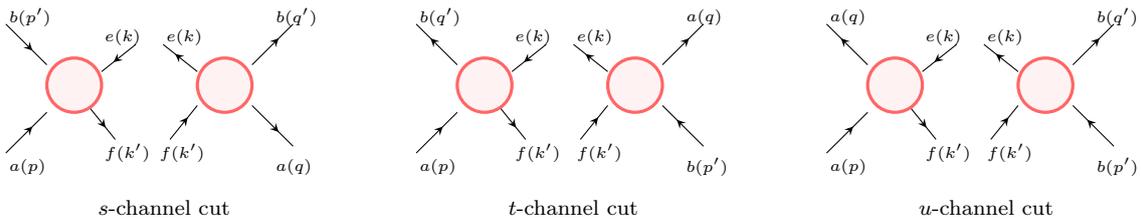
\begin{figure}
\begin{center}
\begin{tikzpicture}
\tikzmath{\y=0.9;}


\draw[directed] (-1.4*\y,0.5*\y) -- (-0.8*\y,1.1*\y);
\draw[directed] (-1.4*\y,2.4*\y) -- (-0.8*\y,1.8*\y);
\draw[directed] (0.5*\y,2.1*\y) -- (0*\y,1.7*\y);
\draw[directed] (1.4*\y,1.7*\y) -- (0.9*\y,2.1*\y);
\draw[directed] (-0.2*\y,1.2*\y) -- (0.2*\y,0.7*\y);
\draw[directed] (1*\y,0.7*\y) -- (1.4*\y,1.2*\y);
\draw[directed] (2.2*\y,1.2*\y) -- (2.8*\y,0.6*\y);
\draw[directed] (2.2*\y,1.8*\y) -- (2.8*\y,2.4*\y);

\filldraw[color=red!60, fill=red!5, very thick](-0.4*\y,1.5*\y) circle (0.4*\y);
\filldraw[color=red!60, fill=red!5, very thick](1.8*\y,1.5*\y) circle (0.4*\y);
\filldraw[black] (-1.5*\y,0.3*\y)  node[anchor=west] {\tiny{$a(p)$}};
\filldraw[black] (-1.5*\y,2.5*\y)  node[anchor=west] {\tiny{$b(p')$}};
\filldraw[black] (2.4*\y,0.3*\y)  node[anchor=west] {\tiny{$a(q)$}};
\filldraw[black] (2.4*\y,2.5*\y)  node[anchor=west] {\tiny{$b(q')$}};
\filldraw[black] (-0.1*\y,2.2*\y)  node[anchor=west] {\tiny{$e(k)$}};
\filldraw[black] (0.8*\y,2.2*\y)  node[anchor=west] {\tiny{$e(k)$}};
\filldraw[black] (-0.1*\y,0.5*\y)  node[anchor=west] {\tiny{$f(k')$}};
\filldraw[black] (0.7*\y,0.5*\y)  node[anchor=west] {\tiny{$f(k')$}};

\filldraw[black] (-0.2*\y,-0.3*\y)  node[anchor=west] {\scriptsize{$s$-channel cut}};


\draw[directed] (4.6*\y,0.5*\y) -- (5.2*\y,1.1*\y);
\draw[directed] (5.2*\y,1.8*\y) -- (4.6*\y,2.4*\y);
\draw[directed] (6.5*\y,2.1*\y) -- (6*\y,1.7*\y);
\draw[directed] (7.4*\y,1.7*\y) -- (6.9*\y,2.1*\y);
\draw[directed] (5.8*\y,1.2*\y) -- (6.2*\y,0.7*\y);
\draw[directed] (7*\y,0.7*\y) -- (7.4*\y,1.2*\y);
\draw[directed] (8.8*\y,0.6*\y) -- (8.2*\y,1.2*\y);
\draw[directed] (8.2*\y,1.8*\y) -- (8.8*\y,2.4*\y);

\filldraw[color=red!60, fill=red!5, very thick](5.6*\y,1.5*\y) circle (0.4*\y);
\filldraw[color=red!60, fill=red!5, very thick](7.8*\y,1.5*\y) circle (0.4*\y);
\filldraw[black] (4.5*\y,0.3*\y)  node[anchor=west] {\tiny{$a(p)$}};
\filldraw[black] (4.5*\y,2.5*\y)  node[anchor=west] {\tiny{$b(q')$}};
\filldraw[black] (8.4*\y,0.3*\y)  node[anchor=west] {\tiny{$b(p')$}};
\filldraw[black] (8.4*\y,2.5*\y)  node[anchor=west] {\tiny{$a(q)$}};
\filldraw[black] (5.9*\y,2.2*\y)  node[anchor=west] {\tiny{$e(k)$}};
\filldraw[black] (6.8*\y,2.2*\y)  node[anchor=west] {\tiny{$e(k)$}};
\filldraw[black] (5.9*\y,0.5*\y)  node[anchor=west] {\tiny{$f(k')$}};
\filldraw[black] (6.8*\y,0.5*\y)  node[anchor=west] {\tiny{$f(k')$}};

\filldraw[black] (5.8*\y,-0.3*\y)  node[anchor=west] {\scriptsize{$t$-channel cut}};


\draw[directed] (10.6*\y,0.5*\y) -- (11.2*\y,1.1*\y);
\draw[directed] (11.2*\y,1.8*\y) -- (10.6*\y,2.4*\y);
\draw[directed] (12.5*\y,2.1*\y) -- (12*\y,1.7*\y);
\draw[directed] (13.4*\y,1.7*\y) -- (12.9*\y,2.1*\y);
\draw[directed] (11.8*\y,1.2*\y) -- (12.2*\y,0.7*\y);
\draw[directed] (13*\y,0.7*\y) -- (13.4*\y,1.2*\y);
\draw[directed] (14.8*\y,0.6*\y) -- (14.2*\y,1.2*\y);
\draw[directed] (14.2*\y,1.8*\y) -- (14.8*\y,2.4*\y);

\filldraw[color=red!60, fill=red!5, very thick](11.6*\y,1.5*\y) circle (0.4*\y);
\filldraw[color=red!60, fill=red!5, very thick](13.8*\y,1.5*\y) circle (0.4*\y);
\filldraw[black] (10.5*\y,0.3*\y)  node[anchor=west] {\tiny{$a(p)$}};
\filldraw[black] (10.5*\y,2.5*\y)  node[anchor=west] {\tiny{$a(q)$}};
\filldraw[black] (14.4*\y,0.3*\y)  node[anchor=west] {\tiny{$b(p')$}};
\filldraw[black] (14.4*\y,2.5*\y)  node[anchor=west] {\tiny{$b(q')$}};
\filldraw[black] (11.9*\y,2.2*\y)  node[anchor=west] {\tiny{$e(k)$}};
\filldraw[black] (12.8*\y,2.2*\y)  node[anchor=west] {\tiny{$e(k)$}};
\filldraw[black] (11.9*\y,0.5*\y)  node[anchor=west] {\tiny{$f(k')$}};
\filldraw[black] (12.8*\y,0.5*\y)  node[anchor=west] {\tiny{$f(k')$}};

\filldraw[black] (11.8*\y,-0.3*\y)  node[anchor=west] {\scriptsize{$u$-channel cut}};

\end{tikzpicture}
\caption{Double-cuts in the $s$-, $t$- and $u$-channel (from left to right). The $s$-channel cut is not allowed for kinematical reasons.}
\label{All_possible_double_cuts_ab_into_cd}
\end{center}
\end{figure}

\subsection{Subtleties in three-to-three tree-level amplitudes}

To evaluate the quantity in~\eqref{4_point_one_loop_from_tree_single_cut_contribution} we need to understand the properties of its integrands.
It is important to note that the momenta of the external particles $p$, $p'$, $q$ and $q'$ cannot be set all on-shell otherwise there are cut diagrams containing retarded propagators
\begin{equation}
\frac{i}{r^2-m_e^2- i r_0 \ \epsilon}
\end{equation}
(there is one of them for each particle type entering the sum in~\eqref{4_point_one_loop_from_tree_single_cut_contribution}) which satisfy $r^2=m^2_e$ and are thus singular. 
\begin{figure}
\begin{center}
\begin{tikzpicture}
\tikzmath{\y=1.4;}

\draw[directed] (4.6*\y,0.5*\y) -- (5.2*\y,1.1*\y);
\draw[directed] (5.2*\y,1.8*\y) -- (4.6*\y,2.4*\y);
\draw[directed] (8.8*\y,0.6*\y) -- (8.2*\y,1.2*\y);
\draw[directed] (8.2*\y,1.8*\y) -- (8.8*\y,2.4*\y);
\filldraw[color=red!60, fill=red!5, very thick](5.6*\y,1.5*\y) circle (0.4*\y);
\draw[directed] (6.5*\y,2.1*\y) -- (6*\y,1.7*\y);
\draw[directed] (7.4*\y,1.7*\y) -- (6.9*\y,2.1*\y);
\draw[directed] (5.95*\y,1.3*\y) arc(240:300:1.5*\y);
\filldraw[color=red!60, fill=red!5, very thick](7.8*\y,1.5*\y) circle (0.4*\y);

\filldraw[black] (4.6*\y,0.3*\y)  node[anchor=west] {\scriptsize{$a(p)$}};
\filldraw[black] (4.6*\y,2.5*\y)  node[anchor=west] {\scriptsize{$a(q)$}};
\filldraw[black] (8.2*\y,0.3*\y)  node[anchor=west] {\scriptsize{$b(p')$}};
\filldraw[black] (8.2*\y,2.5*\y)  node[anchor=west] {\scriptsize{$b(q')$}};
\filldraw[black] (6.15*\y,2.2*\y)  node[anchor=west] {\scriptsize{$e(k)$}};
\filldraw[black] (6.75*\y,2.2*\y)  node[anchor=west] {\scriptsize{$e(k)$}};
\filldraw[black] (6.4*\y,0.9*\y)  node[anchor=west] {\scriptsize{$e(r)$}};

\filldraw[black] (9*\y,1.5*\y)  node[anchor=west] {\scriptsize{$+$}};


\draw[directed] (9.6*\y,0.5*\y) -- (10.2*\y,1.1*\y);
\draw[directed] (10.2*\y,1.8*\y) -- (9.6*\y,2.4*\y);
\draw[directed] (13.8*\y,0.6*\y) -- (13.2*\y,1.2*\y);
\draw[directed] (13.2*\y,1.8*\y) -- (13.8*\y,2.4*\y);
\draw[directed] (12.45*\y,1.7*\y) arc(60:120:1.5*\y);
\draw[directed] (10.8*\y,1.2*\y) -- (11.2*\y,0.7*\y);
\filldraw[color=red!60, fill=red!5, very thick](10.6*\y,1.5*\y) circle (0.4*\y);
\draw[directed] (12*\y,0.7*\y) -- (12.4*\y,1.2*\y);
\filldraw[color=red!60, fill=red!5, very thick](12.8*\y,1.5*\y) circle (0.4*\y);
\filldraw[black] (9.6*\y,0.3*\y)  node[anchor=west] {\scriptsize{$a(p)$}};
\filldraw[black] (9.6*\y,2.5*\y)  node[anchor=west] {\scriptsize{$a(q)$}};
\filldraw[black] (13.2*\y,0.3*\y)  node[anchor=west] {\scriptsize{$b(p')$}};
\filldraw[black] (13.2*\y,2.5*\y)  node[anchor=west] {\scriptsize{$b(q')$}};
\filldraw[black] (11.4*\y,2.1*\y)  node[anchor=west] {\scriptsize{$e(r')$}};
\filldraw[black] (11*\y,0.6*\y)  node[anchor=west] {\scriptsize{$e(k)$}};
\filldraw[black] (11.7*\y,0.6*\y)  node[anchor=west] {\scriptsize{$e(k)$}};

\end{tikzpicture}
\caption{Pair of single cuts which are ill-defined when the external particles are on-shell and satisfy $\{q=p, q'=p'\}$, for which $r=r'=k$ with $k^2=m_e^2$.}
\label{relevant_terms_in_single_cut_t_channel_diagram}
\end{center}
\end{figure}
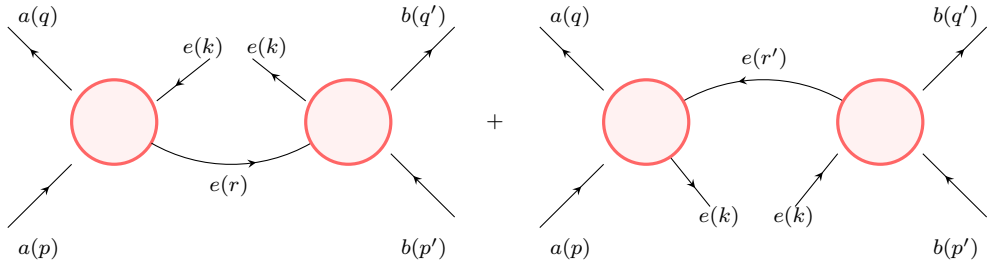
In Figure~\ref{relevant_terms_in_single_cut_t_channel_diagram} we show two combinations of single-cut diagrams which are ill-defined at the kinematical point $\{q=p, \, q'=p'\}$. At this point,
the momenta $r$ and $r'$ flowing in the diagrams are both on-shell and equal to $k$ (the momentum of the propagator that has been cut). To avoid this situation we will adopt the regularization scheme~\eqref{regularization_through_x_introducing_outgoing_mass_deformations}. With this scheme, the momenta $q$ and $q'$ are off-shell and there are no singular propagators. We can therefore set $i \epsilon=0$ in all denominators of propagators, define the amplitude for a small parameter $\mu$ in~\eqref{definition_of_my_regulator_x} and then take the limit $\mu \to 0$.

In all this procedure we should ask ourselves how the fact that the theory is tree-level integrable simplifies the computation.
Let us consider a tree-level amplitude $M_{abe}^{(0)}$ associated with the process
\begin{equation}
\label{3_to_3_tree_level_process}
a(p)+b(p')+e(k) \to a(q)+b(q')+e(\tilde{k}) \, ,
\end{equation}
in the case in which Property~\ref{Condition_tree_level_elasticity_introduction} is satisfied.
If we do not assume the external particles to be on-shell then $M_{abe}^{(0)}$ is a function of ten parameters: 
\begin{equation}
\label{ten_free_parameters_in_3_to_3_amplitudes}
p^2, \, p'^2, \, k^2, \, q^2, \, q'^2, \, \tilde{k}^2, \, \theta_p, \, \theta_{p'}, \, \theta_{k} \ \text{and} \ \theta_q \, .
\end{equation}
The rapidities $\theta_{\tilde{k}}$ and $\theta_{q'}$ are fixed by the overall energy-momentum conservation and have two branches of solutions. 
If we set $i \epsilon=0$ in the denominators of propagators, then for on-shell values of the external momenta it holds that $M_{abe}^{(0)}=0$. Indeed, if $i \epsilon=0$ there is no difference between a two-to-four process and a three-to-three process. Both $M_{2\to 4}^{(0)}$ and $M_{3\to 3}^{(0)}$ are zero everywhere (see, e.g., \cite{Dorey:2021hub}). If $i \epsilon\ne 0$ then $M_{2\to 4}^{(0)}$ and $M_{3\to 3}^{(0)}$ are zero but for a surface of the rapidity-space corresponding to the region where these amplitudes factorise.
In $M_{3\to 3}^{(0)}$ this surface intersects the on-shell space of physical external momenta while in $M_{2\to 4}^{(0)}$ all points of the surface correspond to unphysical momenta. Indeed certain external energies need to be negative in the region in which $M_{2\to 4}^{(0)}$ factorises. 
Due to this fact if we set $i \epsilon=0$ then the amplitude $M_{abe\to abe}^{(0)}$
vanishes on the surface (inside the space spanned by the ten parameters in~\eqref{ten_free_parameters_in_3_to_3_amplitudes}) defined by
\begin{equation} 
p^2=q^2=m_a^2 \quad,\quad p'^2=q'^2=m_b^2 \quad,\quad k^2=\tilde{k}^2=m_e^2 \,.
\end{equation}
Moving along this surface, parameterised by $\theta_p$, $\theta_{p'}$, $\theta_k$ and $\theta_q$, the amplitude $M_{abe\to abe}^{(0)}=0$ and is zero also in the limit $\theta_q \to \theta_p$, at which $\theta_{\tilde{k}}=\theta_{k}$ and $\theta_{q'}=\theta_{p'}$\footnote{Despite there are two branches of the kinematics, from now on we will focus on the elastic branch (the one corresponding to have $\theta_{q'}=\theta_{p'}$ and $\theta_{\tilde{k}}=\theta_k$ when $\theta_q=\theta_p$) since the integrand in~\eqref{4_point_one_loop_from_tree_single_cut_contribution} is evaluated on this branch.}.

Despite this fact, the on-shell point $\theta_q=\theta_p$ is very special and the ten-parameter function $M_{abe\to abe}^{(0)}$ is ill-defined and not continuous at this point. 
When $i \epsilon=0$, the limit of $M_{abe\to abe}^{(0)}$ to this point is zero on the surface where the external particles are all on-shell but is in general non-zero if the limit is taken along a different direction outside this surface. As previously explained we cannot assume all the external particles to be on-shell, otherwise certain propagators would be ill-defined. Therefore, the integrands in~\eqref{4_point_one_loop_from_tree_single_cut_contribution} are nontrivial functions that need to be computed. In the next sections, we will show how Property~\ref{Condition_tree_level_elasticity_introduction} can be used to simplify the computation and find closed expressions for these integrands. We will refer to the limit $\mu \to 0$ adopted to reach the elastic point 
\begin{equation} 
\label{eq:on_shell_point}
q=p \quad,\quad q'=p' \quad,\quad \tilde{k}=k \,,
\end{equation}
with all the particles on-shell, as the `off-shell' limit. This is because, even if all particles are on-shell at this point, we move along off-shell configurations to reach this point and $M_{abe \to abe}^{(0)} \ne 0$ in this limit. In contrast, if the particles are on-shell at each step of the limit used to reach the configuration~\eqref{eq:on_shell_point} then we will call this limit `on-shell limit' and $M_{abe \to abe}^{(0)} = 0$ in this limit.

\subsection{The off-shell limit of single-cuts}

In this section, we will provide a class of suitable directions to evaluate the integrands in~\eqref{4_point_one_loop_from_tree_single_cut_contribution} and we will show that the result is independent on the direction considered along this class. 

We set $i \epsilon=0$ in the denominators of propagators and we split the tree-level amplitude $\hat{M}_{abe}^{(0)}$ into two terms
\begin{equation}
\label{tree_level_amplitude_abe_to_abe}
\hat{M}_{abe}^{(0)}= \hat{V}_{abe}+ \hat{R}_{abe} \ne 0\, .
\end{equation}
As previously explained this amplitude is nonzero by the fact that the external momenta are off-shell along the limit.
The hats on the different terms in~\eqref{tree_level_amplitude_abe_to_abe} mean that we subtract contributions that survive in the limit of large $\theta_k$ otherwise the integrals in~\eqref{4_point_one_loop_from_tree_single_cut_contribution} may not converge. As before, these contributions are due to vertices with self-contracted propagators and can always be cancelled by introducing suitable counterterms in the Lagrangian.
The first term on the r.h.s. of~\eqref{tree_level_amplitude_abe_to_abe} contains collections of Feynman diagrams which, for on-shell values of the momenta, are singular on the branch of the kinematics at which $q=p$, $q'=p'$ and $\tilde{k}=k$. An example of these singular diagrams is reported in Figure~\ref{relevant_terms_in_single_cut_t_channel_diagram}, while the complete list of singular diagrams is represented in Figure~\ref{Factorization_contributions_in_tree_level_6_point_processes}; the diagrams are organised into three sets depending whether the singular propagator is of type $e$, $a$ or $b$. We write $\hat{V}_{abe\to abe}$ as the sum of these three collections of diagrams following the enumeration of Figure~\ref{Factorization_contributions_in_tree_level_6_point_processes}:
\begin{equation}
\label{split_of_V_3_to_3_into_V1_V2_V3}
\hat{V}_{abe}= \hat{V}^{(1)}_{abe}+\hat{V}^{(2)}_{abe}+\hat{V}^{(3)}_{abe} \, .
\end{equation}
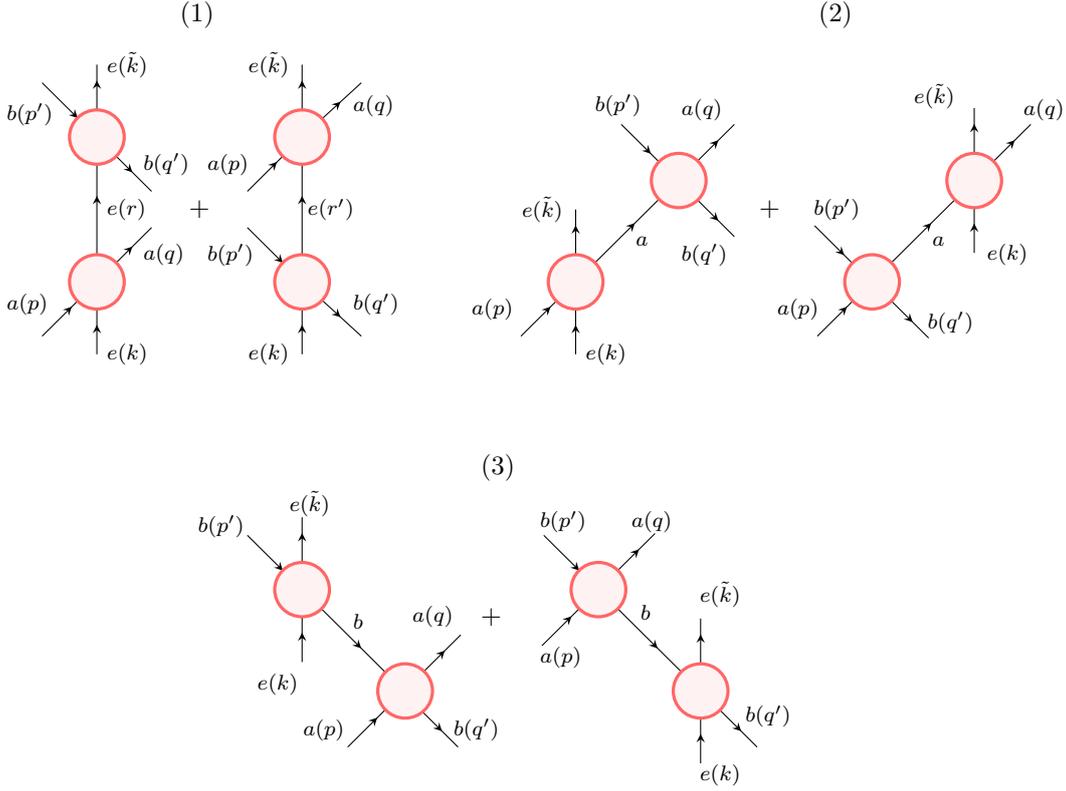
\begin{figure}
\begin{center}
\begin{tikzpicture}
\tikzmath{\y=0.6;}


\draw[directed] (0*\y,-1*\y) -- (0*\y,1*\y);
\draw[directed] (0*\y,2.2*\y) -- (0*\y,3.2*\y);
\draw[directed] (0*\y,-3.2*\y) -- (0*\y,-2.2*\y);
\draw[directed] (-1.2*\y,-2.8*\y) -- (-0.2*\y,-1.8*\y);
\draw[directed] (0*\y,-1.6*\y) -- (1.2*\y,-0.4*\y);
\draw[directed] (0*\y,1.6*\y) -- (1.2*\y,0.4*\y);
\draw[directed] (-1.2*\y,2.8*\y) -- (0*\y,1.6*\y);
\filldraw[color=red!60, fill=red!5, very thick](0*\y,1.6*\y) circle (0.6*\y);
\filldraw[color=red!60, fill=red!5, very thick](0*\y,-1.6*\y) circle (0.6*\y);

\filldraw[black] (-2.2*\y,-2.1*\y)  node[anchor=west] {\scriptsize	{$a(p)$}};
\filldraw[black] (0.8*\y,-1*\y)  node[anchor=west] {\scriptsize{$a(q)$}};
\filldraw[black] (-2.2*\y,2.1*\y)  node[anchor=west] {\scriptsize{$b(p')$}};
\filldraw[black] (0.8*\y,1*\y)  node[anchor=west] {\scriptsize{$b(q')$}};
\filldraw[black] (0*\y,-3.2*\y)  node[anchor=west] {\scriptsize{$e(k)$}};
\filldraw[black] (0*\y,0*\y)  node[anchor=west] {\scriptsize{$e(r)$}};
\filldraw[black] (0*\y,3.2*\y)  node[anchor=west] {\scriptsize{$e(\tilde{k})$}};

\draw[directed] (4.5*\y,-1*\y) -- (4.5*\y,1*\y);
\draw[directed] (4.5*\y,2.2*\y) -- (4.5*\y,3.2*\y);
\draw[directed] (4.5*\y,-3.2*\y) -- (4.5*\y,-2.2*\y);
\draw[directed] (3.3*\y,0.4*\y) -- (4.3*\y,1.4*\y);
\draw[directed] (4.5*\y,1.6*\y) -- (5.8*\y,2.8*\y);
\draw[directed] (4.5*\y,-1.6*\y) -- (5.8*\y,-2.8*\y);
\draw[directed] (3.3*\y,-0.4*\y) -- (4.5*\y,-1.6*\y);
\filldraw[color=red!60, fill=red!5, very thick](4.5*\y,1.6*\y) circle (0.6*\y);
\filldraw[color=red!60, fill=red!5, very thick](4.5*\y,-1.6*\y) circle (0.6*\y);

\filldraw[black] (5.4*\y,-2.1*\y)  node[anchor=west] {\scriptsize{$b(q')$}};
\filldraw[black] (2.2*\y,-1*\y)  node[anchor=west] {\scriptsize{$b(p')$}};
\filldraw[black] (5.4*\y,2.3*\y)  node[anchor=west] {\scriptsize{$a(q)$}};
\filldraw[black] (2.2*\y,1*\y)  node[anchor=west] {\scriptsize{$a(p)$}};
\filldraw[black] (3.1*\y,-3.2*\y)  node[anchor=west] {\scriptsize{$e(k)$}};
\filldraw[black] (4.4*\y,0*\y)  node[anchor=west] {\scriptsize{$e(r')$}};
\filldraw[black] (3.1*\y,3.2*\y)  node[anchor=west] {\scriptsize{$e(\tilde{k})$}};

\filldraw[black] (1.6*\y,4.3*\y)  node[anchor=west] {\small{$(1)$}};

\filldraw[black] (1.8*\y,0*\y)  node[anchor=west] {\small{$+$}};


\draw[directed] (10.5*\y,-1*\y) -- (10.5*\y,0*\y);
\draw[directed] (10.5*\y,-3.2*\y) -- (10.5*\y,-2.2*\y);
\draw[directed] (9.3*\y,-2.8*\y) -- (10.3*\y,-1.8*\y);
\draw[directed] (10.5*\y,-1.6*\y) -- (12.5*\y,0.4*\y);
\draw[directed] (12.98*\y,0.88*\y) -- (13.98*\y,1.88*\y);
\draw[directed] (11.5*\y,1.88*\y) -- (12.5*\y,0.88*\y);
\draw[directed] (12.98*\y,0.4*\y) -- (13.98*\y,-0.6*\y);
\filldraw[color=red!60, fill=red!5, very thick](12.76*\y,0.64*\y) circle (0.6*\y);
\filldraw[color=red!60, fill=red!5, very thick](10.5*\y,-1.6*\y) circle (0.6*\y);

\filldraw[black] (8*\y,-2.2*\y)  node[anchor=west] {\scriptsize{$a(p)$}};
\filldraw[black] (11.6*\y,-0.7*\y)  node[anchor=west] {\scriptsize{$a$}};
\filldraw[black] (12.6*\y,2.2*\y)  node[anchor=west] {\scriptsize{$a(q)$}};
\filldraw[black] (10.7*\y,2.3*\y)  node[anchor=west] {\scriptsize{$b(p')$}};
\filldraw[black] (12.6*\y,-1*\y)  node[anchor=west] {\scriptsize{$b(q')$}};
\filldraw[black] (10.5*\y,-3.2*\y)  node[anchor=west] {\scriptsize{$e(k)$}};
\filldraw[black] (9.1*\y,0*\y)  node[anchor=west] {\scriptsize{$e(\tilde{k})$}};

\draw[directed] (19.24*\y,1.24*\y) -- (19.24*\y,2.24*\y);
\draw[directed] (19.24*\y,-0.96*\y) -- (19.24*\y,0.04*\y);
\draw[directed] (15.8*\y,-2.8*\y) -- (16.8*\y,-1.8*\y);
\draw[directed] (17*\y,-1.6*\y) -- (19*\y,0.4*\y);
\draw[directed] (19.48*\y,0.88*\y) -- (20.48*\y,1.88*\y);
\draw[directed] (15.74*\y,1.88*\y-2.24*\y) -- (16.74*\y,0.88*\y-2.24*\y);
\draw[directed] (17.24*\y,0.4*\y-2.24*\y) -- (18.24*\y,-0.6*\y-2.24*\y);

\filldraw[color=red!60, fill=red!5, very thick](19.24*\y,0.64*\y) circle (0.6*\y);
\filldraw[color=red!60, fill=red!5, very thick](17*\y,-1.6*\y) circle (0.6*\y);

\filldraw[black] (14.7*\y,-2.2*\y)  node[anchor=west] {\scriptsize{$a(p)$}};
\filldraw[black] (18.1*\y,-0.7*\y)  node[anchor=west] {\scriptsize{$a$}};
\filldraw[black] (20.1*\y,2.2*\y)  node[anchor=west] {\scriptsize{$a(q)$}};
\filldraw[black] (15.5*\y,0*\y)  node[anchor=west] {\scriptsize{$b(p')$}};
\filldraw[black] (18*\y,-2.5*\y)  node[anchor=west] {\scriptsize{$b(q')$}};
\filldraw[black] (19.3*\y,-1*\y)  node[anchor=west] {\scriptsize{$e(k)$}};
\filldraw[black] (17.7*\y,2.5*\y)  node[anchor=west] {\scriptsize{$e(\tilde{k})$}};

\filldraw[black] (15.6*\y,4.3*\y)  node[anchor=west] {\small{$(2)$}};
\filldraw[black] (14.3*\y,0*\y)  node[anchor=west] {\small{$+$}};


\draw[directed] (4.5*\y,-10*\y) -- (4.5*\y,-9*\y);
\draw[directed] (4.5*\y,-7.8*\y) -- (4.5*\y,-6.8*\y);
\draw[directed] (3.3*\y,-7.2*\y) -- (4.5*\y,-8.4*\y);
\draw[directed] (4.5*\y,-8.4*\y) -- (6.5*\y,-10.4*\y);
\draw[directed] (6.92*\y,-10.88*\y) -- (7.92*\y,-11.88*\y);
\draw[directed] (5.5*\y,-11.88*\y) -- (6.5*\y,-10.88*\y);
\draw[directed] (6.98*\y,-10.4*\y) -- (7.98*\y,-9.4*\y);
\filldraw[color=red!60, fill=red!5, very thick](4.5*\y,-8.4*\y) circle (0.6*\y);
\filldraw[color=red!60, fill=red!5, very thick](6.76*\y,-10.64*\y) circle (0.6*\y);

\filldraw[black] (4.3*\y,-11.5*\y)  node[anchor=west] {\scriptsize{$a(p)$}};
\filldraw[black] (6.7*\y,-9*\y)  node[anchor=west] {\scriptsize{$a(q)$}};
\filldraw[black] (2*\y,-7*\y)  node[anchor=west] {\scriptsize{$b(p')$}};
\filldraw[black] (5.4*\y,-9.1*\y)  node[anchor=west] {\scriptsize{$b$}};
\filldraw[black] (7.6*\y,-11.5*\y)  node[anchor=west] {\scriptsize{$b(q')$}};
\filldraw[black] (3.3*\y,-10.5*\y)  node[anchor=west] {\scriptsize{$e(k)$}};
\filldraw[black] (4*\y,-6.5*\y)  node[anchor=west] {\scriptsize{$e(\tilde{k})$}};

\draw[directed] (13.24*\y,-12.24*\y) -- (13.24*\y,-11.24*\y);
\draw[directed] (13.24*\y,-10.04*\y) -- (13.24*\y,-9.04*\y);
\draw[directed] (9.8*\y,-7.2*\y) -- (11*\y,-8.4*\y);
\draw[directed] (11*\y,-8.4*\y) -- (13*\y,-10.4*\y);
\draw[directed] (13.48*\y,-10.88*\y) -- (14.48*\y,-11.88*\y);
\draw[directed] (9.76*\y,-9.64*\y) -- (10.76*\y,-8.64*\y);
\draw[directed] (11.24*\y,-8.16*\y) -- (12.24*\y,-7.16*\y);
\filldraw[color=red!60, fill=red!5, very thick](11*\y,-8.4*\y) circle (0.6*\y);
\filldraw[color=red!60, fill=red!5, very thick](13.24*\y,-10.64*\y) circle (0.6*\y);

\filldraw[black] (9.5*\y,-9.9*\y)  node[anchor=west] {\scriptsize{$a(p)$}};
\filldraw[black] (11.5*\y,-6.9*\y)  node[anchor=west] {\scriptsize{$a(q)$}};
\filldraw[black] (9.5*\y,-6.9*\y)  node[anchor=west] {\scriptsize{$b(p')$}};
\filldraw[black] (11.7*\y,-8.9*\y)  node[anchor=west] {\scriptsize{$b$}};
\filldraw[black] (14*\y,-11.2*\y)  node[anchor=west] {\scriptsize{$b(q')$}};

\filldraw[black] (13*\y,-12.5*\y)  node[anchor=west] {\scriptsize{$e(k)$}};
\filldraw[black] (13*\y,-8.5*\y)  node[anchor=west] {\scriptsize{$e(\tilde{k})$}};

\filldraw[black] (8.2*\y,-5.7*\y)  node[anchor=west] {\small{$(3)$}};
\filldraw[black] (8.2*\y,-9*\y)  node[anchor=west] {\small{$+$}};

\end{tikzpicture}
\caption{Collections of singular Feynman diagrams on the on-shell limit that contribute to $\hat{V}_{abe}$.}
\label{Factorization_contributions_in_tree_level_6_point_processes}
\end{center}
\end{figure}
The second contribution on the r.h.s. of~\eqref{tree_level_amplitude_abe_to_abe} contains instead all the remaining Feynman diagrams, which for general on-shell values of $p$, $p'$ and $k$ are finite and well-defined at the point $\{q=p, q'=p', \tilde{k}=k\}$.

Defining
\begin{equation}
\label{definition_of_Vab_to_ab_term_i_123}
\hat{V}^{(i)}_{ab} \equiv \frac{1}{8\pi} \sum_{e=1}^r \int_{-\infty}^{+\infty} d\theta_k \hat{V}^{(i)}_{abe}(p,p',k) \hspace{5mm} \text{with} \ i=1,2,3 \, ,
\end{equation}
and
\begin{equation}
\hat{R}_{ab} \equiv \frac{1}{8\pi} \sum_{e=1}^r \int_{-\infty}^{+\infty} d\theta_k \hat{R}_{abe}(p,p',k) 
\end{equation}
then~\eqref{4_point_one_loop_from_tree_single_cut_contribution} can be written as
\begin{equation}
\label{sisngle_cut_contribution_written_as_V_plus_R}
M_{ab}^{(1\text{-cut})} =\hat{V}_{ab}+\hat{R}_{ab} \, ,
\end{equation}
where
\begin{equation}
\label{definition_of_Vab_to_ab_term}
\hat{V}_{ab}=\hat{V}^{(1)}_{ab}+\hat{V}^{(2)}_{ab}+\hat{V}^{(3)}_{ab} \, .
\end{equation}
To avoid divergencies in the integrands on the r.h.s. of~\eqref{definition_of_Vab_to_ab_term_i_123} we set $p$, $p'$, $q$ and $q'$ off-shell and satisfying $p+p'=q+q'$. The momenta $k$ and $\tilde{k}$ come from the same cut propagator and for this reason they satisfy $\tilde{k}=k$ and need to be on-shell. After having evaluated~\eqref{definition_of_Vab_to_ab_term} for $p$, $p'$, $q$ and $q'$ off-shell we take the limit at which the momenta are on-shell and $q=p$ and $q'=p'$.

We start by evaluating $\hat{V}^{(2)}_{ab\to ab}$ and $\hat{V}^{(3)}_{ab\to ab}$. Focusing on the class of diagrams $(2)$ in Figure~\ref{Factorization_contributions_in_tree_level_6_point_processes}, with $\tilde{k}=k$, we can write
\begin{equation}
\label{expression_for_Vabe_to_abe_2}
\begin{split}
\hat{V}^{(2)}_{abe} &= M^{(0)}_{ae \to ae} (p, k , p , k) \frac{i}{p^2 - m^2_a}  M^{(0)}_{ab \to ab} (p, p' , q , q')\\
&+ M^{(0)}_{ab \to ab} (p, p' , q , q') \frac{i}{q^2 - m^2_a} M^{(0)}_{ae \to ae} (q, k , q , k) \, . 
\end{split}
\end{equation}
Plugging~\eqref{expression_for_Vabe_to_abe_2} into~\eqref{definition_of_Vab_to_ab_term_i_123} and using~\eqref{definition_necessary_for_mass_renormalization_first_time_ga_appears} we obtain
\begin{equation}
\begin{split}
\hat{V}^{(2)}_{ab} &= \Sigma_{aa} (p^2) \frac{1}{p^2 - m^2_a}  M^{(0)}_{ab \to ab} (p, p' , q , q')\\
&+ \Sigma_{aa} (q^2) \frac{1}{q^2 - m^2_a} M^{(0)}_{ab \to ab} (p, p' , q , q') \, .
\end{split}
\end{equation}
An analogous computation can be performed for $\hat{V}^{(3)}_{ab} $.
If we set  $\{q=p, q'=p'\}$ and we expand $\Sigma_{aa}$ and $\Sigma_{bb}$ around $p^2=m_a^2$ and $p'^2=m_b^2$ respectively, taking into account that $\Sigma^{(0)}_{jj}= - \delta m_j^2$, then we end up with
\begin{equation}
\label{V2ab_V3ab_to_ab_final_result}
\begin{split}
\hat{V}^{(2)}_{ab}+\hat{V}^{(3)}_{ab} =& - \frac{2  \delta m_a^2}{p^2 - m^2_a}  M^{(0)}_{ab} (p, p')-\frac{2  \delta m_b^2}{p'^2 - m^2_b}  M^{(0)}_{ab} (p, p')\\
&+ 2 \Sigma^{(1)}_{aa} M^{(0)}_{ab} (p, p') + 2 \Sigma^{(1)}_{bb} M^{(0)}_{ab} (p, p')  \, .
\end{split}
\end{equation}
The expression on the r.h.s. of~\eqref{V2ab_V3ab_to_ab_final_result} is actually ill-defined: the first row on the r.h.s. of~\eqref{V2ab_V3ab_to_ab_final_result} is indeed divergent when the external momenta are on-shell. However, it turns out that the ill-defined terms in~\eqref{V2ab_V3ab_to_ab_final_result} cancel with the contributions in~\eqref{total_counterterms_contribution_production_process} introduced by the renormalization of the propagators.

To evaluate~\eqref{definition_of_Vab_to_ab_term} we still need to compute $\hat{V}^{(1)}_{ab}$. Labelling the momenta as in the diagram class $(1)$ in Figure~\ref{Factorization_contributions_in_tree_level_6_point_processes}, with the additional constraint $\tilde{k}=k$, we can write
\begin{equation}
\label{single_cut_relevant_contribution_u_channel_not_yet_expanded}
\begin{split}
\hat{V}^{(1)}_{ab}&= \sum_{e=1}^r \Bigl(\frac{1}{8\pi}  \int^{+\infty}_{-\infty} d\theta_k \ M^{(0)}_{ae \to a e}(p,k,q,r) \ \Pi^{(R)}_e(r) \ M^{(0)}_{b e \to b e }(p', r, q', k)\\
&+\frac{1}{8\pi}  \int^{+\infty}_{-\infty} d\theta_k \ M^{(0)}_{ae \to a e}(p, r', q, k) \ \Pi^{(R)}_e(r') \  M^{(0)}_{be \to be }(p', k, q', r') \Bigl).
\end{split}
\end{equation}
To compute this term we choose once again the regulator~\eqref{definition_of_my_regulator_x} and we require~\eqref{regularization_through_x_introducing_outgoing_mass_deformations} with $p$ and $p'$ on-shell. With this choice, then it holds that 
\begin{subequations}
\label{r_and_rprime_written_in_terms_of_k_and_x}
\begin{equation}
\label{r_written_in_terms_of_k_and_x}
r=k-x \, ,
\end{equation}
\begin{equation}
\label{rprime_written_in_terms_of_k_and_x}
r'=k+x \, ,
\end{equation}
\end{subequations}
and  the retarded propagators can be expanded around the point $\mu=0$ as
\begin{equation}
\label{expansion_of_propagators_in_u_channel}
    \begin{split}
        &\Pi_e^{(R)}(r)= \frac{i}{-2 k \cdot x +\mu^2}= -\frac{i}{2 k \cdot x}-\frac{i \mu^2}{4 (k \cdot x)^2} + O(\mu) \, , \\
        &\Pi_e^{(R)}(r')= \frac{i}{+2 k \cdot x +\mu^2}= +\frac{i}{2 k \cdot x}-\frac{i \mu^2}{4 (k \cdot x)^2}+ O(\mu) \, .
    \end{split}
\end{equation}
The $i \epsilon$ factors in the denominators of propagators can be removed since, as already mentioned in Section~\ref{sec:double-cut-elastic}, for $\mu>0$ and $\theta_x \in \mathbb{R}$ the momenta $r$ and $r'$ are never on-shell.
Plugging~\eqref{expansion_of_propagators_in_u_channel} into~\eqref{single_cut_relevant_contribution_u_channel_not_yet_expanded} and expanding the tree-level amplitudes appearing in the integrand of~\eqref{single_cut_relevant_contribution_u_channel_not_yet_expanded} around $\mu=0$ we obtain the following result
\begin{equation}
\label{V1ab_to_ab_final_result}
\begin{split}
\hat{V}^{(1)}_{ab}&=\frac{i}{8\pi} \sum_{e=1}^r \int^{+\infty}_{-\infty} d\theta_k \  \frac{\partial}{\partial k^2} \Bigl( M^{(0)}_{ae}(p, k) M^{(0)}_{be}(p', k) \Bigl)\Bigl|_{k^2=m^2_e}\\
&- \frac{i}{8\pi} \sum_{e=1}^r \frac{1}{m_e^2}    M^{(0)}_{ae}(p, \infty) M^{(0)}_{be}(p', \infty) + O(\mu).
\end{split}
\end{equation}
A detailed derivation of~\eqref{V1ab_to_ab_final_result} is reported in Appendix~\ref{appendix_on_off_shell_limit_of_tree_level_amplitudes}. Remarkably we note that the expression in~\eqref{V1ab_to_ab_final_result} is finite in the limit $\mu \to 0$ and does not depend on the parameter $\theta_x$ appearing in~\eqref{definition_of_my_regulator_x}. 
Summing~\eqref{V1ab_to_ab_final_result} with~\eqref{V2ab_V3ab_to_ab_final_result}, at the leading order in $\mu$ we obtain
\begin{equation}
\label{final_result_single_cut_ill_defined_diagrams}
\begin{split}
\hat{V}_{ab}&=\frac{i}{8\pi} \sum_{e=1}^r \int^{+\infty}_{-\infty} d\theta_k \  \frac{\partial}{\partial k^2} \Bigl( M^{(0)}_{ae}(p, k) M^{(0)}_{be}(p', k) \Bigl)\Bigl|_{k^2=m^2_e}\\
&- \frac{i}{8\pi} \sum_{e=1}^r \frac{1}{m_e^2}    M^{(0)}_{ae}(p, \infty) M^{(0)}_{be}(p', \infty)+ 2 \Sigma^{(1)}_{aa} M^{(0)}_{ab} (p, p') + 2 \Sigma^{(1)}_{bb} M^{(0)}_{ab} (p, p') \\
&- \frac{2  \delta m_a^2}{p^2 - m^2_a}  M^{(0)}_{ab} (p, p')-\frac{2  \delta m_b^2}{p'^2 - m^2_b}  M^{(0)}_{ab} (p, p')\, .
\end{split}
\end{equation}
Apart from the ill-defined terms in the last row of~\eqref{final_result_single_cut_ill_defined_diagrams}, which cancel the ill-defined terms in the first-row of~\eqref{total_counterterms_contribution_production_process} introduced by the renormalization of propagators, $\hat{V}_{ab\to ab}$ does not depend on the direction adopted (within the class of directions spanned by the parameter $\theta_x$ in~\eqref{definition_of_my_regulator_x}) to reach the elastic on-shell point
$$
p^2=q^2=m_a^2 \,, \ p'^2=q'^2=m_b^2 \,, \ k^2=\tilde{k}^2=m_e^2 \,, \ \theta_q=\theta_p \, , \ \theta_{q'}=\theta_{p'} \, , \ \theta_{\tilde{k}}=\theta_{k}.
$$

In the next section, we will compute $\hat{R}_{ab}$ in~\eqref{sisngle_cut_contribution_written_as_V_plus_R}, which is the missing ingredient to obtain the single-cut contribution to the one-loop amplitude. We recall that $\hat{R}_{ab}$ is generated by all the remaining Feynman diagrams, which in the limit $\mu \to 0$ are individually finite and well-defined.
To compute this term, we will exploit the fact that the theory satisfies Property~\ref{Condition_tree_level_elasticity_introduction}.

\subsection{On-shell limit of single-cuts in elastic amplitudes}
\label{section_on_properties_of_tree_level_inelastic_amplitudes}

If instead of computing $\hat{M}_{abe}^{(0)}$ following the off-shell limit previously discussed we compute it along a direction in which all the external particles are on-shell then by Property~\ref{Condition_tree_level_elasticity_introduction} it needs to hold that
\begin{equation}
\label{on_shell_limit_of_a_3_to_3_non_elastic_amplitude}
    \hat{M}_{abe}^{(0,\text{on})}= \hat{V}_{abe}^{(\text{on})}+ \hat{R}_{abe}=0.
\end{equation}
We remind that this is the case only if we set $i \epsilon=0$ in the denominators of propagators, as we will do.
The superscript `on' means that we are computing this tree amplitude keeping all the particles on-shell at each step in the limit. To perform such an on-shell limit we consider the momenta of the incoming particles $p$, $p'$ and $k$ fixed and on-shell in the process~\eqref{3_to_3_tree_level_process}; then we move the momentum $\tilde{k}$ (which is also on-shell) sending $\theta_{\tilde{k}} \to \theta_k$. The term $\hat{R}_{abe}$ is the same in~\eqref{tree_level_amplitude_abe_to_abe} and~\eqref{on_shell_limit_of_a_3_to_3_non_elastic_amplitude} since it is reproduced by summing over Feynman diagrams which are finite at the elastic point~\eqref{eq:on_shell_point}
and does not depend on the limit adopted to reach this on-shell point.
Due to this fact, defining
\begin{equation}
\label{definition_of_Vab_to_ab_term_i_123_onshell_limit}
\hat{V}^{(\text{on})}_{ab} \equiv \frac{1}{8\pi} \sum_{e=1}^r \int_{-\infty}^{+\infty} d\theta_k \hat{V}^{(\text{on})}_{abe}(p,p',k) \, ,
\end{equation}
then it needs to hold due to \eqref{on_shell_limit_of_a_3_to_3_non_elastic_amplitude} that
\begin{equation}
\hat{R}_{ab} = -\hat{V}^{(\text{on})}_{ab} \,.
\end{equation}
In Appendix~\ref{app:on_shell_limit} we report a detailed computation of $\hat{V}^{(\text{on})}_{ab}$, through which we prove that
\begin{equation}
\label{final_result_single_cut_well_defined_diagrams}
\begin{split}
\hat{R}_{ab} =& -\bigl( \Sigma^{(1)}_{aa}+ \Sigma^{(1)}_{bb} \bigl) M^{(0)}_{ab}(p, p')+\delta m_a^2 \frac{\partial}{\partial p^2} M^{(0)}_{ab}(p, p' )\Bigl|_{p^2=m_a^2}+\delta m_b^2 \frac{\partial}{\partial p'^2} M^{(0)}_{ab}(p, p' )\Bigl|_{p'^2=m_b^2}\\
&- \frac{i}{8 \pi} \sum_{e=1}^r \int_{-\infty}^{+\infty} d \theta_k \frac{\partial}{\partial k^2} \Bigl( M^{(0)}_{ae} (p, k) M^{(0)}_{b e} (p', k) \Bigl)\Bigl|_{k^2=m^2_e}\\
&-\frac{i \ai}{8 \pi} \theta_{p p'} \coth{\theta_{p p'}} M^{(0)}_{ab} (p, p') \, \sum_{e=1}^r m_e^2-\frac{1}{2} \Bigl( \frac{\delta m_a^2}{m_a^2} +\frac{\delta m_b^2}{m_b^2} \Bigl) M^{(0)}_{ab}(p, p' )\\
&+ \frac{i \ai}{8 \pi} \theta_{p p'}\frac{\partial}{\partial \theta_p} M^{(0)}_{ab} (p, p') \sum_{e=1}^r \, m_e^2+ \frac{i (\ai)^2}{8 \pi} m_a^2 m_b^2  \sum_{e=1}^r m_e^2\\
&+\frac{i}{\pi} m_a m_b \sinh{\theta_{pp'}}  \sum_{e=1}^r \frac{\partial}{\partial \theta_{p p'}} \ \text{p.v.} \int^{+\infty}_{-\infty} d\theta_k \,  S^{(0)}_{ae} (p, k ) S^{(0)}_{b e } (p', k )  \,.
\end{split}
\end{equation}
In the expression above the constant $\ai$ corresponds to the common value of the amplitudes at infinity properly normalised by the masses of the scattered particles, as defined in~\eqref{values_of_tree_level_amplitudes_at_infty}.
We stress that $\hat{V}_{abe \to abe}^{(\text{on})}$ and $\hat{V}_{abe \to abe}$ are both obtained by summing over the diagrams in Figure~\ref{Factorization_contributions_in_tree_level_6_point_processes}. However, they are evaluated at two different kinematical configurations and their limits to the point~\eqref{eq:on_shell_point} are different.
Finally if we now sum the expressions in~\eqref{total_counterterms_contribution_production_process}, \eqref{eq_double_cut_contribution_elastic}, \eqref{final_result_single_cut_ill_defined_diagrams} and~\eqref{final_result_single_cut_well_defined_diagrams} we obtain formula~\eqref{eq:final_result_amplitude} for the renormalized one-loop amplitude.

\section{Integrable counterterms and one-loop S-matrices}
\label{sec:integr_counterterms}

In~\cite{Polvara:2023vnx} one-loop production amplitudes were computed starting from tree-level integrable Lagrangians of type~\eqref{eq0_1} and it was shown that the contribution 
to the amplitudes arising from the black and blue terms in~\eqref{eq:renormalised_Lagrangian} (the analogue of $M^{(1\text{-loop})}_{ab}+ M^{(\text{ct.I})}_{a b}$)  was nonzero in general. These terms correspond to the part of the Lagrangian containing the renormalized quantities, and the masses and field counterterms, respectively.
To avoid production amplitudes at one loop it was necessary to introduce the additional counterterms $\delta C^{(n)}_{a_1 \cdots a_n}$, depicted in red in~\eqref{eq:renormalised_Lagrangian}.
These counterterms will affect the elastic amplitudes as well through the contribution $M^{(\text{ct.II})}_{a b}$, defined in~\eqref{eq:countertermsII}. In this section, we will provide a universal expression for this contribution in theories having mass ratios which do not renormalise at one loop and we will use it to write a universal expression for one-loop renormalized S-matrices in terms of tree-level S-matrices. Then we will check our expression on well-known examples of integrable theories of this type. 

Even if it is beyond the purpose of this paper, we recall that in certain integrable theories these counterterms arise by integrating out certain auxiliary fields in apparently different models, as is the case of the generalised sine-Gordon theories considered in~\cite{Hoare:2010fb}. It should be interesting to check if a similar mechanism to the one used in~\cite{Hoare:2010fb} applies more in general also for the class of theories considered here.

\subsection{Theories with all masses scaling in the same manner}
\label{sec:mass-scale-same}

As aforementioned in Section \ref{sec:main_results_conventions}, the counterterms $\delta C^{(n)}_{a_1 \cdots a_n}$ are fixed by imposing the renormalization condition (3), i.e., absence of inelastic and production processes at one loop.
While fixing these counterterms for general Lagrangians of type~\eqref{eq0_1} remains an open problem, which is possible only on a case-by-case study of different theories, in~\cite{Polvara:2023vnx} it was explained how to set these counterterms universally for Lagrangian having mass ratios that do not renormalize at one-loop. In particular, based on the results of~\cite{Polvara:2023vnx} we observe the following property.
\begin{mytheorem}
\label{Property_masses_renormalization}
Given a tree-level purely elastic theory of type~\eqref{eq0_1} with mass counterterms satisfying
\begin{equation}
\label{eq:mass_renormalz_condition}
 \frac{\delta m_a^2}{m_a^2}= \dd \quad \forall \ a=1, \dots, r \ \textrm{and} \ \dd \in \mathbb{R} \,,
\end{equation}
then if we set
\begin{equation}
\label{eq:couplings-special-case3}
\frac{\delta C_{a b c}^{(3)}}{C_{a b c}^{(3)}}=\gamma_3 \quad \textrm{with} \ \gamma_3 \in \mathbb{R}
\end{equation}
and
\begin{equation}
\label{eq:couplings-special-casen}
\frac{\delta C_{a_1 \dots a_n}^{(n)}}{C_{a_1 \dots a_n}^{(n)}} = (n-2) \gamma_3 + (3- n)\gamma \ \ \ , \ \ \ n\geq 4\,,
\end{equation}
all one-loop production amplitudes vanish.
\end{mytheorem} 

Here we briefly sketch how this property can be proved, while we remand the reader to~\cite{Polvara:2023vnx} for a more complete discussion. Let us make the following ansatz for the counterterms
\begin{equation}
\label{eq:counterterm_ansatz}
\frac{\delta C^{(n)}_{a_1 \cdots a_n}}{C^{(n)}_{a_1 \cdots a_n}}= \gamma_n \,,
\end{equation}
where $\{\gamma_n\}^{\infty}_{n=3}$ are parameters independent of the particle types, and consider the inelastic scattering process in~\eqref{eq:inel_scattering}, with $\{a, b\} \ne \{c, d\}$.
In this case, as shown in \cite{Polvara:2023vnx}, expression \eqref{eq:final_result_amplitude} reduces to
\begin{equation}
\label{eq:inel_amp_all_massess_scale_well}
\begin{split}
M^{(1\text{-loop})}_{ab \to cd}+ M^{(\text{ct.I})}_{a b \to cd}&= \sum_{j \in \{\text{prop, ext} \}} \delta m_j^2 \frac{\partial}{\partial m_j^2} M^{(0)}_{ab \to cd}.
\end{split}
\end{equation}
Note that even though the inelastic tree-level amplitude $M^{(0)}_{ab \to cd}$ vanishes on-shell, its derivatives in \eqref{eq:inel_amp_all_massess_scale_well} are evaluated off-shell and the masses are set to be on-shell only afterwards. For this reason, the expression in~\eqref{eq:inel_amp_all_massess_scale_well} is non-vanishing in general. We then need to add the contribution
\begin{equation}
\label{eq:inel_4_point_ctII}
    \begin{split}
        M_{ab \to cd}^{(\text{ct. II})}=&-i \sum_{i \in s} \frac{1}{s - m^2_i} \bigl(\delta C^{(3)}_{ab\bar{i}}  C^{(3)}_{i \bar{c} \bar{d}} + C^{(3)}_{ab\bar{i}}   \delta C^{(3)}_{i \bar{c} \bar{d}} \bigl) 
-i \sum_{j \in t} \frac{1}{t - m^2_j} \bigl(\delta C^{(3)}_{a \bar{d}\bar{j}}  C^{(3)}_{j b \bar{c}} + C^{(3)}_{a \bar{d} \bar{j}}   \delta C^{(3)}_{j b \bar{c}} \bigl)\\
&-i \sum_{l \in u} \frac{1}{u - m^2_l} \bigl(\delta C^{(3)}_{a \bar{c}\bar{l}}  C^{(3)}_{l b \bar{d}} + C^{(3)}_{a \bar{c} \bar{l}}   \delta C^{(3)}_{l b \bar{d}} \bigl)- i \delta C_{ab \bar{c} \bar{d}}^{(4)}\ ,\\
    \end{split}
\end{equation}
due to the coupling counterterms, to expression
\eqref{eq:inel_amp_all_massess_scale_well} and get the full one-loop result that we require to vanish. With this requirement, we fix the cubic and quartic counterterms. 
Let us write $M^{(0)}_{ab \to cd}$ as a function of the Mandelstam variables as in~\eqref{eq:tree-level-scattering} and consider a Feynman diagram with a particle of type $i$ propagating in the $s$-channel. The derivatives with respect to the squares of the masses will act on this propagator as follows 
\begin{equation}
\begin{split}
&\left(m_a^2 \frac{\partial}{\partial m_a^2}+m_b^2 \frac{\partial}{\partial m_b^2}+m_i^2 \frac{\partial}{\partial m_i^2}\right)\frac{1}{s - m_i^2}=\\
&\frac{1}{(s - m_i^2)^2} \left[ m_i^2 - m_a^2 \frac{\partial s}{\partial m^2_a}  - m_b^2 \frac{\partial s}{\partial m^2_b}  \right] = -\frac{1}{s - m_i^2} \,,
\end{split}
\end{equation}
where the definition of $s$ in \eqref{eq:Mandelstam} has been used.
Similar relations apply to diagrams with particles propagating in the $t$- and $u$- channels; diagrams with $4$-point couplings are instead annihilated by the derivatives since they do not contain any propagators.

Then summing~\eqref{eq:inel_amp_all_massess_scale_well} and~\eqref{eq:inel_4_point_ctII}, and using the ansatz in~\eqref{eq:counterterm_ansatz}, we obtain
\begin{equation}
\label{eq:inelastic_one_loop_amplitude}
\begin{split}
&M^{(1)}_{ab \to cd}=M^{(1\text{-loop})}_{ab \to cd}+ M^{(\text{ct.I})}_{a b \to cd}+M_{ab \to cd}^{(\text{ct. II})}\\
&=- i (2 \dd_3-\dd) \Bigl(\sum_{i \in s} \frac{ C^{(3)}_{ab\bar{i}}  C^{(3)}_{i \bar{c} \bar{d}}}{s-m^2_i} +\sum_{j \in t} \frac{ C^{(3)}_{a\bar{c}\bar{j}}  C^{(3)}_{j b \bar{d}}}{t-m^2_j}+ \sum_{k \in u} \frac{ C^{(3)}_{a\bar{d}\bar{k}}  C^{(3)}_{k b \bar{c}}}{u-m^2_k}   \Bigl)- i \dd_4 C_{ab \bar{c} \bar{d}}^{(4)} \,.
\end{split}
\end{equation}
The requirement for this amplitude to vanish is
\begin{equation}
\label{eq:constraint_coupling_renormalization}
2 \dd_3-\dd= \dd_4 \,,
\end{equation}
which follows from the vanishing of the tree-level inelastic amplitude by Property~\ref{Condition_tree_level_elasticity_introduction}. The space of solutions to~\eqref{eq:constraint_coupling_renormalization} has one degree of freedom spanned by $\dd_3$; note indeed that the factor $\gamma$ appearing in the mass corrections is fixed by summing bubble diagrams and is therefore a function of the classical  Lagrangian data (as shown in~\eqref{definition_of_mass_renormalization_and_tab_to_diagonalise_the_mass_matrix}). 
We see that $\gamma_4$ is exactly of the type described in~\eqref{eq:couplings-special-casen}.
Similarly, it is possible to prove that all on-shell production amplitudes vanish at one loop if relation~\eqref{eq:couplings-special-casen} is satisfied on all the remaining higher-order couplings.

These conditions can be derived following the same analysis of Section 3.6 of~\cite{Polvara:2023vnx} and are obtained thanks to the fact that all couplings are connected by recursion relations of type~\eqref{eq0_6}. It is important to notice that the coupling counterterms depend on a single degree of freedom $\gamma_3$. 
This is not surprising if we think that one-loop effects do not change the mass ratios.
As discussed in~\cite{Dorey:2021hub} the knowledge of what $3$-point couplings are non-zero and of the masses should provide enough information to define unambiguously the model under discussion in the space of tree-level integrable theories with Lagrangians of type~\eqref{eq0_1}. The only remaining degree of freedom should be indeed the interaction scale common to all $3$-point couplings. In this case, the parameter $\gamma_3$ corresponds to our freedom to change this scale without modifying the theory. 

Let us now show how the one-loop elastic amplitudes are affected by our choice of coupling counterterms in~\eqref{eq:couplings-special-case3} and~\eqref{eq:couplings-special-casen}.
It can be seen that the dependence of the full one-loop amplitude by $\gamma$ and $\gamma_3$ is just
\begin{multline}
     M_{ab}^{(1)}(\theta_{pp'})=2 (\dd_3-\dd) M_{ab}^{(0)}(\theta_{pp'}) +\frac{\bigl(M^{(0)}_{ab}(\theta_{p p'}) \bigl)^2}{8 m_a m_b \sinh{\theta_{p p'}}}\\
     +\frac{i \ai}{2 \pi} m_a m_b \sinh{\theta_{p p'}} \theta_{p p'} \frac{\partial}{\partial \theta_{p p'}} S^{(0)}_{ab}(\theta_{p p'}) \sum^r_{e=1} m_e^2 \\ 
     +\frac{i}{\pi} m_a m_b \sinh{\theta_{pp'}}  \sum_{e=1}^r \frac{\partial}{\partial \theta_{p p'}} \ \text{p.v.} \int^{+\infty}_{-\infty} d\theta_k \,  S^{(0)}_{ae} (\theta_{pk} ) S^{(0)}_{b e} (\theta_{p' k}) \,.
\end{multline}
This one-loop amplitude simplifies further if we set $\gamma_3 = \gamma$, which following \eqref{eq:couplings-special-casen} would imply that
\begin{equation}
\label{eq:SL-affine-toda-cond}
 \frac{\delta m_a^2}{m_a^2}= \frac{\delta C^{(n)}_{a_1 \cdots a_n}}{C^{(n)}_{a_1 \cdots a_n}} = \dd  \,.
\end{equation}
This is exactly what happens in simply-laced affine Toda models when we choose to renormalize the overall mass scale of the potential but not the interaction scale given by a dimensionless coupling $\g$, as shown for example in~\cite{Braden:1991vz,Braden:1990qa}.
Thus if the mass ratios do not renormalise at one-loop, if we require~\eqref{eq:SL-affine-toda-cond} and introduce the normalization factor in~\eqref{S_matrix_amplitude_connection}, we obtain the following closed expression for the one-loop S-matrices
\begin{equation}
\label{eq:result_S_mat_eq_m_ren_2}
\begin{split}
S^{(1)}_{ab}(\theta_{pp'})&=\frac{\bigl(S^{(0)}_{ab}(\theta_{p p'}) \bigl)^2}{2}+\frac{i \ai}{8 \pi}  \theta_{p p'} \frac{\partial}{\partial \theta_{p p'}} S^{(0)}_{ab}(\theta_{p p'}) \sum^r_{e=1} m_e^2 \\
&+\frac{i}{4 \pi}  \sum_{e=1}^r \frac{\partial}{\partial \theta_{p p'}} \ \text{p.v} \int^{+\infty}_{-\infty} d\theta_k \   S^{(0)}_{ea} (\theta_{k p} ) S^{(0)}_{e b} (\theta_{k p'}) \,,
\end{split}
\end{equation}
which depends only on the tree-level S-matrix elements and not on the counterterms.

\subsection{Analytic continuation to the complex rapidity plane}

So far we have worked with the assumption that the rapidities $\theta_{p}$ and $\theta_{p'}$ of the scattered particles were real with $\theta_{p}> \theta_{p'}$. It is however possible to analytically continue formula~\eqref{eq:result_S_mat_eq_m_ren_2} to the full complex plane. To do that we first remove the principal value prescription by surrounding the two potential singularities of the integral at $\theta_k=\theta_p$ and $\theta_k=\theta_{p'}$ with two half-circles of infinitesimal radius $\epsilon$ and subtract the contributions of these circles. 
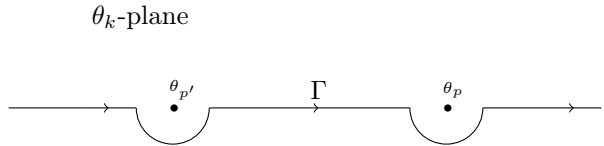
\begin{figure}[t]
\begin{center}
\begin{tikzpicture}
\tikzmath{\y=1.2;}

\filldraw[black] (5*\y,1*\y)  node[anchor=west] {\small{$\theta_{k}$-plane}};
\draw[] (4.2*\y,0*\y) -- (5*\y,0*\y);
\draw[->] (5*\y,0*\y) -- (5.3*\y,0*\y);
\draw[] (5.3*\y,0*\y) -- (5.6*\y,0*\y);
\draw[] (5.6*\y,0*\y) arc(-180:0:0.4*\y);
\draw[->] (6.4*\y,0*\y) -- (7.6*\y,0*\y);
\draw[] (7.6*\y,0*\y) -- (8.6*\y,0*\y);
\draw[] (8.6*\y,0*\y) arc(-180:0:0.4*\y);
\draw[->] (9.4*\y,0*\y) -- (10.4*\y,0*\y);
\draw[] (10.4*\y,0*\y) -- (10.7*\y,0*\y);
\filldraw[black] (5.85*\y,0*\y)  node[anchor=west] {\tiny{$\bullet$}};
\filldraw[black] (8.85*\y,0*\y)  node[anchor=west] {\tiny{$\bullet$}};
\filldraw[black] (5.85*\y,0.2*\y)  node[anchor=west] {\tiny{$\theta_{p'}$}};
\filldraw[black] (8.85*\y,0.2*\y)  node[anchor=west] {\tiny{$\theta_{p}$}};

\filldraw[black] (7.4*\y,0.2*\y)  node[anchor=west] {\small{$\ip$}};

\end{tikzpicture}
\caption{Integration path entering the analytically continued formula \eqref{eq:result_S_mat_eq_m_ren} for the one-loop S-matrix.}
\label{Defermation_of_real_thetak_integration_line}
\end{center}
\end{figure}
Labelling the new integration path by $\ip$, which is shown in Figure~\ref{Defermation_of_real_thetak_integration_line}, the formula for the one-loop renormalized S-matrix can be written as
\begin{equation}
\label{eq:result_S_mat_eq_m_ren}
\begin{split}
S^{(1)}_{ab}(\theta_{pp'})&=\frac{\bigl(S^{(0)}_{ab}(\theta_{p p'}) \bigl)^2}{2}+\frac{i \ai}{8 \pi}  \theta_{p p'} \frac{\partial}{\partial \theta_{p p'}} S^{(0)}_{ab}(\theta_{p p'}) \sum^r_{e=1} m_e^2 \\
&+\frac{i}{4 \pi}  \sum_{e=1}^r \frac{\partial}{\partial \theta_{p p'}} \int_{\ip} d\theta_k S^{(0)}_{ea} (\theta_{k p} ) S^{(0)}_{e b} (\theta_{k p'})\\
&+\frac{1}{16}  \left( \frac{M^{(0)}_{aa}(0)}{m_a^2}- \frac{M^{(0)}_{bb}(0)}{m_b^2} \right) \frac{\partial}{\partial \theta_{p p'}} S^{(0)}_{ab}(\theta_{p p'}) \,.
\end{split}
\end{equation}
We can now move the rapidities $\theta_p$ and $\theta_{p'}$ in the complex plane; this is possible as far as there are no singularities of $S^{(0)}_{ea}$ and $S^{(0)}_{eb}$ overlapping the integration path. If a singularity is approaching $\ip$ then it drags the integration path which needs to be deformed to avoid poles on the path. In summary: formula~\eqref{eq:result_S_mat_eq_m_ren} is the way to analytically continue to all values of $\theta_p$,$\theta_{p'} \in \mathbb{C}$. We will return to this point later when we will discuss the pole structure of the one-loop S-matrix.

It is interesting to note that the tree-level S-matrix satisfies:
\begin{equation}
\label{eq:crossing_equation}
S^{(0)}_{ea}(\theta)=-S^{(0)}_{\bar{e} a}(\theta \pm i \pi) \,.
\end{equation}
This is clear by the fact that under a shift $\theta \to \theta \pm i \pi$ the Mandelstam variables $s$ and $t$ are exchanged in 
the tree-level amplitude $M^{(0)}_{ea}(\theta)$ and therefore it needs to hold that
\begin{equation}
M^{(0)}_{ea}(\theta)=M^{(0)}_{\bar{e}a}(\theta \pm i \pi) \,.
\end{equation}
This is known as crossing symmetry and is a property satisfied individually by each of the tree-level Feynman diagrams.
Note that the S-matrix is connected to $M^{(0)}_{ea}(\theta)$ through the normalization factor~\eqref{S_matrix_amplitude_connection} and for this reason it satisfies~\eqref{eq:crossing_equation}. 

From this fact, it follows that the function defined by
\begin{equation}
\sum_{e=1}^r S^{(0)}_{ea}(\theta_{kp}) S^{(0)}_{eb}(\theta_{kp'}),
\end{equation}
is periodic in $\theta_k$ with period $i \pi$.
We then introduce an auxiliary parameter $\beta$ and use this periodicity to rewrite the integral on $\Gamma$ appearing in~\eqref{eq:result_S_mat_eq_m_ren} as the integral on a closed contour $\Gamma_n$ defined as a rectangle of infinite horizontal side and vertical side equal to $i n \pi$ with $n \in \mathbb{N}$ and $n \ge 1$. Then we can write
\begin{equation}
\label{eq:introducing_beta}
\begin{split}
\sum^r_{e=1}\int_{\ip} d\theta_k S^{(0)}_{ea} (\theta_{k p} ) S^{(0)}_{e b} (\theta_{k p'}) &= \lim_{\beta \to 0}  \sum^r_{e=1} \int_{\ip} d\theta_k e^{i \beta \theta_k }S^{(0)}_{ea} (\theta_{k p} ) S^{(0)}_{e b} (\theta_{k p'})\\
&=\lim_{\beta \to 0} \frac{1}{1 - e^{- n \pi \beta}} \sum^r_{e=1} \oint_{\ip_n} d\theta_k e^{i \beta \theta_k }S^{(0)}_{ea} (\theta_{k p} ) S^{(0)}_{e b} (\theta_{k p'}) \,.
\end{split}
\end{equation}
The new integration contour $\ip_n$ is shown in Figure~\ref{Clesed_Gamma_integration_line}.
\begin{figure}
\begin{center}
\begin{tikzpicture}
\tikzmath{\y=1.2;}

\draw[->] (4*\y,0*\y) -- (5.3*\y,0*\y);
\draw[] (5.3*\y,0*\y) -- (5.6*\y,0*\y);
\draw[] (5.6*\y,0*\y) arc(-180:0:0.4*\y);
\draw[->] (6.4*\y,0*\y) -- (7.6*\y,0*\y);
\draw[] (7.6*\y,0*\y) -- (8.6*\y,0*\y);
\draw[] (8.6*\y,0*\y) arc(-180:0:0.4*\y);
\draw[->] (9.4*\y,0*\y) -- (11*\y,0*\y);
\draw[] (11*\y,0*\y) -- (12*\y,0*\y);
\draw[->] (12*\y,0*\y) -- (12*\y,1*\y);
\draw[] (12*\y,1*\y) -- (12*\y,2*\y);
\draw[->] (12*\y,2*\y) -- (11*\y,2*\y);
\draw[] (11*\y,2*\y) -- (9.4*\y,2*\y);
\draw[] (8.6*\y,2*\y) arc(-180:0:0.4*\y);
\draw[->] (8.6*\y,2*\y) -- (7.6*\y,2*\y);
\draw[] (7.6*\y,2*\y) -- (6.4*\y,2*\y);
\draw[] (5.6*\y,2*\y) arc(-180:0:0.4*\y);
\draw[->] (5.6*\y,2*\y) -- (5.3*\y,2*\y);
\draw[] (5.3*\y,2*\y) -- (4*\y,2*\y);
\draw[->] (4*\y,2*\y) -- (4*\y,1*\y);
\draw[] (4*\y,1*\y) -- (4*\y,0*\y);

\filldraw[black] (5.85*\y,0*\y)  node[anchor=west] {\tiny{$\bullet$}};
\filldraw[black] (8.85*\y,0*\y)  node[anchor=west] {\tiny{$\bullet$}};
\filldraw[black] (5.85*\y,0.2*\y)  node[anchor=west] {\tiny{$\theta_{p'}$}};
\filldraw[black] (8.85*\y,0.2*\y)  node[anchor=west] {\tiny{$\theta_{p}$}};

\filldraw[black] (5.85*\y,2*\y)  node[anchor=west] {\tiny{$\bullet$}};
\filldraw[black] (8.85*\y,2*\y)  node[anchor=west] {\tiny{$\bullet$}};
\filldraw[black] (5.85*\y,2.2*\y)  node[anchor=west] {\tiny{$\theta_{p'}+i n \pi$}};
\filldraw[black] (8.85*\y,2.2*\y)  node[anchor=west] {\tiny{$\theta_{p}+i n \pi$}};

\filldraw[blue] (5.85*\y,0.4*\y)  node[anchor=west] {\tiny{$\times$}};
\filldraw[blue] (5.85*\y,1*\y)  node[anchor=west] {\tiny{$\times$}};

\filldraw[red] (8.85*\y,0.7*\y)  node[anchor=west] {\tiny{$\times$}};
\filldraw[red] (8.85*\y,1.1*\y)  node[anchor=west] {\tiny{$\times$}};
\filldraw[red] (8.85*\y,1.5*\y)  node[anchor=west] {\tiny{$\times$}};

\filldraw[black] (7.4*\y,-0.3*\y)  node[anchor=west] {\small{$\ip_n$}};
\end{tikzpicture}
\caption{Contour used to perform the integral of the product of the tree-level S-matrices.  The coloured marks correspond to the poles of the two tree-level S-matrices in the integrand.}
\label{Clesed_Gamma_integration_line}
\end{center}
\end{figure}
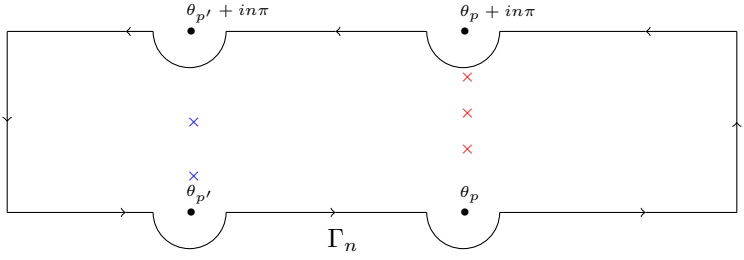
Taylor expanding the numerator and denominator of the expression on the r.h.s. of~\eqref{eq:introducing_beta} in $\beta$ and taking the limit $\beta \to 0$ we obtain finally that
\begin{equation}
\label{eq:integral}
\sum^r_{e=1}\int_{\ip} d\theta_k S^{(0)}_{ea} (\theta_{k p} ) S^{(0)}_{e b} (\theta_{k p'}) = \frac{i}{n \pi} \sum^r_{e=1} \oint_{\ip_n} d\theta_k \theta_k S^{(0)}_{ea} (\theta_{k p} ) S^{(0)}_{e b} (\theta_{k p'}) \,.
\end{equation}
The integral can now be computed by using Cauchy's residue theorem in terms of the poles and residues of the tree-level S-matrices. 
Then formula~\eqref{eq:result_S_mat_eq_m_ren} for the renormalized one-loop S-matrices can be written as
\begin{equation}
\label{eq:result_S_mat_eq_m_ren_3}
\begin{split}
S^{(1)}_{ab}(\theta_{pp'})&=\frac{\bigl(S^{(0)}_{ab}(\theta_{p p'}) \bigl)^2}{2}+\frac{i \ai}{8 \pi}  \theta_{p p'} \frac{\partial}{\partial \theta_{p p'}} S^{(0)}_{ab}(\theta_{p p'}) \sum^r_{e=1} m_e^2 \\
&-\frac{1}{4n \pi^2}  \sum_{e=1}^r \frac{\partial}{\partial \theta_{p p'}} \oint_{\ip_n} d\theta_k \theta_k S^{(0)}_{ea} (\theta_{k p} ) S^{(0)}_{e b} (\theta_{k p'})\\
&+\frac{1}{16}  \left( \frac{M^{(0)}_{aa}(0)}{m_a^2}- \frac{M^{(0)}_{bb}(0)}{m_b^2} \right) \frac{\partial}{\partial \theta_{p p'}} S^{(0)}_{ab}(\theta_{p p'}) \,.
\end{split}
\end{equation}

\subsection{Some examples}

In this section, we will test the simplified expression \eqref{eq:result_S_mat_eq_m_ren_3} for some integrable models and then compare it with the one-loop expansion of their bootstrapped S-matrices. For this purpose, we recall from~\cite{Braden:1989bu} the following building blocks for the bootstrapped S-matrices:
\begin{equation}
\label{eq:building-blocks}
    (x)_{\theta} =  \frac{\sinh\left( \frac{\theta}{2} + \frac{i \pi x}{2h} \right)}{\sinh\left( \frac{\theta}{2} - \frac{i \pi x}{2h} \right)} \ \ \ \textrm{,} \ \ \ 
    \{x\}_{\theta} =  \frac{(x-1)_{\theta} (x+1)_{\theta}}{(x-1+B)_{\theta} (x+1-B)_{\theta}} \,.
\end{equation}
In the expressions above $h$ is an integer specific to each model and $B$ is a function of the Lagrangian coupling $\g$ defined by
\begin{equation}
\label{eq:B_function_of_g}
B = \frac{\g^2}{2\pi}\frac{1}{1+\g^2 /4\pi} \,.
\end{equation}

\paragraph{Sinh-Gordon model.}

Let us first consider the sinh-Gordon theory defined by the following Lagrangian
\begin{equation}
\label{eq:sG_Lagrangian}
    \mathcal{L}=\frac{\partial_\mu \phi  \partial^\mu \phi}{2} - \frac{m^2}{g^2} \left( \cosh(\sqrt{2} g\phi)-1 \right) \,.
\end{equation}
It is a known fact that the sinh-Gordon S-matrix is the analytically continued breather-breather S-matrix of the sine-Gordon theory and is given by~\cite{Zamolodchikov:1978xm}
\begin{equation}
\label{eq:sG_S_Matrix}
    S(\theta) = \frac{\sinh{\theta} - i \sin \bigl(\frac{\pi B}{2}\bigl)}{\sinh{\theta} + i \sin \bigl(\frac{\pi B}{2} \bigl)} \,,
\end{equation}
where $B$ is given in~\eqref{eq:B_function_of_g} in terms of the Lagrangian coupling $\g$ and we label the difference between the rapidities of the scattered particles by $ \theta \equiv \theta_{p p'}$. Expressed in terms of the building blocks~\eqref{eq:building-blocks} this is equivalent to
\begin{equation}
    S(\theta) = \{1\}_{\theta} \quad \text{with} \quad h=2
\end{equation}
and corresponds to the S-matrix of the $a^{(1)}_1$ affine Toda theory in the classification presented in~\cite{Braden:1989bu}.

As usual, amplitudes can be computed perturbatively order by order in $\g$ after having expanded the Lagrangian around the stable vacuum ($\phi=0$) 
\begin{equation}
    \mathcal{L}=\frac{\partial_\mu \phi  \partial^\mu \phi}{2} -m^2 \phi^2 - m^2 \sum_{n=2}^{\infty} \frac{2^{n}}{(2n)!} g^{2n-2} \phi^{2n} \,.
\end{equation}
In this manner, the Lagrangian is written in the form~\eqref{eq0_1} and is now suitable for Feynman diagrams computations.
In this case, the tree-level S-matrix and the value of the tree-level two-to-two amplitudes at infinity (properly normalised by the masses of the scattered particles) are
\begin{equation}
    S^{(0)}(\theta) = -\frac{i \g^2 }{2 \sinh{\theta}} \ \ \ \textrm{and} \ \ \ \ai = -\frac{i \g^2}{m^2}.
\end{equation}
Inserting these tree-level data as input in~\eqref{eq:result_S_mat_eq_m_ren_3} we obtain
\begin{equation}
\label{eq:1-loop-SG}
     S^{(1)}(\theta) = -\frac{\g^4 (\pi -i \sinh {\theta}) }{8 \pi
   \sinh^2{\theta}}.
\end{equation}
The result has been easily generated by taking $n=1$ in~\eqref{eq:result_S_mat_eq_m_ren_3} and noting that the only poles contributing to the integral in~\eqref{eq:result_S_mat_eq_m_ren_3} are the collinear singularities at $\theta_k=\theta_p$ and $\theta_k=\theta_{p'}$. It is easy to show that~\eqref{eq:1-loop-SG} exactly matches the expansion of~\eqref{eq:sG_S_Matrix} at the order $\g^4$.

\paragraph{Bullogh–Dodd model.} A more interesting theory is the so-called Bullogh–Dodd model, defined by the Lagrangian
\begin{equation}
    \mathcal{L}=\frac{\partial_\mu \phi  \partial^\mu \phi}{2} - \frac{m^2}{6g^2} \left( 2 e^{g\phi}+ e^{-2g\phi} - 3 \right) \,,
\end{equation}
 whose expansion around $\phi=0$ is 
\begin{equation}
    \mathcal{L}=\frac{\partial_\mu \phi  \partial^\mu \phi}{2} -\frac{m^2 \phi^2}{2} - \frac{m^2}{6} \sum_{n=3}^{\infty} \frac{(2+ (-2)^n ) g^{n-2}}{n!} \phi^n.
\end{equation}
Unlike before, this Lagrangian contains a $3$-point coupling and the associated tree-level S-matrix has contributions from bound states propagating in the $s$-, $t$- and $u$-channels. We can write it as
\begin{equation}
\begin{split}
S^{(0)}(\theta)&= -\frac{1}{4\sinh(\theta)} \bigl(\frac{i m^2 g^2}{s-m^2}+\frac{i m^2 g^2}{t-m^2}+\frac{i m^2 g^2}{u-m^2} + 3i g^2 \bigr)\\
&= -i g^2 \frac{\cosh (2 \theta)}{\sinh(3 \theta)} \,,
\end{split}
\end{equation}
where the second equality is obtained by using the definitions for the Mandelstam variables in~\eqref{eq:Mandelstam}.  Using that here $\ai=-2ig^2/m^2$, we can insert the tree-level S-matrix in \eqref{eq:result_S_mat_eq_m_ren_3} and get the following one-loop expression:
\begin{equation}
\label{eq:1-loop-BD}
S^{(1)}(\theta)= \frac{\left( S^{(0)}(\theta) \right)^2}{2} - \frac{g^2}{4\pi} S^{(0)}(\theta)  + \frac{ig^2}{36} \frac{\partial}{\partial \theta} \left(S^{(0)}\left(\theta+\frac{i\pi}{3}\right)-S^{(0)}\left(\theta-\frac{i\pi}{3}\right) \right) \,.
\end{equation}
In this case, the result of the integral in~\eqref{eq:result_S_mat_eq_m_ren_3} was obtained by choosing the integration contour $\ip_1$ and summing over the poles in the contour, which are located at $\theta_k = \theta_p$, $\theta_k = \theta_p + i \frac{\pi}{3}$, $\theta_k = \theta_p + i \frac{2\pi}{3}$, $\theta_k = \theta_{p'}$, $\theta_k = \theta_{p'} + i \frac{\pi}{3}$ and $\theta_k = \theta_{p'} + i \frac{2 \pi}{3}$.
The conjectured non-perturbative S-matrix for this model was found in~\cite{Arinshtein:1979pb} and can be written in terms of the building blocks~\eqref{eq:building-blocks} as
\begin{equation}
S(\theta)=\{ 1 \}_{\theta} \{ 2 \}_{\theta} \ \ \ \textrm{with} \ \ \ h=3 \,.
\end{equation}
This corresponds to the S-matrix of affine Toda theories of $a^{(2)}_2$ type, obtained by folding the simply-laced theory $d^{(1)}_4$. S-matrices of $a_{2n}^{(2)}$ affine Toda models (with $n=1, \,2, \,\dots$) can be found for example in Section 7 of~\cite{Corrigan:1993xh}. It is easy to check that the one-loop expansion of this S-matrix matches exactly the expression~\eqref{eq:1-loop-BD} we derived.

\paragraph{$\mathbf{d_4^{(1)}}$ affine Toda theory.} 
The $d_4^{(1)}$ model belongs to the class of simply-laced affine Toda theories and, as such, its mass ratios do not renormalise at one-loop\footnote{This was of course the case also for the sinh-Gordon and Bullough-Dodd models previously considered since they contain a single particle.} \cite{Braden:1989bu}. A detailed study of the one-loop S-matrices of simply-laced affine Toda theories will be given in Section~\ref{sec:SL-affine-toda} while here we just focus on the properties of this particular model. 
It features four massive scalar fields $\phi_1$, $\phi_2$, $\phi_3$, and $\phi_4$ carrying masses
\begin{equation}
    m_1 = m_2 = m_3 = \sqrt{2} \ \ \ \textrm{and} \ \ \ m_4 = \sqrt{6} \ .
\end{equation}
The nonvanishing cubic couplings are
\begin{equation}
\begin{split}
   &C^{(3)}_{114}=C^{(3)}_{224}=C^{(3)}_{334}= - \sqrt{2} \g \,,\\
   &C^{(3)}_{123}= \sqrt{2} \g \quad, \quad C^{(3)}_{444}= 3 \sqrt{2} \g  \,,
   \end{split}
\end{equation}
and are responsible for the propagation of bound states in the tree-level S-matrices. 
The constant $\g$ common to all couplings is the interaction scale of the model. The exact S-matrix elements of this model can be found for example in~\cite{Braden:1989bu} in terms of the building blocks~\eqref{eq:building-blocks}. In this case $h=6$ and we can express the S-matrix elements as
\begin{equation}
\label{eq:SM-d4}
    \begin{array}{l}
    S_{12} (\theta) = S_{23} (\theta) = S_{13} (\theta) = \{ 3 \}_{\theta}  \\
    S_{11} (\theta) = S_{22} (\theta) = S_{33} (\theta) = \{ 1 \}_{\theta} \{ 5 \}_{\theta}, \\ 
    S_{14} (\theta) = S_{24} (\theta) = S_{34} (\theta) = \{ 2 \}_{\theta} \{ 4 \}_{\theta}, \\
    S_{44} (\theta) = \{ 1 \}_{\theta} \{ 3 \}_{\theta}^2 \{ 5 \}_{\theta}.
    \end{array}
\end{equation}

The one-loop S-matrices of the different two-to-two processes can be obtained in two different ways. The first way is to expand the expressions in~\eqref{eq:SM-d4} to the order $\g^4$, which corresponds to one-loop in perturbation theory. The second way is 
to extract the tree-level S-matrices by computing the order $\g^2$ of the expressions in~\eqref{eq:SM-d4} and plug these tree-level S-matrices into~\eqref{eq:result_S_mat_eq_m_ren_3}. For all processes, we obtained an exact matching of the two expressions, confirming the correctness of formula~\eqref{eq:result_S_mat_eq_m_ren_3}.
This check is quite non-trivial since unlike before, the integral \eqref{eq:integral} contributes with four terms due to the presence of four different stable particles. 

Having tested the one-loop formula \eqref{eq:result_S_mat_eq_m_ren_3} for some models, we now 
show how first- and second-order singularities at one loop are captured by our relation. In this regard, simple poles are due to bound state propagators (the same that appear in tree-level diagrams) attached to vertex corrections. Second-order poles correspond instead to Landau singularities and are associated with Feynman diagrams of Coleman-Thun type~\cite{Coleman:1978kk}.

\subsection{Unitarity cuts and Landau poles}

While expression~\eqref{eq:result_S_mat_eq_m_ren_3} was successful in reproducing the one-loop S-matrices of well-known examples of integrable theories, it is important to mention that another one-loop formula for these S-matrices was previously obtained in~\cite{Bianchi:2013nra,Bianchi:2014rfa}. 
In particular, in~\cite{Bianchi:2014rfa} a universal formula for the one-loop S-matrices of generic integrable theories was advanced and subsequently tested on different non-trivial examples, 
finding evidence that (up to missing rational terms always proportional to the tree-level S-matrix) the one-loop S-matrices of integrable theories are cut constructible.

For the class of theories defined by~\eqref{eq0_1} and satisfying Property~\ref{Condition_tree_level_elasticity_introduction}, the formula proposed in~\cite{Bianchi:2014rfa} reduces to\footnote{We write the formula using the conventions presented in Section \ref{sec:main_results_conventions}, which are slightly different compared with those used in~\cite{Bianchi:2013nra,Bianchi:2014rfa}.}
\begin{equation}
\label{Ben_Lorenzo_Valentina_cut_formula_in_purely_elastic_models}
    S_{ab}^{(\text{un. cuts})}(\theta) = \frac{\left(S^{(0)}_{ab}(\theta) \right)^2}{2}  -\frac{i}{16 \pi} \left(\frac{M^{(0)}_{aa}(0)}{m_a^2} +\frac{M^{(0)}_{bb}(0) }{m_b^2} \right) S_{ab}^{(0)}(\theta).
\end{equation}
As before, we define $\theta \equiv \theta_{p p'}$. We should then ask if formulas~\eqref{Ben_Lorenzo_Valentina_cut_formula_in_purely_elastic_models} (obtained from unitarity cuts) and~\eqref{eq:result_S_mat_eq_m_ren_3} are equivalent. This is supported by the fact that both formulas reproduce the correct one-loop S-matrix of the sinh-Gordon model, which is quite surprising given that the only term that~\eqref{Ben_Lorenzo_Valentina_cut_formula_in_purely_elastic_models} and~\eqref{eq:result_S_mat_eq_m_ren_3} have in common is the square of the tree-level S-matrix. Noting this fact, we may hope that after some hidden simplification~\eqref{eq:result_S_mat_eq_m_ren_3} always reduce to the simpler formula~\eqref{Ben_Lorenzo_Valentina_cut_formula_in_purely_elastic_models}. However, this is not the case; indeed, as soon as we consider Lagrangians of type~\eqref{eq0_1} containing $3$-point couplings (as it is the case for the Bullogh-Dodd model and the $d^{(1)}_4$ theory previously considered) the formula in~\eqref{Ben_Lorenzo_Valentina_cut_formula_in_purely_elastic_models} does not reproduce the correct result, while formula~\eqref{eq:result_S_mat_eq_m_ren_3} does. More in general
we observe that for a model with bound states formula~\eqref{Ben_Lorenzo_Valentina_cut_formula_in_purely_elastic_models} transforms simple poles at the tree-level into double poles at one-loop. While these double poles are sometimes present in the one-loop S-matrices of integrable models, they require an explanation in terms of Feynman diagrams of Coleman-Thun type~\cite{Coleman:1978kk} and are not a feature common to all processes and integrable theories. Therefore the unitarity cut method used in~\cite{Bianchi:2014rfa} misses important rational terms that would cancel these undesired singularities.

Let us be more specific about the singularity structure of the S-matrix. First of all we notice that both formulas in~\eqref{eq:result_S_mat_eq_m_ren_3} and~\eqref{Ben_Lorenzo_Valentina_cut_formula_in_purely_elastic_models} satisfy unitarity at one loop, which is\footnote{Often the convention of defining $S_{ab}^{(0)}=i T_{ab}^{(0)}$ and $S_{ab}^{(1)}=i T_{ab}^{(1)}$ is used, from which one-loop unitarity is $\text{Im} \bigl( T_{ab}^{(1)}(\theta)\bigl) = |T_{ab}^{(0)}|^2 /2$.}
\begin{equation}
\label{eq:one_loop_phys_unitarity}
\text{Re} \bigl( S_{ab}^{(1)}(\theta)\bigl) = -\frac{|S_{ab}^{(0)}|^2}{2} = \frac{\bigl(S_{ab}^{(0)}\bigl)^2}{2}  \quad \text{for} \ \theta \in \mathbb{R} \,.
\end{equation}
The second equality above follows from the fact that both the tree-level amplitudes and S-matrices are purely imaginary for real values of $\theta$. Due to this fact the r.h.s. of equations~\eqref{eq:result_S_mat_eq_m_ren_3} and~\eqref{Ben_Lorenzo_Valentina_cut_formula_in_purely_elastic_models} is imaginary but for the factor proportional to the square of the tree-level S-matrix, which is real (even if this fact is not immediately evident from~\eqref{eq:result_S_mat_eq_m_ren_3}, we need to remember that~\eqref{eq:result_S_mat_eq_m_ren_3} is a reformulation of~\eqref{eq:result_S_mat_eq_m_ren_2} which clearly satisfies this property), and unitarity is so satisfied.
Despite the real part of the one-loop S-matrix must be proportional to the square of the tree-level S-matrix this does not imply that a pole of order one in the tree-level S-matrix corresponds to a pole of order two in the one-loop S-matrix. Indeed relation~\eqref{eq:one_loop_phys_unitarity} is satisfied only for real values of the rapidity; at the pole position (associated with the propagation of a bound state) $\theta$ becomes imaginary and in general we should expect a term cancelling the second order singularity arising from the square of the tree-level S-matrix. This term must be imaginary for $\theta \in \mathbb{R}$ and does not appear in~\eqref{eq:one_loop_phys_unitarity}. We now show how formula~\eqref{eq:result_S_mat_eq_m_ren_3} captures these important additional terms, responsible for the cancellation of the unwanted second-order poles in one-loop S-matrix of the Bullough-Dodd model; these terms are also responsible for the generation of second-order poles in those theories 
in which these poles are expected from the Landau analysis.
In the following, we focus on the one-loop expression \eqref{eq:result_S_mat_eq_m_ren} instead of \eqref{eq:result_S_mat_eq_m_ren_3} since this formulation of the  S-matrix is simpler for the analysis under discussion; the same conclusions follow from~\eqref{eq:result_S_mat_eq_m_ren_3} since the two formulas are equivalent.

Let us start with the simple case of the Bullough-Dodd model, for which formula~\eqref{eq:result_S_mat_eq_m_ren} reduces to
\begin{equation}
\label{eq:result_S_mat_eq_m_renBD}
\begin{split}
S^{(1)}(\theta)&=\frac{\bigl(S^{(0)}(\theta) \bigl)^2}{2}+\frac{\g^2}{4 \pi}  \theta \frac{\partial}{\partial \theta} S^{(0)}(\theta) +\frac{i}{4 \pi}  \frac{\partial}{\partial \theta_{p p'}} \int_{\ip} d\theta_k S^{(0)} (\theta_{k p} ) S^{(0)} (\theta_{k p'}) \,.
\end{split}
\end{equation}
The tree-level S-matrix of this model has two simple poles in the physical strip at $\theta= i\pi/3$ and $\theta= 2 \pi i/3$, associated with the propagation of bound states in the $t$- and $s$-channels, respectively.
The expansion of the tree-level S-matrix around these poles is
\begin{equation}
\begin{split}
S^{(0)}(\theta)\Bigl|_{\theta \sim i \frac{\pi}{3}} &=  \frac{-i \g^2}{6 \bigl( \theta - i \frac{\pi}{3} \bigl)} -i \frac{\g^2}{\sqrt{3}} + O\left(\theta - i \frac{\pi}{3}\right) \,,\\
S^{(0)}(\theta)\Bigl|_{\theta \sim i \frac{2 \pi}{3}} &=  \frac{+i \g^2}{6 \bigl( \theta - i \frac{2 \pi}{3} \bigl)} -i \frac{\g^2}{\sqrt{3}} + O\left(\theta - i \frac{2 \pi}{3}\right) \,.
\end{split}
\end{equation}
The first two terms on the r.h.s. of~\eqref{eq:result_S_mat_eq_m_renBD} turn bound state poles at the tree level into double poles at one loop. Then we should ask ourselves how the integral in~\eqref{eq:result_S_mat_eq_m_renBD} avoids these double poles. 

We consider the limit $\theta \to i \pi/3$ (the case $\theta \to 2\pi i/3$ can be similarly studied). To reach this singular point we should move $\theta_p$ upward into the complex plane. By doing so the tower of singularities of $S^{(0)} (\theta_{k p} )$, which lie on a vertical line passing through $\theta_p$ and are spaced from each other with distance $i\pi/3$ (see Figure~\ref{fig_sing_tower_BD} where these singularities are represented with crosses), move and some of them can eventually approach the real line. When this happens the integration path $\Gamma$ must be deformed to avoid the singularity. 
\begin{figure}
\begin{center}
\includegraphics*[width=0.5\textwidth]{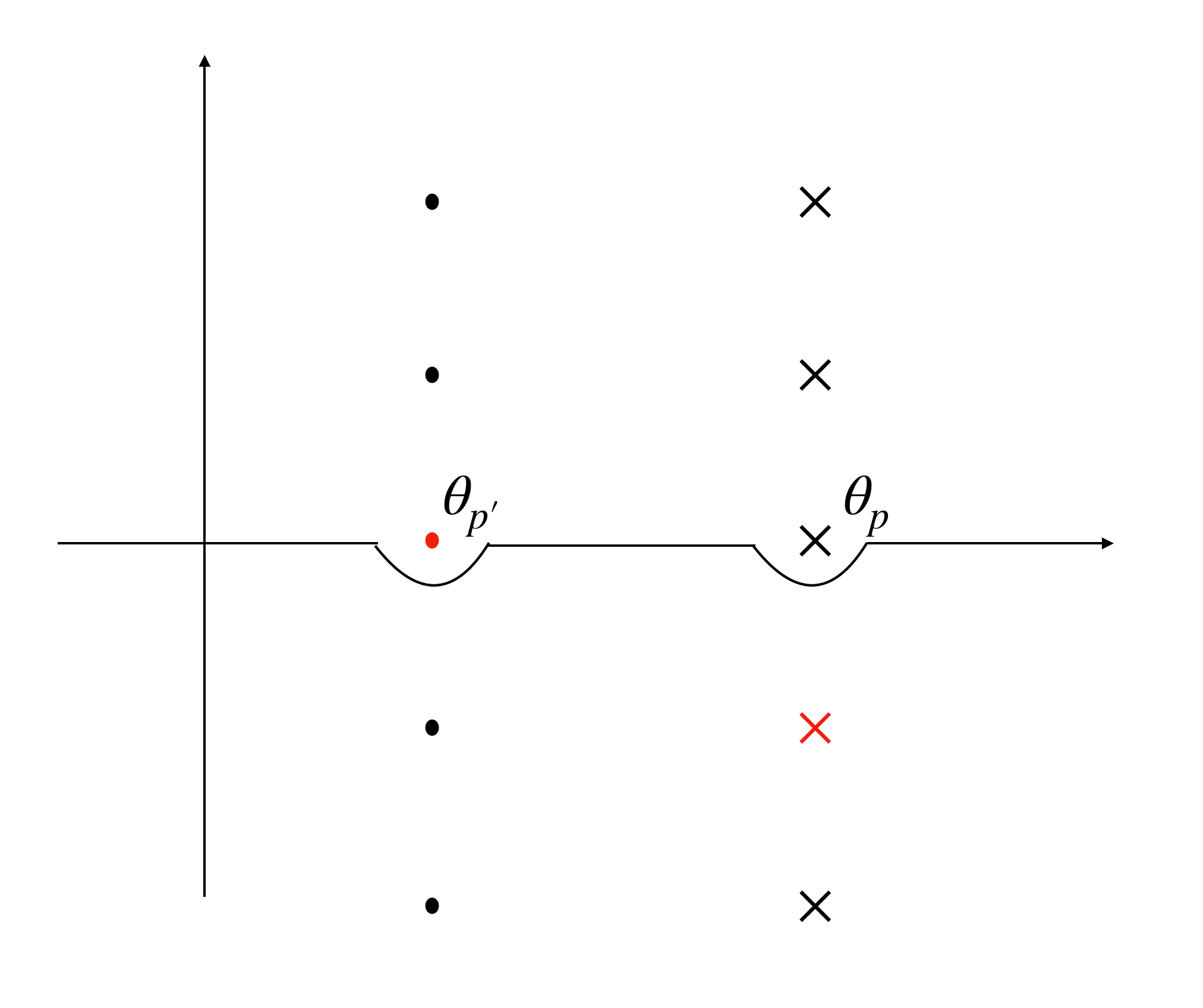} 
\end{center}
\caption{Towers of singularities of $S^{(0)} (\theta_{k p} )$ (represented with bullets) and $S^{(0)} (\theta_{k p'})$ (represented with crosses) for the Bullough-Dodd S-matrix. The two singularities coloured red trap the contour in the limit $\theta_{p p'} \to i\pi/3$ generating a pole in the integrated expression.}
\label{fig_sing_tower_BD}
\end{figure}
However, pairs of these singularities may approach the contour from opposite sides, trapping the path $\Gamma$ which cannot be deformed anymore. This is exactly what happens when we take the limit $\theta \to i \pi/3$. As depicted in Figure~\ref{fig_sing_tower_BD_trap}, in this case the collinear singularity at $\theta_k=\theta_{p'}$ and the bound state singularity at $\theta_k=\theta_{p} - i \pi/3$ (coloured red in the figure) trap the contour.
\begin{figure}
\begin{center}
\includegraphics*[width=1\textwidth]{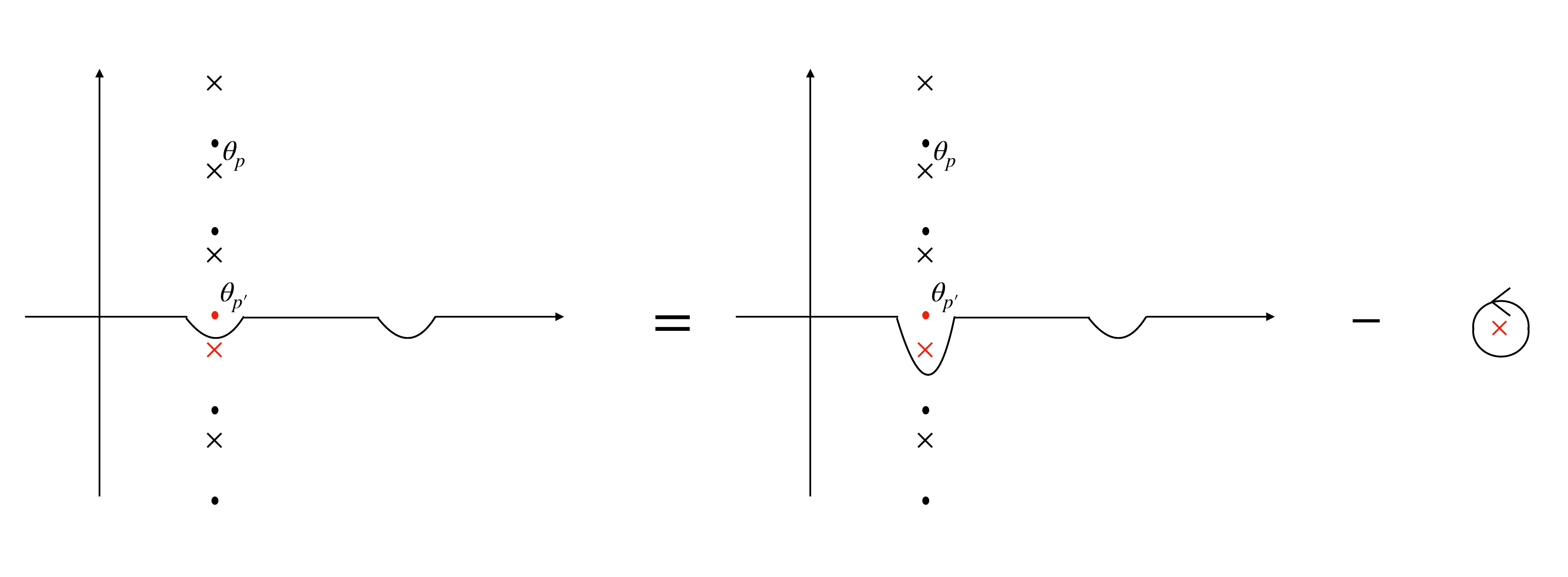} 
\end{center}
\caption{Deformation of contour $\Gamma$ when pairs of singularities trap the contour for the Bullough-Dodd model when $\theta_{p p'} \to i\pi/3$.}
\label{fig_sing_tower_BD_trap}
\end{figure}
Thus we change the contour $\Gamma$ discontinuously to bypass the singular point $\theta_{k}=\theta_{p} - i \pi/3$ and by doing so we pick its residue.  Now all the singular information is in this residue and the remaining integral with modified contour is finite as $\theta_{pp'}\rightarrow i\pi/3$. Then the singular behaviour of the integral in~\eqref{eq:result_S_mat_eq_m_renBD} at $\theta_{pp'}\rightarrow i\pi/3$ is given by
\begin{equation}
    \frac{\partial}{\partial \theta_{p p'}} \int_{\ip} d\theta_k S^{(0)} (\theta_{k p} ) S^{(0)} (\theta_{k p'}) \sim (2 \pi i) \frac{\partial}{\partial \theta_{p p'}} \textrm{Res}_{\theta_k = \theta_{p}- \frac{i \pi}{3}}(S^{(0)} (\theta_{k p} ) S^{(0)} (\theta_{k p'})).
\end{equation}
Inserting this residue into formula \eqref{eq:result_S_mat_eq_m_renBD} we correctly reproduce the simple pole at $\theta_{pp'}= i\pi/3$ of the Bullogh-Dodd model, avoiding the double pole. By doing a similar analysis we can reproduce the singular behaviour of the one-loop S-matrix at all the other poles.

This trapping mechanism generates not only simple poles corresponding to bound states but also Landau singularities, which at one loop manifest themselves as double poles (see, e.g., \cite{Dorey:2023cuq}). To see this let us look at the $d_4^{(1)}$ affine Toda model previously considered. 
It is possible to show that at the order $\g^4$, which corresponds to one-loop in perturbation theory, the S-matrix element $S_{14} (\theta_{pp'})$ has a double pole at $\theta_{p p'}= i \pi/2$.
This is a second-order Landau pole which is not related to the propagation of any bound states at tree-level~\cite{Braden:1990wx} and thus can be generated only by the integral contributions in \eqref{eq:result_S_mat_eq_m_ren}. Now these contributions are of the form
\begin{equation}
\label{eq:gen-integral}
    \sum^4_{e=1} \frac{\partial}{\partial \theta_{p p'}}  \int_{\ip} d\theta_k S^{(0)}_{e1} (\theta_{k p} ) S^{(0)}_{e 4} (\theta_{k p'}).
\end{equation}
We have to analyze which poles trap $\Gamma$ for each term in the sum, then change the contour to bypass these poles and finally pick their residue as in the Bullogh-Dodd case. 
We observe that each integral entering the sum~\eqref{eq:gen-integral} has a pair of poles trapping $\Gamma$ when we move $\theta_p \to \theta_{p'} + i\pi/2$. They are $\theta_k=\theta_{p'}+ i\pi/6$ and $\theta_k=\theta_{p}- i\pi/3$ for $e=1,2,3$, and $\theta_{p}-i\pi/6$ and $\theta_k=\theta_{p'}$ for $e=4$. By summing over all these residues we get the double pole at $\theta_{pp'} = i\pi/2$ with the correct coefficient. Similarly, we can also reproduce the correct one-loop coefficients of the Landau singularities at $\theta_{pp'} = i\pi/3$ and $\theta_{pp'} = 2i\pi/3$ in the S-matrix element $S_{44} (\theta_{pp'})$.

In summary, we find evidence that all singularities in the one-loop S-matrix, whether they are simple or anomalous thresholds, are correctly reproduced by formula~\eqref{eq:result_S_mat_eq_m_ren}; the residues at these singularities can be computed by analysing which tree-level poles trap the contour $\Gamma$. Then we deform $\Gamma$ to avoid each of these poles and in doing so we pick up a residue that contains the singular behavior. We then put this residue back into \eqref{eq:result_S_mat_eq_m_ren}. Summing over all residues we reproduce the correct one-loop coefficient of the expansion of the S-matrix at the singularity.
Clearly, this reproduces only the singular part of the S-matrix since the integral (with now a deformed contour) still contributes with a finite amount. While this analysis has been carried out only for a couple of models it would be interesting to connect it with the more standard approach of generating the residues at the poles by summing multiple Landau diagrams~\cite{Braden:1990wx}. It is possible that some universal features of these poles, such as the multiplicity of the singular one-loop networks discussed in~\cite{Dorey:2022fvs,Dorey:2023cuq}, have a correspondence with the number of trapping singularities discussed above.

Moreover, three additional conditions should be checked further: generalised unitarity, crossing and fusion. It is known that a combination of generalised unitarity and crossing should lead to S-matrices which are $2 \pi i$-periodic in $\theta$ (see, e.g., the discussion in~\cite{Dorey:1996gd})\footnote{This is the case only if the theory is purely elastic.}. From~\eqref{eq:result_S_mat_eq_m_ren} we note that there is a term proportional to $\theta$, which violates this periodicity.
Even though in all the models we studied we observed that
this $\theta$-proportional term always cancels with contributions coming from the integral we have not been able to prove this fact in full generality.
The same cancellation should allow for fusion, which also involves shifts in the rapidity. However, testing fusion is a more complicated issue, especially for theories in which the mass ratios renormalise, as is the case of nonsimply-laced affine Toda theories. In this case, the masses receive quantum corrections which completely change the analysis at weak coupling and for these theories equation~\eqref{eq:result_S_mat_eq_m_ren} does not apply anymore. We should instead study what counterterms are required by the absence of inelastic processes at one-loop and armed with these counterterms obtain the S-matrix by combining~\eqref{eq:definition_one_loop_amplitude}, \eqref{eq:final_result_amplitude} and \eqref{eq:countertermsII} and dividing by the normalization factor~\eqref{S_matrix_amplitude_connection}. It should then be interesting to check how fusion works in these more complicated theories at the level of one-loop S-matrices and confirm the conjecture advance in~\cite{Corrigan:1993xh} that the fusion relations are spoiled for certain singularities. We hope to return to some of these problems in the future. For now, we limit ourselves to checking the validity of our formula on the full class of simply-laced affine Toda theories, as reported in the next section.

\section{Simply-laced affine Toda theories}
\label{sec:SL-affine-toda}

Affine Toda models are a famous class of (1+1)-dimensional quantum field theories known to be classically integrable~\cite{Mikhailov:1980my, Olive:1984mb} and are believed to be also quantum integrable. Indeed their S-matrices were conjectured in the past using the bootstrap approach: results include simply- \cite{Arinshtein:1979pb, Freund:1989jq, Destri:1989pg, Christe:1989ah, Christe:1989my, Klassen:1989ui, Braden:1989bg, Braden:1989bu, Dorey:1990xa, Dorey:1991zp, Fring:1991gh} and nonsimply-laced models~\cite{Delius:1991cu, Delius:1991kt, Corrigan:1993xh, Dorey:1993np, Oota:1997un}, and certain supersymmetric extensions of nonsimply-laced theories \cite{Delius:1990ij,Delius:1991sv}. Despite these S-matrices, especially for simply-laced models, have passed multiple perturbative checks (see, e.g., \cite{Braden:1990wx, Braden:1991vz, Braden:1992gh, Braden:1990qa, Sasaki:1992sk, Dorey:2022fvs, Dorey:2023cuq}) a universal confirmation of their S-matrices at one-loop directly from the quantum field theory has never been achieved so far and this will be matter of this section.  
Our motivation for studying affine Toda theories comes from the fact that all these models universally satisfy Property~\ref{Condition_tree_level_elasticity_introduction} \cite{Dorey:2021hub}.

\subsection{Lagrangians and couplings}

Affine Toda theories are massive bosonic quantum field theories, each associated with a Lie algebra $\mathfrak{g}$ of rank $r$, describing the interaction of $r$ bosonic scalar fields $(\phi_1, \dots, \phi_r)$ through a Lagrangian
\begin{equation}
\label{eq:Toda_Lag}
    \mathcal{L}=\frac{\partial_{\mu} \phi_a \partial^{\mu}\phi_a}{2} - \frac{m^2}{\g^2} \left( \sum_{i=0}^{r} n_i e^{g \alpha^a_i \phi_a} -h \right) \,,
\end{equation}
where $m$ and $\g$ are the mass and the interaction scales of the model, respectively, and the sum over $a=1,\, \dots, \,r$ is implicit. The vectors $\{\alpha_i \}^r_{i=0}$ lying in $\mathbb{R}^r$ comprise both the simple roots of $\mathfrak{g}$ (which are $\alpha_1, \, \dots, \, \alpha_r$) and the lowest root $\alpha_0$
\begin{equation}
\label{eq:definition_min_root_a0}
\alpha_0 = -\sum_{i=1}^{r} n_i \alpha_i \,,
\end{equation}
necessary for the theory to have a stable vacuum at $\phi=0$, \footnote{The presence of this additional root in the potential also prevents the theory from being conformal, as discussed for example in~\cite{Corrigan:1994nd}.}. The integers  $\{ n_i \}^r_{i=1}$ are the Kac labels of $\mathfrak{g}$ and we set $n_0=1$. The constant $h$ is the Coxeter number, defined by 
\begin{equation}
h=\sum^r_{i=0} n_i \,.
\end{equation}

In this section, we focus on simply-laced theories, 
i.e., models associated with Lie algebras $\mathfrak{g}$ constructed from simply-laced Dynkin diagrams.
This class of theories is also known as the ADE series since it comprises all the $a_n$, $d_n$ algebras and some of the exceptional ones like $e_6$, $e_7$, and  $e_8$. For this class of models, all roots have the same length, which we assume to be $\sqrt{2}$.
As already mentioned in Section~\ref{sec:mass-scale-same}, a feature of simply-laced theories is that the mass ratios do not renormalise at one-loop and their one-loop S-matrices must therefore satisfy formula~\eqref{eq:result_S_mat_eq_m_ren_2} (or analogously \eqref{eq:result_S_mat_eq_m_ren_3}).

We can Taylor expand~\eqref{eq:Toda_Lag} around $\phi_a=0$ and put it in the same form as~\eqref{eq0_1}. Condition~\eqref{eq:definition_min_root_a0}, together with the fact that $n_0=1$, ensures that the expansion of~\eqref{eq:Toda_Lag} around $\phi_a=0$ does not have a linear term in $\phi_a$ and the model has a stable vacuum. Then the expansion takes the form
\begin{equation}
\label{eq:Lagrangian-AT}
    \mathcal{L} = \frac{\partial_{\mu} \phi_a \partial^{\mu}\phi_a}{2}  - \frac{(\mathbb{M}^2)_{ab}}{2}\phi_a \phi_b  - \sum_{n=3}^{\infty} \frac{1}{n!} \sum_{a_1, \dots, a_n=1}^r \tilde{C}^{(n)}_{a_1 \cdots a_n} \phi_{a_1} \cdots \phi_{a_n} \,,
\end{equation}
where the squared of the mass matrix $\mathbb{M}^2$ and the couplings $\tilde{C}^{(n)}_{a_1 \cdots a_n}$ are given by 
\begin{equation}
\label{eq:mass_matrix}
    (\mathbb{M}^2)_{ab} = m^2 \sum_{i=0}^{r} n_i \alpha_i^a \alpha_i^b \,,
\end{equation}
\begin{equation}
    \tilde{C}^{(n)}_{a_1 \cdots a_n} = m^2 g^{n-2} \sum_{i=0}^{r} n_i \alpha_i^{a_1} \cdots \alpha_i^{a_n} \,.
\end{equation}
Note that, differently from~\eqref{eq0_1}, the couplings $\tilde{C}^{(n)}_{a_1 \cdots a_n}$ cannot be directly used to compute Feynman diagrams since the mass matrix is not diagonal. 
To write the Lagrangian in the form of~\eqref{eq0_1}, suitable for computing scattering amplitudes
in perturbation theory, we need to diagonalize the squared of the mass matrix. In the basis in which this matrix is diagonal, it is possible to prove that all nonzero 3-point couplings satisfy the following area rule~\cite{Braden:1989bu,Fring:1991me}
\begin{equation}
C^{(3)}_{abc}= \pm \frac{4 \g}{\sqrt{h}} \Delta_{abc} \,,
\end{equation}
where $\Delta_{abc}$ is the area of the triangle having as sides the masses of the fusing particles $m_a$, $m_b$ and $m_c$. The signs of the couplings are related to the structure constants of the underlying Lie algebra in such a way as to ensure the cancellation of all singularities in two-to-two tree-level inelastic processes. Higher-order couplings can then be entirely written in terms of the masses and $3$-point couplings. This can be done either using universal properties of roots~\cite{Fring:1992tt} or imposing recursively absence of tree-level production processes as shown in~\cite{Gabai:2018tmm} and quickly reviewed in Section~\ref{sec:tree_level_int}. An interesting fact noted in~\cite{Dorey:2021hub} is that these two independent constructions led to the same set of recursion relations for higher-order couplings, meaning that Lie algebra information can in principle be extracted from the elastic tree-level scattering requirements. This observation was used in~\cite{Dorey:2021hub} to universally prove that all affine Toda field theories satisfy Property~\ref{Condition_tree_level_elasticity_introduction}.

\subsection{On the Coxeter geometry}
\label{sec:coxeter-geometry}

In this section we briefly review the Coxeter geometry of the root systems underlying simply-laced affine Toda field theories; this will be necessary to describe the exact S-matrices of these models. For more details on this topic, the reader is invited to look at~\cite{3d8d9559-20d1-3b1a-b70e-fcc39f9953fd}. 

Firstly, let us consider a semisimple Lie algebra $\mathfrak{g}$ of rank $r$. For any root $\alpha$ belonging to the root system of $\mathfrak{g}$ we define the Weyl reflection $w_{\alpha}$ as
\begin{equation}
    w_{\alpha} (x) = x- \frac{2 (x , \alpha)}{\alpha^2} x \ ,
\end{equation}
where $x\in\mathbb{R}^r$ and $(\cdot \,,\cdot)$ denotes the scalar product. Geometrically this operation reflects $x$ with respect to the hyperplane orthogonal to $\alpha$. Given a set of simple roots, the so-called Coxeter element $w$ is the product of Weyl reflections with respect to these simple roots. While there are various choices of ordering for this product, a useful one to describe the exact S-matrix is the Steinberg ordering~\cite{MR0106428}. 
This corresponds to splitting the simple roots into two sets, `black' ($\bullet$) and `white' ($\circ$), each containing mutually orthogonal roots. This means that two simple roots in the same set are not connected by any link in the Dynkin diagram, as shown in Figure~\ref{Dyn_diagram_with_balckwhite}.
Note that for each Lie algebra $\mathfrak{g}$ there are two possible ways of splitting the roots
into $\bullet$ and $\circ$, which comes from exchanging $\bullet\leftrightarrow \circ$; however this is inconsequential since both yield the same S-matrix~\cite{Dorey:1990xa, Dorey:1991zp}. 
\begin{figure}
\begin{center}
\begin{tikzpicture}
\tikzmath{\y=1.1;}

\filldraw[color=black,fill=black, very thick](0*\y ,0*\y) circle (0.1*\y);
\draw[] (0.1*\y,0*\y) -- (0.9*\y,0*\y);
\filldraw[color=black,fill=white, very thick](1*\y ,0*\y) circle (0.1*\y);
\draw[] (1.1*\y,0*\y) -- (1.9*\y,0*\y);
\filldraw[color=black,fill=black, very thick](2*\y ,0*\y) circle (0.1*\y);
\draw[] (2.1*\y,0*\y) -- (2.9*\y,0*\y);
\filldraw[color=black,fill=white, very thick](3*\y ,0*\y) circle (0.1*\y);
\draw[] (3.1*\y,0*\y) -- (3.9*\y,0*\y);
\filldraw[color=black,fill=black, very thick](4*\y ,0*\y) circle (0.1*\y);
\draw[] (2*\y,0.1*\y) -- (2*\y,0.9*\y);
\filldraw[color=black,fill=white, very thick](2*\y,1*\y) circle (0.1*\y);
\end{tikzpicture}
\end{center}
\caption{$E_6$ Dynkin diagram split into $\bullet$ and $\circ$ roots.}
\label{Dyn_diagram_with_balckwhite}
\end{figure}
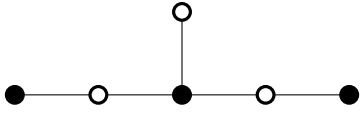
Then we define the Coxeter element as
\begin{equation}
    w= \prod_{\alpha \in \bullet}  w_{\alpha}\prod_{\beta \in \circ} w_{\beta}.
\end{equation}
The ordering of roots within each set is irrelevant since their Weyl reflections commute due to the orthogonality of the roots. 
It is possible to prove that $w$ is diagonalised by an orthonormal basis $\{ z_s \}$
\begin{equation}
\label{eq:eigenvectors_w}
    w z_s = e^{2 i \theta_s} z_s \ \ \ \text{with} \ \ \ \theta_s = \frac{\pi}{h} s
\end{equation}
with $s$ taking $r$ integer values in the set of exponents of $\mathfrak{g}$. Physically these are the spins of the conserved charges of the affine Toda model~\cite{Dorey:1990xa}.
From this fact it follows that $w$ is periodic with periodicity $h$ (i.e. $w^h=1$), being $h$ the Coxeter number of $\mathfrak{g}$. Moreover, no matter the Lie algebra considered $s=1$ and $s=h-1$ are always exponents and the eigenvectors $z_1$ and $z_{h-1}$ define the so-called spin-1 eigenplane of the Coxeter element.

Let $\{\lambda_a\}_{a=1}^{r}$ be the fundamental weights of $\mathfrak{g}$, which satisfy
\begin{equation}
(\lambda_a, \alpha_b)= \delta_{ab} \,.
\end{equation}
Then we can define a new set of roots $\{\gamma_a\}_{a=1}^{r}$
\begin{equation}
\label{eq:roots_weigths_connection}
\gamma_a=(1-w^{-1}) \lambda_a = \begin{cases} 
\alpha_a \hspace{15mm} \text{if} \  a \in \circ\\
-w^{-1} \alpha_a \hspace{5mm} \text{if} \ a \in \bullet\\
\end{cases} \,,
\end{equation}
from which generating the entire root system through the action of the Coxeter element~\cite{3d8d9559-20d1-3b1a-b70e-fcc39f9953fd}. Note that the first definition in~\eqref{eq:roots_weigths_connection} is general while the second one assumes a specific Steinberg ordering; nonetheless these definitions are equivalent. Acting with $w$ on each root $\gamma_a$ we generate an orbit of length $h$
\begin{equation}
    \Gamma_a = \{w^{-p} \gamma_a : p=0,\cdots, h-1\}.
\end{equation}
These orbits do not intersect and their union is the entire root system of $\mathfrak{g}$. It turns out that the particles of the affine Toda theories are in one-to-one correspondence with these orbits and that the masses and eigenvalues of the higher spin conserved charges can be obtained by projecting these orbits onto the different eigenplanes of $w$~\cite{Dorey:1990xa, Freeman:1991xw} (for a given spin $s$ the associated eigenplane of $w$ is spanned by the vectors $z_{s}$ and $z_{h-s}$).

For any root $\alpha$ we label the angle formed by $\alpha$ when projected onto the spin-1 eigenplane of the Coxeter element by $U_{\alpha}$. Then the action of $w$ rotates the projection of $\alpha$ on this plane by $2\pi/h$, which means we have
\begin{equation}
U_{w \alpha} = U_{\alpha}+ \frac{2 \pi}{h} \,.
\end{equation}
It is possible to choose a reference frame for the axis such that $U_{\gamma_a} = 0$ when $a \in \circ$, and $U_{\gamma_a} = -\pi/h$ when $a \in \bullet$. In Figure~\ref{fig:bw_orbit_representatives} we show the projections of $\gamma_a$, with $a \in \circ$ and $a \in \bullet$ onto the spin $1$ eigenplane of $w$.
If we define $u_{\alpha}$ by
\begin{equation}
\label{eq:angle-def}
    U_\alpha = \frac{\pi}{h} u_{\alpha},
\end{equation}
then for $\alpha\in\Gamma_a$ we have $u_{\alpha}= 2p - 1$ if $a \in \bullet$ and  $u_\alpha = 2p$ if $a \in \circ$ (with $p \in \mathbb{Z}$).

\begin{figure}
\begin{center}
\begin{tikzpicture}
\tikzmath{\y=0.8;}
\filldraw[black] (3*\y,2.5*\y)  node[anchor=west] {\footnotesize{$a \in \circ$}};
\draw[->](-3*\y,0*\y) -- (5*\y,0*\y);
\draw[->](0*\y,-2*\y) -- (0*\y,3*\y);
\draw[->,thick](0*\y,0*\y) -- (3.5*\y,0*\y);
\filldraw[black] (2*\y,-0.35*\y)  node[anchor=west] {\footnotesize{$P_1(\gamma_a)$}};

\filldraw[black] (13*\y,2.5*\y)  node[anchor=west] {\footnotesize{$a \in \bullet$}};
\draw[->](7*\y,0*\y) -- (15*\y,0*\y);
\draw[->](10*\y,-2*\y) -- (10*\y,3*\y);
\draw[->,thick](10*\y,0*\y) -- (10*\y+4*0.8*\y,-2*0.8*\y);
\filldraw[black] (11.6*\y,-1.75*\y)  node[anchor=west] {\footnotesize{$P_1(\gamma_a)$}};
\draw[][] (11.2*\y,0*\y) arc(0:-27:1.2*\y);
\filldraw[black] (11.2*\y,-0.4*\y)  node[anchor=west] {\tiny{$\frac{\pi}{h}$}};

\end{tikzpicture}
\end{center}
\caption{Projections of white (on the left) and black (on the right) orbit representatives on the spin $1$ eigenplane of the Coxeter element. After properly orienting the plane, the action of the Coxeter element corresponds to the rotation in the counter-clockwise direction by an angle $2 \pi/h$.}
\label{fig:bw_orbit_representatives}
\end{figure}
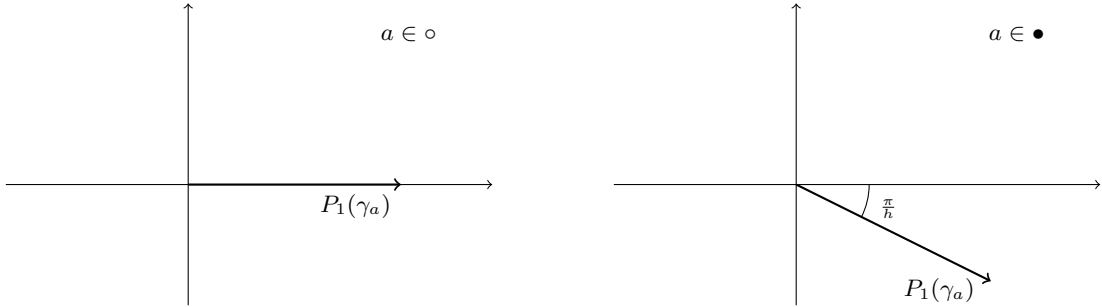

The Coxeter element is central in the construction of the exact S-matrix.

\subsection{Exact S-matrices for the ADE series}

With this brief review of the Coxeter geometry, we have now the tools to describe the exact S-matrix of simply-laced affine Toda models. Universal expressions for these S-matrices were proposed in~\cite{Dorey:1990xa, Dorey:1991zp} in terms of the building blocks 
\begin{equation}
\label{eq:bBlock_Patrick}
\{x\}_- \equiv \frac{(x-1)_- (x+1)_-}{(x-1+B)_- (x+1-B)_-} \,,
\end{equation}
where
\begin{equation}
(x)_- \equiv \frac{1}{\sinh \Bigl( \frac{\theta}{2} - \frac{i \pi x}{2h} \Bigl)} \,.
\end{equation}
The parameter $B$ appearing in~\eqref{eq:bBlock_Patrick} was previously defined in~\eqref{eq:B_function_of_g} and is a function of the coupling $\g$ appearing in the affine Toda Lagrangian. The S-matrices of the ADE series of affine Toda models can be written as functions of the roots and weights of the underlying Lie algebras as~\cite{Dorey:1990xa, Dorey:1991zp}
\begin{equation}
\label{eq:bootstrap_univ_S_mat}
S_{ab}(\theta)= \prod_{\beta \in \Gamma_b} \{1+ u_{\gamma_a \beta} \}_{-}^{(\lambda_a, \beta)} \ ,
\end{equation}
where $u_{\alpha \beta}$ is the difference between the angles $u_\alpha$ and $u_\beta$ in units of $\pi/h$ as defined in \eqref{eq:angle-def}. We remark that in \eqref{eq:bootstrap_univ_S_mat} the root $\gamma_a$, related to a particle of type $a$, is kept fixed in the orbit $\Gamma_a$ while we perform the product over all elements of $\Gamma_b$\footnote{It is possible to show that taking $d_4$ as an example one recovers all the S-matrices given in \eqref{eq:SM-d4}.}. In \cite{Dorey:1990xa, Dorey:1991zp}, with the use of Coxeter geometry tools, it was shown that these S-matrices universally satisfy unitarity, generalised unitarity, parity, crossing symmetry and fusion for the full class of simply-laced affine Toda theories.
Alternative expressions for these S-matrices can also be found in~\cite{Fring:1991gh}.

Given the exact S-matrix~\eqref{eq:bootstrap_univ_S_mat} we now derive its perturbative expansion, as given in~\eqref{eq:1-loop-def}. First, we expand the building blocks around the point $B=0$, which yields
\begin{multline}
 \{x\}_-= 1 - \frac{i \pi}{2h} B \left(1 - \frac{\pi B}{2 h} \cot \left(\frac{\pi}{h}\right) \right)\\
 \times \left[ \coth \left( \frac{\theta}{2} - \frac{i \pi}{2h} (x-1) \right) - \coth \left( \frac{\theta}{2} - \frac{i \pi}{2h} (x+1) \right) \right] + O(B^3) \,.
\end{multline}
Using the fact that at weak coupling $\g$
\begin{equation}
B=\frac{\g^2}{2 \pi} - \frac{\g^4}{8 \pi^2} + O(\g^6)
\end{equation}
then we get the following tree-level S-matrix
\begin{equation}
\label{eq:tree_level_S}
S^{(0)}_{ab}(\theta)= - \frac{i \g^2}{4h} \sum_{\beta \in \Gamma_b} (\lambda_a, \beta) \left[ \coth \left( \frac{\theta}{2} - \frac{i \pi}{2h} u_{\gamma_a \beta} \right) - \coth \left( \frac{\theta}{2} - \frac{i \pi}{2h} (u_{\gamma_a \beta}+2) \right) \right].
\end{equation}
By using that the Coxeter element acts as a rotation on the spin-1 eigenplane, that is
\begin{equation}
u_{\gamma_a  \beta} + 2 = u_{\gamma_a, w^{-1} \beta},
\end{equation}
we can rewrite the tree-level S-matrix \eqref{eq:tree_level_S} as
\begin{equation}
\label{eq:tree_level_S_2}
\begin{split}
S^{(0)}_{ab}(\theta)&= - \frac{i \g^2}{4h} \left[\sum_{\beta \in \Gamma_b} (\lambda_a, \beta)  \coth \left( \frac{\theta}{2} - \frac{i \pi}{2h} u_{\gamma_a \beta} \right) - \sum_{\beta \in \Gamma_b} (\lambda_a, \beta) \coth \left( \frac{\theta}{2} - \frac{i \pi}{2h} u_{\gamma_a, w^{-1} \beta} \right) \right]\\
&= - \frac{i \g^2}{4h} \left[\sum_{\beta \in \Gamma_b} (\lambda_a, \beta)  \coth \left( \frac{\theta}{2} - \frac{i \pi}{2h} u_{\gamma_a \beta} \right) - \sum_{\beta \in \Gamma_b} (\lambda_a, w \beta) \coth \left( \frac{\theta}{2} - \frac{i \pi}{2h} u_{\gamma_a \beta} \right) \right]\\
&= - \frac{i \g^2}{4h} \sum_{\beta \in \Gamma_b} ((1+w^{-1})\lambda_a, \beta)  \coth \left( \frac{\theta}{2} - \frac{i \pi}{2h} u_{\gamma_a \beta} \right)\\
&= - \frac{i \g^2}{4h} \sum_{\beta \in \Gamma_b} (\gamma_a, \beta)  \coth \left( \frac{\theta}{2}- \frac{i \pi}{2h} u_{\gamma_a \beta} \right) \, .
\end{split}
\end{equation}
In the last line of the expression above we used the definition in~\eqref{eq:roots_weigths_connection}. Analogously we can derive the one-loop S-matrix and obtain
\begin{equation}
\label{eq:one_loop_SIV}
S^{(1)}_{ab}(\theta)= \frac{\left( S^{(0)}_{ab}(\theta) \right)^2}{2}  -\frac{\g^2}{4 \pi} S^{(0)}_{ab}(\theta)+ \frac{\g^4}{32 h^2} \sum_{\beta \in \Gamma_b} \frac{ ((1+w^{-1})\lambda_a, \beta) }{\sinh^2 \bigl( \frac{\theta}{2} - \frac{i \pi}{2h} u_{\gamma_a \beta} \bigl)} \,.
\end{equation}
We remark that formula~\eqref{eq:tree_level_S_2} for the tree-level S-matrix was also obtained from standard Feynman diagrams computations in \cite{Dorey:2021hub}. 
In the remaining part of this section, we reproduce the one-loop S-matrix by plugging~\eqref{eq:tree_level_S_2} into~\eqref{eq:result_S_mat_eq_m_ren_3}.  We show that the expression generated in this way matches exactly~\eqref{eq:one_loop_SIV}.
For the derivation, we will use the properties of the Coxeter geometry previously discussed and some additional identities which we prove in Appendix~\ref{app:orbit_relation}.

\subsection{Derivation of one-loop S-matrices}

In all simply-laced affine Toda theories, the following properties apply
\begin{equation}
\begin{split}
&\frac{M^{(0)}_{aa}(0)}{m^2_a}= -i \frac{4 \g^2}{h} \quad \forall \ a \in \{1,\, \dots, \, r\} \,,\\
&\ai= - i \frac{\g^2}{m^2 h}  \,,
\end{split}
\end{equation}
whose proofs follow from the tree-level S-matrix~\eqref{eq:tree_level_S_2}, as shown in~\cite{Dorey:2021hub}. Moreover, since we set the root length equal to $\sqrt{2}$, we also have 
\begin{equation}
\sum^r_{e=1} m_e^2=\text{Tr} \bigl(\mathbb{M}^2 \bigl) = m^2 \sum^r_{i=0} n_i \alpha^2_i = 2 m^2 h \,.
\end{equation}
Using these relations the last term in~\eqref{eq:result_S_mat_eq_m_ren_3} vanishes and choosing the integration contour $\Gamma_2$ we are left with the following formula for the one-loop S-matrix
\begin{equation}
\label{eq:result_S_mat_eq_m_ren_4}
\begin{split}
S^{(1)}_{ab}(\theta_{pp'})&=\frac{\left(S^{(0)}_{ab}(\theta_{p p'}) \right)^2}{2}+\frac{\g^2}{4 \pi}  \theta_{p p'} \frac{\partial}{\partial \theta_{p p'}} S^{(0)}_{ab}(\theta_{p p'}) \\
&-\frac{1}{8 \pi^2}  \sum_{e=1}^r \frac{\partial}{\partial \theta_{p p'}} \oint_{\ip_2} d\theta_k \theta_k S^{(0)}_{ea} (\theta_{k p} ) S^{(0)}_{e b} (\theta_{k p'}) \,.
\end{split}
\end{equation}

Let us focus on the computation of the integral 
\begin{equation}
I \equiv \sum_{e=1}^r  \oint_{\ip_2} d\theta_k \theta_k S^{(0)}_{ea} (\theta_{k p} ) S^{(0)}_{e b} (\theta_{k p'}) \,,
\end{equation}
which is the most difficult term in the formula.
Plugging the tree-level S-matrix \eqref{eq:tree_level_S_2} into this integral we obtain
\begin{multline}
\label{eq:int_Toda}
I=-\frac{\g^4}{16 h^2} \, \sum^r_{e=1} \sum_{\alpha \in \Gamma_a} \sum_{\beta \in \Gamma_b} (\gamma_e, \alpha) (\gamma_e, \beta) \oint_{\ip_2} d\theta_k \theta_k \coth \left( \frac{\theta_{kp}}{2} - \frac{i \pi}{2h} u_{\gamma_e \alpha} \right) \times \\
\coth \left( \frac{\theta_{kp'}}{2} - \frac{i \pi}{2h} u_{\gamma_e \beta} \right) \,.
\end{multline}
The integrand is a meromorphic function in the $\theta_k$ complex plane and the integral can be computed using Cauchy's theorem. However, some care is needed on which poles are inside the contour. So far the difference between the projected angles (in units of $\pi/h$) has been defined up to shifts of $2h$. This is not a problem since the integrand in~\eqref{eq:int_Toda} is invariant under shifting $u_{\gamma_e \alpha} \to u_{\gamma_e \alpha} \pm 2h$ or $u_{\gamma_e \beta} \to u_{\gamma_e \beta} \pm 2h$. However, when solving the integral it is important to be careful that all residues over which we sum correspond to angles enclosed in the integration contour and lie therefore in the strip $\text{Im}(\theta_k) \in [0 , 2 \pi)$. To perform the integration properly we choose the following convention for the angles.
\begin{enumerate}
    \item If $e \in \circ$ then $u_{\gamma_e}=0$ and the angles $u_\alpha$ and $u_\beta$ are chosen as follows:
    \begin{enumerate}
        \item If $a \in \circ$ then
        \begin{equation}
        \alpha=w^{-p} \gamma_a \quad , \quad u_{\alpha}=-2p
        \end{equation}
        with $p=0, \dots, h-1$.
        \item If $a \in \bullet$ then
        \begin{equation}
        \alpha=w^{-p} \gamma_a \quad , \quad u_{\alpha}=-1-2p
        \end{equation}
        with $p=0, \dots, h-1$.
    \end{enumerate}
\item If $e \in \bullet$ then $u_{\gamma_e}=-1$ and the angles $u_\alpha$ and $u_\beta$ are chosen as follows:
\begin{enumerate}
        \item If $a \in \circ$ then
        \begin{equation}
        \alpha=w^{-(p+1)} \gamma_a \quad , \quad u_{\alpha}=-2p-2
        \end{equation}
        with $p=0, \dots, h-1$.
        \item If $a \in \bullet$ then
        \begin{equation}
        \alpha=w^{-p} \gamma_a \quad , \quad u_{\alpha}=-1-2p
        \end{equation}
        with $p=0, \dots, h-1$.
    \end{enumerate}
\end{enumerate}
The same convention is applied to the angles obtained by projecting the roots $\beta \in \Gamma_b$. With this convention, all the angles $u_{\gamma_e \alpha}$ and $u_{\gamma_e \beta}$ are in the interval $[0, 2h)$ and each term in the double sum on $\alpha \in \Gamma_a$ and $\beta \in \Gamma_b$ contains two poles inside the integration contour $\ip_2$ located at $\theta_{kp}=i \pi u_{\gamma_e \alpha}/h$ and $\theta_{kp'}=i\pi u_{\gamma_e \beta}/h$. 
Using Cauchy's theorem we obtain
\begin{multline}
I=-\frac{\g^4 i \pi}{4 h^2} \sum^r_{e=1} \sum_{\alpha \in \Gamma_a} \sum_{\beta \in \Gamma_b} (\gamma_e, \alpha) (\gamma_e, \beta) \biggl[ \left(\theta_p + \frac{i \pi}{h} u_{\gamma_e \alpha}\right) \coth \left( \frac{\theta_{pp'}}{2} +\frac{i \pi}{2 h} (u_{\gamma_e \alpha}- u_{\gamma_e \beta}) \right)+ \\ 
\left(\theta_{p'} + \frac{i \pi}{h} u_{\gamma_e \beta}\right) \coth \left( \frac{\theta_{p' p}}{2} +\frac{i \pi}{2h} (u_{\gamma_e \beta}- u_{\gamma_e \alpha}) \right)  \biggl].
\end{multline}
Noting that $u_{\gamma_e \beta}- u_{\gamma_e \alpha}=u_{\alpha \beta}$, $u_{\gamma_e \alpha}- u_{\gamma_e \beta}=u_{\beta \alpha}$ and $\coth{(-x)}=-\coth{(x)}$ the expression above can be written as
\begin{equation}
\label{eq:IplusII}
\begin{split}
I=I_1+I_2
\end{split}
\end{equation}
where
\begin{subequations}
\label{eq:I_II_terms}
\begin{equation}
\label{eq:I_term}
I_1=-\frac{\g^4 i \pi}{4 h^2} \theta_{p p'} \sum^r_{e=1} \sum_{\alpha \in \Gamma_a} \sum_{\beta \in \Gamma_b} (\gamma_e, \alpha) (\gamma_e, \beta) \coth \left( \frac{\theta_{pp'}}{2} -\frac{i \pi}{2 h} u_{\alpha \beta} \right) \,,
\end{equation}
\begin{equation}
\label{eq:II_term}
I_2=- \frac{\pi^2 \g^4}{4 h^3} \sum^r_{e=1} \sum_{\alpha \in \Gamma_a} \sum_{\beta \in \Gamma_b} (\gamma_e, \alpha) (\gamma_e, \beta) u_{\alpha \beta}\coth \left( \frac{\theta_{pp'}}{2} -\frac{i \pi}{2h} u_{\alpha \beta} \right) \, .
\end{equation}
\end{subequations}
The term $I_1$ is simple to compute due to its invariance under shifts of $u_{\alpha\beta} \to u_{\alpha\beta} \pm 2h$, which is lost in $I_2$ due to the presence of $u_{\alpha \beta}$ outside the $\coth$ function. Since 
\begin{equation}
u_{w^p \alpha, w^p \beta}=u_{ \alpha, \beta} \ \text{mod} \ 2h \quad \text{and} \quad (\alpha, w^p \beta)= (w^{-p} \alpha, \beta) 
\,,
\end{equation}
then we can write
\begin{equation}
 \label{eq:I_term2}
 I_1=-\frac{\g^4 i \pi}{4 h^2} \theta_{p p'} \sum^r_{e=1}  \sum_{p=0}^{h-1} \ \sum_{q=0}^{h-1} (w^{-p} \gamma_e,  \gamma_a) (w^{-p}\gamma_e, w^{q}\gamma_b) \coth \left( \frac{\theta_{pp'}}{2} -\frac{i \pi}{2 h} u_{ \gamma_a , w^{q} \gamma_b} \right) \,.
 \end{equation}
The sum on $e$ and $p$ can be performed using the completeness relation in~\eqref{eq:completeness-relation}, leading to
\begin{equation}
\label{eq:final_expression_for_I1}
\begin{split}
 I_1=2 \g^2 \pi\theta_{p p'} S^{(0)}_{ab}(\theta_{p p'}) \,.
 \end{split}
\end{equation}

The computation of $I_2$ is more convoluted due to the presence of $u_{\alpha \beta}$ outside the $\coth$ function and depends on what set ($\circ$ or $\bullet$) the scattered particles $a$ and $b$ belong to, according with the conventions for the angles previously discussed.
Let us consider the case $a,b\in\bullet$, for which we have
\begin{equation}
I_2=- \frac{\pi^2 \g^4}{2 h^3} \sum^r_{e=1} \sum_{p=0}^{h-1} \sum_{q=0}^{h-1} (\gamma_e, w^{-p} \gamma_a) (\gamma_e, w^{-q} \gamma_b) (q-p)\coth \left( \frac{\theta_{pp'}}{2} -\frac{i \pi}{h} (q-p) \right) \, .
\end{equation}
By defining $l\equiv q-p$ we can write this quantity as
\begin{equation}
\label{eq:I2-1}
\begin{split}
I_2=&- \frac{\pi^2 \g^4}{2 h^3} \sum^r_{e=1} \sum_{q=0}^{h-1} \sum_{l=q+1-h}^{q} (\gamma_e, w^{(l-q)} \gamma_a) (\gamma_e, w^{-q} \gamma_b) l \coth \left( \frac{\theta_{pp'}}{2} -\frac{i \pi}{h} l \right)\\
=&- \frac{\pi^2 \g^4}{2 h^3} \sum^r_{e=1} \sum_{q=0}^{h-1} \left(\sum_{l=q+1-h}^{0} + \sum_{l=1}^{q} \right) (\gamma_e, w^{(l-q)} \gamma_a) (\gamma_e, w^{-q} \gamma_b) l \coth \left( \frac{\theta_{pp'}}{2} -\frac{i \pi}{h} l \right)\\
=&- \frac{\pi^2 \g^4}{2 h^3} \sum^r_{e=1} \sum_{q=0}^{h-1} \sum_{l=1}^{h}  (\gamma_e, w^{(l-q)} \gamma_a) (\gamma_e, w^{-q} \gamma_b) l \coth \left( \frac{\theta_{pp'}}{2} -\frac{i \pi}{h} l \right)\\
&+ \frac{\pi^2 \g^4}{2 h^2} \sum^r_{e=1} \sum_{q=0}^{h-1} \sum_{l=q+1}^{h}  (\gamma_e, w^{(l-q)} \gamma_a) (\gamma_e, w^{-q} \gamma_b) \coth \left( \frac{\theta_{pp'}}{2} -\frac{i \pi}{h} l \right),
\end{split}
\end{equation}
where in the last equality we used the periodicity of the Coxeter element and $\coth$ functions under shifting $l$ by $h$. Exchanging the sums on $q$ and $l$, which is
\begin{equation}
\sum_{q=0}^{h-1} \sum_{l=q+1}^{h}= \sum_{l=1}^h \sum_{q=0}^{l-1} \,,
\end{equation}
and using the completeness relation~\eqref{eq:completeness-relation} we end up with 
\begin{equation}
\label{eq:I2blackblack1}
I_2=\frac{\pi^2 \g^4}{2 h^2} \sum_{l=1}^{h} \left[ -2 l (\gamma_b, w^{l} \gamma_a) + \sum^r_{e=1} \sum^{l-1}_{q=0} (\gamma_e, w^{(l-q)} \gamma_a) (\gamma_e, w^{-q} \gamma_b) \right] \coth \left(\frac{\theta_{pp'}}{2} -\frac{i \pi}{h} l \right)\ .
\end{equation}
Using the identity~\eqref{eq:rel_white_white} this expression is equivalent to
\begin{equation}
\label{eq:I2blackblack2}
I_2=\frac{\pi^2 \g^4}{2 h^2} \sum_{\beta\in\Gamma_b} ((1+w^{-1})\lambda_a , \beta) \coth \left(\frac{\theta_{pp'}}{2} -\frac{i \pi}{2h} u_{\gamma_a \beta} \right).
\end{equation}

The case $\{a \in \circ, \,b \in \bullet\}$ can be similarly addressed and we end up with the following expression
\begin{multline}
    I_2=\frac{\pi^2 \g^4}{2 h^2} \sum_{l=1}^h \Bigl[-(2l+1) (\gamma_b, w^{l} \gamma_a) + \sum^r_{e=1} \sum^{l-1}_{q=0} (\gamma_e, w^{l-q} \gamma_a) (\gamma_e, w^{-q} \gamma_b) \\
    +\sum_{e \in \bullet} (\gamma_e, \gamma_a) (\gamma_e, w^{-l} \gamma_b) \Bigl] \coth \left( \frac{\theta_{pp'}}{2} -\frac{i \pi}{2h} (2l+1) \right) \, ,
\end{multline}
which after using the identity \eqref{eq:rel_white_black} is also equivalent to~\eqref{eq:I2blackblack2}. 
As aforementioned the choice of Steinberg's ordering, i.e., setting which roots are black and white is arbitrary, therefore the cases $\{a \in \circ$, $b \in \circ\}$ and $\{a \in \bullet$, $b \in \circ\}$ follow immediately from the two cases dealt here. In conclusion, combining \eqref{eq:final_expression_for_I1} with \eqref{eq:I2blackblack2} we find that
\begin{equation}
\label{eq:final-eq-I}
I= 2 \g^2 \pi\theta_{p p'} S^{(0)}_{ab}(\theta_{p p'}) +  \frac{\pi^2 \g^4}{2 h^2} \sum_{\beta \in \Gamma_b} ((1+w^{-1})\lambda_a, \beta)  \coth \left( \frac{\theta}{2} - \frac{i \pi}{2h} u_{\gamma_a \beta} \right) \,.
\end{equation}
Inserting \eqref{eq:final-eq-I} into the one-loop expression \eqref{eq:result_S_mat_eq_m_ren_4} we observe that the terms proportional to $\theta_{p p'}$ cancel in the sum and we reproduce exactly its counterpart~\eqref{eq:one_loop_SIV} obtained from the bootstrap. Therefore for all simply-laced affine Toda models, we showed that the exact S-matrix can be universally reproduced from perturbation theory at one-loop, thus further validating the bootstraped formula.

\section{Conclusions}

In this paper, we derived expressions for one-loop S-matrices of 1+1 dimensional bosonic models described by (bare) Lagrangians of type~\eqref{eq0_1} and satisfying the tree-level pure elasticity Property~\ref{Condition_tree_level_elasticity_introduction}. 
For the one-loop analysis, we needed the renormalized Lagrangian \eqref{eq:renormalised_Lagrangian} and we set the appropriate renormalization conditions.  As noted in~\cite{Polvara:2023vnx} counterterms are necessary not only to avoid UV divergences but also to ensure the absence of particle production at one-loop, which leads to the establishment of the renormalization condition (3) in Section~\ref{sec:ren_condition_123} to fix them. 
By properly cutting Feynman diagrams into Dirac delta functions and retarded propagators, and keeping into account all counterterms, we have been able to write one-loop amplitudes in terms of tree-level quantities; these amplitudes are obtained by combining~\eqref{eq:definition_one_loop_amplitude}, \eqref{eq:final_result_amplitude} and \eqref{eq:countertermsII}.
While the coupling counterterms 
necessary to have absence of production processes at one-loop have not been determined for general theories, we have been able to universally find these counterterms for all models having mass ratios which do not renormalise at one-loop. For these theories, we showed that the sum of~\eqref{eq:definition_one_loop_amplitude}, \eqref{eq:final_result_amplitude} and \eqref{eq:countertermsII} simplifies and the one-loop S-matrices can be written through the closed formula~\eqref{eq:result_S_mat_eq_m_ren_2Int}.
This formula was then tested on some simple models having the nice properties of preserving the mass ratios, like the sinh-Gordon, Bullough-Dodd and $d^{(1)}_4$ affine Toda model; in all these cases the S-matrices obtained through our formula reproduced exactly the S-matrices obtained in the past through the bootstrap. As a more elaborated test, we used our formula to reproduce universal expressions for the one-loop S-matrices of simply-laced affine Toda models finding perfect match with the universal expressions for these S-matrices proposed in~\cite{Dorey:1990xa,Dorey:1991zp}. Our results generalise the perturbative computations for one-loop S-matrices performed in~\cite{Braden:1990qa,Braden:1991vz}, which were based on a case-by-case study of different simply-laced affine Toda theories.

We also showed how Landau singularities in one-loop S-matrices can be explained in terms of multiple pairs of tree-level bound state singularities trapping the integration path used in formula~\eqref{eq:result_S_mat_eq_m_ren_2Int}. 
Landau poles in the S-matrices of simply-laced affine Toda models can also be studied in terms of threshold Feynman diagrams~\cite{Braden:1990wx}; in~\cite{Dorey:2022fvs,Dorey:2023cuq} it was observed that the coefficient of a second-order singularity appearing in the Laurent expansion of the S-matrix around a general higher-order Landau pole is proportional to the number of singular networks of diagrams associated with the pole. It would be interesting to compare the number of trapping punctures discussed in this paper with the multiplicity of singular networks discussed in~\cite{Dorey:2022fvs,Dorey:2023cuq} and connect them with properties of the underlying root systems. 
We also remark that these higher-order Landau singularities are not captured by unitarity cut methods used in the past and would be interesting to revisit the formulas for one-loop S-matrices proposed in~\cite{Bianchi:2013nra,Bianchi:2014rfa} in light of this.

Another possible direction to explore is to use the formula for one-loop S-matrix obtained by combining \eqref{eq:definition_one_loop_amplitude}, \eqref{eq:final_result_amplitude} and \eqref{eq:countertermsII} for nonsimply-laced affine Toda models. For this class of theories, the mass ratios renormalise in a nontrivial way~\cite{Braden:1989bu, Christe:1989my} which is highly model dependent. It would then be interesting to find how to fix the coupling counterterms in terms of the mass ratios in such a way as to avoid one-loop production amplitudes.
It is possible that an in-depth analysis of these models using the method proposed in this work can lead to a further understanding of the quantum properties of these theories and their relation with the conjectured exact S-matrices advanced in~\cite{Delius:1991kt,Corrigan:1993xh}.

We stress that the results obtained in this paper apply to massive bosonic theories with polynomial-like interactions. It is then natural to ask whether it is possible to extend our study to Lagrangians containing different types of matter (such as fermionic or vector fields) and different types of interactions; a natural direction would be to consider derivative interactions in the potential. Famous examples of integrable theories of this type are provided by the class of generalised sine-Gordon models considered for example in~\cite{deVega:1981ka,Hoare:2010fb}. One-loop S-matrices for these models are correctly reproduced by the unitarity cuts formula advanced in~\cite{Bianchi:2014rfa}; if we were able to extend our study to these models we would probably obtain a more general formula for one-loop S-matrices of integrable theories that would capture the unitarity cut results and at the same time would be able to keep into account the intricated structure of bound states responsible for Landau poles discussed in this paper.
An interesting direction would also be the extension of the results obtained here to the analysis of integrable non-linear sigma models and their deformations, which have been receiving attention recently~\cite{Hoare:2021dix}. Remaining on the generalization front, another open problem is the study of massless models at one-loop. A place to start would be the tree-level analysis made in~\cite{Hoare:2018jim}; in this case, there are ambiguities in perturbation theory and curiously the tree-level S-matrices of integrable models exhibit production. The study of these massless theories could potentially solve some incompatibilities between the perturbative S-matrices for AdS$_3 \times$ S$^3 \times$ T$^4$ derived in~\cite{Roiban:2014cia, Sundin:2016gqe} and the ones obtained through the bootstrap in~\cite{Borsato:2013hoa, Frolov:2021fmj} since these conflicts are partially due to the presence of massless modes.

Finally, we should ask ourselves whether the on-shell methods used in this paper can be applied to multiple-loop S-matrices and if there is a way to generate these S-matrices in terms of tree-level data as we did at one-loop. It would be fascinating if for integrable theories it was possible to generate these higher-loop S-matrices recursively in terms of S-matrices with a lower number of loops. As seen here there are many open problems in the realm of perturbative integrability, whose study may shed light on what makes integrable theories special and how they feature in the vast space of quantum field theories.

\section*{Acknowledgments}

We thank Patrick Dorey, Ben Hoare, Anton Pribytok and Alessandro Sfondrini for related discussions. DP especially thanks Ben Hoare for useful discussions and the kind hospitality in Durham where this work was started. The authors also thank the participants of the workshop ``Integrability
in Low-Supersymmetry Theories'' in Filicudi, Italy, for a stimulating environment where
part of this work was carried out. 
The authors acknowledge support from the European Union – NextGenerationEU, from the program
STARS@UNIPD, under the project ``Exact-Holography – A
new exact approach to holography: harnessing the power
of string theory, conformal field theory and integrable
models'', also from the PRIN Project n. 2022ABPBEY,
``Understanding quantum field theory through its deformations'', and from the CARIPLO Foundation ``Supporto ai giovani talenti italiani nelle competizioni dell'European Research Council'' grant n. 2022-1886 ``Nuove basi per la teoria delle stringhe''.

\appendix

\section{Computing the off-shell limit of tree-level amplitudes}
\label{appendix_on_off_shell_limit_of_tree_level_amplitudes}

In this appendix, we derive equation~\eqref{V1ab_to_ab_final_result} starting from~\eqref{single_cut_relevant_contribution_u_channel_not_yet_expanded}. 
We start by substituting~\eqref{expansion_of_propagators_in_u_channel} into~\eqref{single_cut_relevant_contribution_u_channel_not_yet_expanded}:
\begin{equation}
\label{single_cut_relevant_contribution_u_channel_propagators_expanded}
\begin{split}
\hat{V}^{(1)}_{ab\to ab}=\frac{i}{16\pi} \sum_{e=1}^r \int_{-\infty}^{+\infty} \frac{d\theta_k}{k \cdot x} \ \Bigl( &-M^{(0)}_{ae \to a e}(p, k, q, r) \ M^{(0)}_{b e \to b e }(p', r, q', k)\\
&+M^{(0)}_{ae \to a e}(p, r', q, k) \  M^{(0)}_{be \to be }(p', k, q', r') \Bigr)\\
-\frac{i}{32 \pi} \sum_{e=1}^r \int_{-\infty}^{+\infty} \frac{d\theta_k \mu^2}{(k \cdot x)^2} \ \Bigl( &+M^{(0)}_{ae \to a e}(p, k, q, r) \ M^{(0)}_{b e\to b e }(p', r, q', k)\\
&+M^{(0)}_{ae \to a e}(p, r', q, k) \  M^{(0)}_{be \to be }(p', k, q', r') \Bigr) + O(\mu).
\end{split}
\end{equation}
At this point, we expand each tree-level amplitude on the r.h.s. of~\eqref{single_cut_relevant_contribution_u_channel_propagators_expanded} around $\mu=0$. 
In performing the expansion we assume $p$, $p'$ and $k$ fixed and on-shell and we move $q$, $q'$, $r$ and $r'$; these vectors can be written as functions of $\mu$ thanks to~(\eqref{regularization_through_x_introducing_outgoing_mass_deformations}, \eqref{r_and_rprime_written_in_terms_of_k_and_x}). The real parameter $\theta_x$ in~\eqref{definition_of_my_regulator_x} is fixed and spans a set of possible directions that can be followed to reach the elastic configuration in which incoming and outgoing momenta become equal and on-shell. 
We consider $M^{(0)}_{ae \to a e}(p, k, q, r)$ first, whose expansion is
\begin{equation}
\label{appendix_expansion_of_Mae_pk_qr_tree_amplitude}
\begin{split}
&M^{(0)}_{ae \to a e}(p, k, q, r)= M^{(0)}_{ae \to a e}(p, k, p, k)\\
&+ \frac{\partial}{\partial q^2} M^{(0)}_{ae \to a e}(p, k, q, r)\Bigl|_{q=p, r=k} (q^2 - m_a^2)+\frac{\partial}{\partial r^2} M^{(0)}_{ae \to a e}(p, k, q, r)\Bigl|_{q=p, r=k} (r^2 - m_e^2)\\
&+\frac{\partial}{\partial \theta_q} M^{(0)}_{ae \to a e}(p, k, q, r)\Bigl|_{q=p, r=k} (\theta_q - \theta_p)+\frac{\partial}{\partial \theta_r} M^{(0)}_{ae \to a e}(p, k, q, r)\Bigl|_{q=p, r=k} (\theta_r - \theta_k) \, .
\end{split}
\end{equation}
Plugging~\eqref{definition_of_my_regulator_x} into~(\eqref{regularization_through_x_introducing_outgoing_mass_deformations}, \eqref{r_and_rprime_written_in_terms_of_k_and_x}) it is straightforward to write $q^2$ and $r^2$ as functions of $\mu$; they are given by
\begin{subequations}
\label{qsquare_and_rsquare_as_functions_of_mu}
\begin{equation}
\label{qsquare_as_functions_of_mu}
q^2= p^2+ 2 p \cdot x +x^2= m_a^2 + 2 p \cdot x + O(\mu^2) \, ,
\end{equation}
\begin{equation}
\label{rsquare_as_functions_of_mu}
r^2= k^2 - 2 k \cdot x +x^2  = m_e^2- 2 k \cdot x + O(\mu^2) \, .
\end{equation}
\end{subequations}
The expansion for the rapidities is slightly more complicated but can be performed as well and leads to
\begin{subequations}
\label{thetaq_and_thetar_as_functions_of_mu}
\begin{equation}
\label{thetaq_as_functions_of_mu}
\theta_q= \theta_p - \frac{\mu}{m_a} \sinh{\theta_{px}} + O(\mu^2)\, ,
\end{equation}
\begin{equation}
\label{thetar_as_functions_of_mu}
\theta_r= \theta_k + \frac{\mu}{m_e} \sinh{\theta_{kx}} + O(\mu^2) \, ,
\end{equation}
\end{subequations}
where convention~\eqref{convention_on_difference_between_rapidities} for the difference between rapidities has been followed. 
Substituting~\eqref{qsquare_and_rsquare_as_functions_of_mu} and~\eqref{thetaq_and_thetar_as_functions_of_mu} into~\eqref{appendix_expansion_of_Mae_pk_qr_tree_amplitude} we obtain
\begin{equation}
\label{appendix_expansion_of_Mae_pk_qr_tree_amplitude_second_formulation}
\begin{split}
&M^{(0)}_{ae \to a e}(p, k, q, r)= M^{(0)}_{ae \to a e}(p, k, p, k)\\
&+ \frac{\partial}{\partial (p^{\text{out}})^2} M^{(0)}_{ae \to a e}(p, k, p, k) 2 p \cdot x-\frac{\partial}{\partial (k^{\text{out}})^2} M^{(0)}_{ae \to a e}(p, k, p, k) 2 k \cdot x\\
&-\frac{\partial}{\partial \theta^{\text{out}}_p} M^{(0)}_{ae \to a e}(p, k, p, k) \frac{\mu}{m_a} \sinh{\theta_{px}}+\frac{\partial}{\partial \theta^{\text{out}}_k} M^{(0)}_{ae \to a e}(p, k, p, k) \frac{\mu}{m_e} \sinh{\theta_{kx}}\\
\end{split}
\end{equation}
A superscript word, `in' or `out', has been added to indicate
if the derivative is performed with respect to the rapidity or the momentum squared of an incoming or outgoing particle. 
These derivatives are performed as if incoming and outgoing momenta were independent; indeed the overall energy-momentum conservation is taken into account by the fact that $q=p+x$ and $r=k-x$. The amplitude on the l.h.s. of~\eqref{appendix_expansion_of_Mae_pk_qr_tree_amplitude_second_formulation} should then be thought of as a function of $4$ momenta, each containing a certain $\mu$ dependence. 
Note that there is no unique manner to write a Feynman diagram in terms of the external momenta and the amplitude on the l.h.s. of~\eqref{appendix_expansion_of_Mae_pk_qr_tree_amplitude_second_formulation} can be written in terms of $p$, $k$, $q$ and $r$ in several different ways. However, once we require that $q=p+x$ and $r=k-x$, the expansion on the r.h.s. of~\eqref{appendix_expansion_of_Mae_pk_qr_tree_amplitude_second_formulation} does not depend on the way adopted to write the amplitude as a function of $p$, $k$, $q$ and $r$.

Similarly, the other tree-level amplitudes entering the integrand of~\eqref{single_cut_relevant_contribution_u_channel_propagators_expanded} can be expanded as
\begin{equation}
\label{appendix_expansion_of_Mae_prprime_qk_tree_amplitude}
\begin{split}
&M^{(0)}_{ae \to a e}(p, r', q, k)= M^{(0)}_{ae \to a e}(p, k, p, k)\\
&+\frac{\partial}{\partial (k^{\text{in}})^2} M^{(0)}_{ae \to a e}(p, k, p, k) 2 k \cdot x+ \frac{\partial}{\partial (p^{\text{out}})^2} M^{(0)}_{ae \to a e}(p, k, p, k) 2 p \cdot x\\
&-\frac{\partial}{\partial \theta^{\text{in}}_k} M^{(0)}_{ae \to a e}(p, k, p, k) \frac{\mu}{m_e} \sinh{\theta_{kx}}-\frac{\partial}{\partial \theta^{\text{out}}_p} M^{(0)}_{ae \to a e}(p, k, p, k) \frac{\mu}{m_a} \sinh{\theta_{px}} \, ,
\end{split}
\end{equation}
\begin{equation}
\label{appendix_expansion_of_Mbe_pprimer_qprimek_tree_amplitude}
\begin{split}
&M^{(0)}_{be \to b e}(p', r, q', k)= M^{(0)}_{be \to b e}(p', k, p', k)\\
&- \frac{\partial}{\partial (k^{\text{in}})^2} M^{(0)}_{be \to b e}(p', k, p', k) 2 k \cdot x-\frac{\partial}{\partial (p'^{\text{out}})^2} M^{(0)}_{be \to b e}(p', k, p', k) 2 p' \cdot x\\
&+\frac{\partial}{\partial \theta^{\text{in}}_k} M^{(0)}_{be \to b e}(p', k, p', k) \frac{\mu}{m_e} \sinh{\theta_{kx}}+\frac{\partial}{\partial \theta^{\text{out}}_{p'}} M^{(0)}_{be \to b e}(p', k, p', k) \frac{\mu}{m_b} \sinh{\theta_{p' x}}
\end{split}
\end{equation}
and
\begin{equation}
\label{appendix_expansion_of_Mbe_pprimek_qprimerprime_tree_amplitude}
\begin{split}
&M^{(0)}_{be \to b e}(p', k, q', r')= M^{(0)}_{be \to b e}(p', k, p', k)\\
&-\frac{\partial}{\partial (p'^{\text{out}})^2} M^{(0)}_{be \to b e}(p', k, p', k) 2 p' \cdot x+ \frac{\partial}{\partial (k^{\text{out}})^2} M^{(0)}_{be \to b e}(p', k, p', k) 2 k \cdot x\\
&+\frac{\partial}{\partial \theta^{\text{out}}_{p'}} M^{(0)}_{be \to b e}(p', k, p', k) \frac{\mu}{m_b} \sinh{\theta_{p' x}}-\frac{\partial}{\partial \theta^{\text{out}}_k} M^{(0)}_{be \to b e}(p', k, p', k) \frac{\mu}{m_e} \sinh{\theta_{kx}} \, .
\end{split}
\end{equation}
Substituting~\eqref{appendix_expansion_of_Mae_pk_qr_tree_amplitude_second_formulation}, \eqref{appendix_expansion_of_Mae_prprime_qk_tree_amplitude}, \eqref{appendix_expansion_of_Mbe_pprimer_qprimek_tree_amplitude} and~\eqref{appendix_expansion_of_Mbe_pprimek_qprimerprime_tree_amplitude} into~\eqref{single_cut_relevant_contribution_u_channel_propagators_expanded}, after some simplifications, we obtain
\begin{equation}
\label{Vab1_appendix_intermediate_step}
\begin{split}
&\hat{V}^{(1)}_{ab\to ab}=\\
&+\frac{i}{8\pi} \sum_{e=1}^r \int_{-\infty}^{+\infty} d\theta_k \  \frac{\partial}{\partial k^2} \Bigl( M^{(0)}_{ae \to a e}(p, k, p, k) M^{(0)}_{be \to b e}(p', k, p', k) \Bigl)\Bigl|_{k^2=m^2_e}\\
&-\frac{i}{16\pi} \sum_{e=1}^r  \int_{-\infty}^{+\infty} \frac{d\theta_k}{k \cdot x} \frac{\mu}{m_e} \sinh{\theta_{kx}} \ \frac{\partial}{\partial \theta_k} \Bigl( M^{(0)}_{ae \to a e}(p, k, p, k) M^{(0)}_{be \to b e}(p', k, p', k) \Bigl)\\
&-\frac{i}{16 \pi} \sum_{e=1}^r \int_{-\infty}^{+\infty} \frac{d\theta_k \mu^2}{(k \cdot x)^2} \ M^{(0)}_{ae \to a e}(p, k, p, k) \ M^{(0)}_{b e\to b e }(p', k, p', k) + O(\mu) \, ,
\end{split}
\end{equation}
where we define
$$
\frac{\partial}{ \partial k^2} \equiv \frac{\partial}{ \partial (k^{\text{in}})^2}+\frac{\partial}{ \partial (k^{\text{out}})^2} \hspace{5mm} \text{and} \hspace{5mm} \frac{\partial}{ \partial \theta_k} \equiv \frac{\partial}{ \partial \theta_k^{\text{in}}}+\frac{\partial}{ \partial \theta_k^{\text{out}}} \, .
$$
Finally, using the fact that
$$
k \cdot x = \mu m_e \cosh{\theta_{kx}}
$$
and integrating the second row of~\eqref{Vab1_appendix_intermediate_step} by parts we obtain~\eqref{V1ab_to_ab_final_result}.

\section{The on-shell limit of 3-to-3 tree-level processes}
\label{app:on_shell_limit}
In this appendix, we describe tree-level processes of the form in~\eqref{3_to_3_tree_level_process}. First, we compute 
the on-shell limit of the sum of Feynman diagrams in Figure~\ref{Factorization_contributions_in_tree_level_6_point_processes}.
Despite these diagrams being singular at the values of momenta $\{q=p, q'=p', \tilde{k}=k\}$, we show that reaching this singular configuration 
keeping all the particles on-shell and the overall energy-momentum conserved the singularity disappears and the sum of these diagrams is finite.
In performing this computation we prove that Property~\ref{Property_4_point_couplings} is a consequence of Property~\ref{Condition_tree_level_elasticity_introduction}. Then, we plug the obtained result into~\eqref{definition_of_Vab_to_ab_term_i_123_onshell_limit} and compute in this way $\hat{V}^{(\text{on})}_{ab}$.

\subsection{Summing singular diagrams}

Let us define
\begin{equation}
\label{app:split_of_V_3_to_3_on}
\hat{V}^{(\text{on})}_{abe}= \hat{V}^{(1, \text{on})}_{abe}+\hat{V}^{(2, \text{on})}_{abe}+\hat{V}^{(3, \text{on})}_{abe}
\end{equation}
to be the sum of the three combinations of Feynman diagrams in Figure~\ref{Factorization_contributions_in_tree_level_6_point_processes} obtained by reaching the point $\{q=p, q'=p', \tilde{k}=k\}$ keeping all the particles on-shell. We perform the limit on the collection of diagrams number $(2)$ in Figure~\ref{Factorization_contributions_in_tree_level_6_point_processes}, contributing to $\hat{V}^{(2, \text{on})}_{abc}$; the collections numbers $(1)$ and $(3)$, contributing to $\hat{V}^{(1, \text{on})}_{abc}$ and $\hat{V}^{(3, \text{on})}_{abc}$ respectively, can be computed similarly. 

We label  by $p_a$ the momentum of the propagator on the l.h.s. and by $\tilde{p}_a$ the momentum of the propagator on the r.h.s. of picture $(2)$, given by
\begin{equation}
p_a=p+k - \tilde{k} \hspace{4mm} \text{and} \hspace{4mm} \tilde{p}_a=p+p'-q'
\end{equation}
respectively.
We keep all the external momenta on-shell with $\theta_p$, $\theta_{p'}$ and $\theta_k$ fixed; then due to the overall energy-momentum conservation the rapidities $\theta_{\tilde{k}}$, $\theta_{q}$ and $\theta_{q'}$ depend on a single degree of freedom. We assume this degree of freedom to be $p_a^2$ (the momentum squared of the intermediate propagator) for the collection of diagrams generating the l.h.s. of picture number $(2)$ and $\tilde{p}_a^2$ for the collection of diagrams generating the r.h.s. of the picture.
We will consider the elastic branch of the kinematics on which $\{ \theta_{q}=\theta_p \,, \theta_{q'}=\theta_{p'} \,, \theta_{\tilde{k}}=\theta_k\}$ when $p_a^2=p^2$ (or equivalently $\tilde{p}_a^2=p^2$) \,.

We evaluate the l.h.s. part of picture number $(2)$ in Figure~\ref{Factorization_contributions_in_tree_level_6_point_processes} first. 
For fixed $p$ and $k$, we can imagine the blob on the bottom left of picture number $(2)$ as a function of $p_a^2$ alone. 
Indeed, once we fix the value of $p_a^2$, the rapidities $\theta_a$ and $\theta_{\tilde{k}}$ are determined by the energy-momentum conservation applied to this blob.
The expansion of this blob around the elastic point can therefore be written as
\begin{equation}
\label{On_shell_limit_picture_2_LHS_bottom_blob}
M^{(0)}_{ae \to ae} (p, k , p_a , \tilde{k})=M^{(0)}_{ae \to ae} (p, k , p , k) + \frac{\partial}{\partial p_a^2}M^{(0)}_{ae \to ae} (p, k , p_a , \tilde{k})\Bigl|_{p_a^2=m_a^2} (p_a^2 - m_a^2) + \dots \, ,
\end{equation}
where $\theta_a$ and $\theta_{\tilde{k}}$ implicitly depend on $p_a^2$ when the derivative is performed. The ellipses contain subleading contributions in $p_a^2 - m_a^2$.

On the other hand, ignoring the existence of the energy-momentum conservation on the blob just described, the blob on the top left of picture number $(2)$ in Figure~\ref{Factorization_contributions_in_tree_level_6_point_processes} should be thought as a function of two independent variables: $p_a^2$ and $\theta_a$. Indeed $p'$ is fixed and $\theta_q$ and $\theta_{q'}$ are unambiguously determined by the energy-momentum conservation applied to this blob only if $p_a^2$ and $\theta_a$ are both specified. This blob can be expanded as
\begin{multline}
\label{On_shell_limit_picture_2_LHS_top_blob}
    M^{(0)}_{ab \to ab} (p_a, p' , q , q')=M^{(0)}_{ab \to ab} (p, p' , p , p') + \frac{\partial}{\partial p_a^2}M^{(0)}_{ab \to ab} (p_a, p' , q , q')\Bigl|_{\substack{ p_a^2=m_a^2 \\ \theta_a= \theta_p}} (p_a^2 - m_a^2)\\
    + \frac{\partial}{\partial \theta_a}M^{(0)}_{ab \to ab} (p_a, p' , q , q')\Bigl|_{\substack{ p_a^2=m_a^2 \\ \theta_a= \theta_p}} (\theta_a - \theta_p) + \dots \, .
\end{multline}
The ellipses contain higher order powers in $(\theta_a - \theta_p)$ and $(p_a^2 - m_a^2)$.
Applying the energy-momentum conservation on the blob on the bottom (having as external momenta $p$, $k$, $\tilde{k}$ and $p_a$) we can write the rapidity $\theta_a$ of the particle propagating in the middle of the two blobs in terms of $p_a^2$:
$$
\theta_a=\theta_p+ \frac{1}{2 m_a^2 \tanh{\theta_{kp}}} (p_a^2 - m_a^2) +O((p_a^2 - m_a^2)^2) \, .
$$
Substituting this expression in~\eqref{On_shell_limit_picture_2_LHS_top_blob}, the on-shell limit of the l.h.s. of picture $(2)$ in Figure~\ref{Factorization_contributions_in_tree_level_6_point_processes} can be written as
\begin{equation}
\label{On_shell_limit_picture_2_LHS}
\begin{split}
&V^{(2, \, \text{on}, \, \text{l.h.s.})}_{abe\to abe}= \Bigl(M^{(0)}_{ae \to ae} (p, k , p , k) + \frac{\partial}{\partial p_a^2}M^{(0)}_{ae \to ae} (p, k , p_a , \tilde{k})\Bigl|_{p_a^2=m_a^2} (p_a^2 - m_a^2) \Bigl) \frac{i}{p_a^2 - m^2_a} \\
& \times \Bigl(M^{(0)}_{ab \to ab} (p, p' , p , p') + \frac{\partial}{\partial p_a^2}M^{(0)}_{ab \to ab} (p_a, p' , q , q')\Bigl|_{\substack{ p_a^2=m_a^2 \\ \theta_a= \theta_p}} (p_a^2 - m_a^2)\\
&+ \frac{\partial}{\partial \theta_a}M^{(0)}_{ab \to ab} (p_a, p' , q , q')\Bigl|_{\substack{ p_a^2=m_a^2 \\ \theta_a= \theta_p}} \frac{1}{2 m_a^2 \tanh{\theta_{kp}}} (p_a^2 - m_a^2) \Bigl) +O(p_a^2 - m_a^2) \,,
\end{split}
\end{equation}
which after having been expanded becomes
\begin{equation}
\label{On_shell_limit_picture_2_LHS_final_result}
\begin{split}
V^{(2, \, \text{on}, \, \text{l.h.s.})}_{abe\to abe}&= M^{(0)}_{ae \to ae} (p, k , p , k) \frac{i}{p_a^2 -m^2_a} M^{(0)}_{ab \to ab} (p, p' , p , p')\\
&+ i M^{(0)}_{ab \to ab} (p, p' , p , p') \frac{\partial}{\partial p_a^2}M^{(0)}_{ae \to ae} (p, k , p_a , \tilde{k})\Bigl|_{p_a^2=m_a^2}\\
&+iM^{(0)}_{ae \to ae} (p, k , p , k) \frac{\partial}{\partial p_a^2}M^{(0)}_{ab \to ab} (p_a, p' , q , q')\Bigl|_{\substack{ p_a^2=m^2_a \\ \theta_a= \theta_p}}\\
&+\frac{i M^{(0)}_{ae \to ae} (p, k , p , k)}{2 m_a^2 \tanh{\theta_{kp}}} \frac{\partial}{\partial \theta_a}M^{(0)}_{ab \to ab} (p_a, p' , q , q')\Bigl|_{\substack{ p_a^2=m^2_a \\ \theta_a= \theta_p}} \, .
\end{split}
\end{equation}
A similar computation can be performed on the r.h.s. of picture number $(2)$ in Figure~\ref{Factorization_contributions_in_tree_level_6_point_processes} and leads to
\begin{equation}
\label{On_shell_limit_picture_2_RHS_final_result}
\begin{split}
V^{(2, \, \text{on}, \, \text{r.h.s.})}_{abe\to abe}&= M^{(0)}_{ae \to ae} (p, k , p , k) \frac{i}{\tilde{p}_a^2 -m^2_a} M^{(0)}_{ab \to ab} (p, p' , p , p')\\
&+ i M^{(0)}_{ae \to ae} (p, k , p , k) \frac{\partial}{\partial \tilde{p}_a^2}M^{(0)}_{ab \to ab} (p, p' , \tilde{p}_a , q')\Bigl|_{\tilde{p}_a^2=m_a^2}\\
&+iM^{(0)}_{ab \to ab} (p, p' , p , p') \frac{\partial}{\partial \tilde{p}_a^2}M^{(0)}_{ae \to ae} (\tilde{p}_a, k , q , \tilde{k})\Bigl|_{\substack{ \tilde{p}_a^2=m_a^2 \\ \tilde{\theta}_a= \theta_p}}\\
&+\frac{iM^{(0)}_{ab \to ab} (p, p' , p , p')}{2 m_a^2 \tanh{\theta_{p' p}}} \frac{\partial}{\partial \tilde{\theta}_a}M^{(0)}_{ae \to ae} (\tilde{p}_a, k , q , \tilde{k})\Bigl|_{\substack{ \tilde{p}_a^2=m_a^2 \\ \tilde{\theta}_a= \theta_p}} \,.
\end{split}
\end{equation}

To evaluate the sum of~\eqref{On_shell_limit_picture_2_LHS_final_result} and~\eqref{On_shell_limit_picture_2_RHS_final_result} we need to figure out what are the values of the propagators $\frac{i}{p_a^2 -m^2_a}$ and $\frac{i}{\tilde{p}_a^2 -m^2_a}$ at the pole $\{ \theta_{q}=\theta_p \,, \theta_{q'}=\theta_{p'} \,, \theta_{\tilde{k}}=\theta_k\}$. 
We can expand both propagators as functions of $\theta_q$ around the point $\theta_q=\theta_p$ at fixed values of $\theta_{p}$, $\theta_{p'}$ and $\theta_{k}$; as already discussed $\theta_{q'}$ and $\theta_{\tilde{k}}$ are function of $\theta_q$ thanks to the overall energy-momentum conservation. Performing the computation we obtain that the singular parts of both propagators cancel in the sum and the result is given by
\begin{equation}
\label{sum_of_propagators_onshell_limit_singularities_cancellation}
\frac{i}{\tilde{p}_a^2 -m^2_a}+\frac{i}{p_a^2 -m^2_a} = \frac{i}{2 m_a^2} \frac{\cosh{\theta_{p' k}}}{\sinh{\theta_{p p'}} \sinh{\theta_{p k}}} \, .
\end{equation}
Exploiting this fact the sum of~\eqref{On_shell_limit_picture_2_LHS_final_result} and~\eqref{On_shell_limit_picture_2_RHS_final_result} is given by
\begin{equation}
\label{Vabe_to_abe_on_shell_1}
\begin{split}
&V^{(2, \, \text{on})}_{abe\to abe}=\frac{i}{2 m_a^2} \frac{\cosh{\theta_{p' k}}}{\sinh{\theta_{p p'}} \sinh{\theta_{p k}}} M^{(0)}_{ae \to ae} (p, k , p , k) M^{(0)}_{ab \to ab} (p, p' , p , p')\\
&+ i \frac{\partial}{\partial p^2} \Bigl( M^{(0)}_{ae \to ae} (p, k , p , k) M^{(0)}_{ab \to ab} (p, p' , p , p') \Bigl)\Bigl|_{p^2=m^2_a}\\
&+\frac{i M^{(0)}_{ae \to ae} (p, k , p , k)}{2 m_a^2 \tanh{\theta_{kp}}} \frac{\partial}{\partial \theta_p}M^{(0)}_{ab \to ab} (p, p' , p , p')+\frac{iM^{(0)}_{ab \to ab} (p, p' , p , p')}{2 m_a^2 \tanh{\theta_{p' p}}} \frac{\partial}{\partial \theta_p}M^{(0)}_{ae \to ae} (p, k , p , k) \,.
\end{split}
\end{equation}
Repeating similar computations for the collections of Feynman diagrams contributing to pictures number $(1)$ and $(3)$ in Figure~\ref{Factorization_contributions_in_tree_level_6_point_processes} we obtain
\begin{equation}
\label{Vabe_to_abe_on_shell_2}
\begin{split}
&V^{(\text{on})}_{abe\to abe}= {\color{red} \frac{i}{2 m_e^2} \frac{\cosh{\theta_{p p'}}}{\sinh{\theta_{k p}} \sinh{\theta_{k p'}}} M^{(0)}_{ae \to ae} (p, k , p , k) M^{(0)}_{b e \to b e} (p', k , p' , k)}\\
&{\color{red}+\frac{i}{2 m_a^2} \frac{\cosh{\theta_{p' k}}}{\sinh{\theta_{p p'}} \sinh{\theta_{p k}}} M^{(0)}_{ae \to ae} (p, k , p , k) M^{(0)}_{ab \to ab} (p, p' , p , p')}\\
&{\color{red}+\frac{i}{2 m_b^2} \frac{\cosh{\theta_{p k}}}{\sinh{\theta_{p' p}} \sinh{\theta_{p' k}}} M^{(0)}_{ab \to ab} (p, p' , p , p') M^{(0)}_{b e \to b e} (p', k , p' , k) }\\
&{\color{blue}+ i \frac{\partial}{\partial k^2} \Bigl( M^{(0)}_{ae \to ae} (p, k , p , k) M^{(0)}_{be \to be} (p', k , p' , k) \Bigl)\Bigl|_{k^2=m^2_e}}\\
&{\color{blue}+ i \frac{\partial}{\partial p^2} \Bigl( M^{(0)}_{ae \to ae} (p, k , p , k) M^{(0)}_{ab \to ab} (p, p' , p , p') \Bigl)\Bigl|_{p^2=m^2_a}}\\
&{\color{blue}+ i \frac{\partial}{\partial p'^2} \Bigl( M^{(0)}_{ab \to ab} (p, p' , p , p') M^{(0)}_{be \to be} (p', k , p' , k) \Bigl)\Bigl|_{p'^2=m^2_b}}\\
&{\color{orange}+\frac{i M^{(0)}_{ae \to ae} (p, k , p , k)}{2 m_e^2 \tanh{\theta_{pk}}} \frac{\partial}{\partial \theta_k}M^{(0)}_{be \to be} (p', k , p' , k)+\frac{iM^{(0)}_{be \to be} (p', k , p' , k)}{2 m_e^2 \tanh{\theta_{p' k}}} \frac{\partial}{\partial \theta_k}M^{(0)}_{ae \to ae} (p, k , p , k)}\\
&{\color{orange}+\frac{i M^{(0)}_{ae \to ae} (p, k , p , k)}{2 m_a^2 \tanh{\theta_{kp}}} \frac{\partial}{\partial \theta_p}M^{(0)}_{ab \to ab} (p, p' , p , p')+\frac{iM^{(0)}_{ab \to ab} (p, p' , p , p')}{2 m_a^2 \tanh{\theta_{p' p}}} \frac{\partial}{\partial \theta_p}M^{(0)}_{ae \to ae} (p, k , p , k)}\\
&{\color{orange}+\frac{i M^{(0)}_{ab \to ab} (p, p' , p , p')}{2 m_b^2 \tanh{\theta_{pp'}}} \frac{\partial}{\partial \theta_{p'}}M^{(0)}_{b e \to b e} (p' , k, p', k)+\frac{iM^{(0)}_{b e \to b e} (p', k , p' , k)}{2 m_b^2 \tanh{\theta_{k p'}}} \frac{\partial}{\partial \theta_{p'}}M^{(0)}_{ab \to ab} (p, p' , p , p') \,.}
\end{split}
\end{equation}
We split the result above into three different blocks identified by three different colours. This will be useful in the next sections as we will study these blocks in more detail.
In integrable theories, this result has then to be cancelled by all the remaining non-singular diagrams.

\subsection{Large rapidity limit}

In order for the expression in~\eqref{Vabe_to_abe_on_shell_2} to be cancelled by the sum of all the remaining Feynman diagrams we need to require 
\begin{equation}
\label{requirement_equal_limit_PlusMinusInfinite}
\lim_{\theta_k \to +\infty} V^{(\text{on})}_{abe\to abe}=\lim_{\theta_k \to -\infty} V^{(\text{on})}_{abe\to abe} \,.
\end{equation}
Note indeed that the remaining non-singular diagrams have equal values in the two limits $\theta_k \to + \infty$ and $\theta_k \to -\infty$ (this is a consequence of the fact that the Lagrangian potential is of polynomial type).
Examples of these diagrams are provided in Figure~\ref{Example_of_diagrams_in_large_rapidity_limit} where two diagrams are shown, the first surviving and the second cancelling in the limit of large $\theta_k$.
\begin{figure}
\begin{center}
\begin{tikzpicture}
\tikzmath{\y=1;}


\filldraw[black] (-4.4*\y,0.2*\y)  node[anchor=west] {\normalsize{$\lim_{\theta_k \to \pm \infty}$}};

\draw[] (-2.2*\y,2*\y) arc(165:195:8*\y);
\draw[] (+2.4*\y,2*\y) arc(15:-15:8*\y);

\draw[directed] (0*\y,1*\y) -- (0*\y,0*\y);
\draw[directed] (-0.7*\y,1.7*\y) -- (0*\y,1*\y);
\draw[directed] (0*\y,1*\y) -- (+0.7*\y,1.7*\y);
\draw[directed] (0.866025*\y,-0.5*\y) -- (0*\y,0*\y);
\draw[directed] (1.82224*\y,-0.243782*\y) -- (0.866025*\y,-0.5*\y);
\draw[directed] (0.866025*\y,-0.5*\y) -- (1.12224*\y,-1.45622*\y);
\draw[directed] (-0.866025*\y,-0.5*\y) -- (0*\y,0*\y);
\draw[directed] (-1.82224*\y,-0.243782*\y) -- (-0.866025*\y,-0.5*\y);
\draw[directed] (-0.866025*\y,-0.5*\y) -- (-1.12224*\y,-1.45622*\y);

\filldraw[black] (-1*\y,2*\y)  node[anchor=west] {\tiny{$e(k)$}};
\filldraw[black] (0.4*\y,2*\y)  node[anchor=west] {\tiny{$e(k)$}};
\filldraw[black] (-2.4*\y,-0.4*\y)  node[anchor=west] {\tiny{$a(p)$}};
\filldraw[black] (-1.5*\y,-1.7*\y)  node[anchor=west] {\tiny{$b(p')$}};
\filldraw[black] (1.8*\y,-0.4*\y)  node[anchor=west] {\tiny{$b(p')$}};
\filldraw[black] (0.9*\y,-1.7*\y)  node[anchor=west] {\tiny{$a(p)$}};

\filldraw[black] (0.1*\y,0.6*\y)  node[anchor=west] {\tiny{$i$}};
\filldraw[black] (0.4*\y,0*\y)  node[anchor=west] {\tiny{$j$}};
\filldraw[black] (-0.8*\y,0*\y)  node[anchor=west] {\tiny{$n$}};

\filldraw[black] (2.8*\y,0.2*\y)  node[anchor=west] {\normalsize{$= \frac{(-i C^{(3)}_{e \bar{e} \bar{i}})}{0-m^2_i} \frac{(-i C^{(3)}_{a \bar{b} \bar{n}})}{(p-p')^2-m^2_n} \frac{(-i C^{(3)}_{b \bar{a} \bar{j}})}{(p'-p)^2-m^2_j} (-i C^{(3)}_{i j n})$}};


\filldraw[black] (-4.4*\y,0.2*\y-6*\y)  node[anchor=west] {\normalsize{$\lim_{\theta_k \to \pm \infty}$}};

\draw[] (-2.2*\y,2*\y-6*\y) arc(165:195:8*\y);
\draw[] (+2.4*\y,2*\y-6*\y) arc(15:-15:8*\y);

\draw[directed] (0*\y,1*\y-6*\y) -- (0*\y,-6*\y);
\draw[directed] (-0.7*\y,1.7*\y-6*\y) -- (0*\y,1*\y-6*\y);
\draw[directed] (0*\y,1*\y-6*\y) -- (+0.7*\y,1.7*\y-6*\y);
\draw[directed] (0.866025*\y,-0.5*\y-6*\y) -- (0*\y,0*\y-6*\y);
\draw[directed] (1.82224*\y,-0.243782*\y-6*\y) -- (0.866025*\y,-0.5*\y-6*\y);
\draw[directed] (0.866025*\y,-0.5*\y-6*\y) -- (1.12224*\y,-1.45622*\y-6*\y);
\draw[directed] (-0.866025*\y,-0.5*\y-6*\y) -- (0*\y,0*\y-6*\y);
\draw[directed] (-1.82224*\y,-0.243782*\y-6*\y) -- (-0.866025*\y,-0.5*\y-6*\y);
\draw[directed] (-0.866025*\y,-0.5*\y-6*\y) -- (-1.12224*\y,-1.45622*\y-6*\y);

\filldraw[black] (-1*\y,2*\y-6*\y)  node[anchor=west] {\tiny{$e(k)$}};
\filldraw[black] (0.4*\y,2*\y-6*\y)  node[anchor=west] {\tiny{$a(p)$}};
\filldraw[black] (-2.4*\y,-0.4*\y-6*\y)  node[anchor=west] {\tiny{$a(p)$}};
\filldraw[black] (-1.5*\y,-1.7*\y-6*\y)  node[anchor=west] {\tiny{$b(p')$}};
\filldraw[black] (1.8*\y,-0.4*\y-6*\y)  node[anchor=west] {\tiny{$b(p')$}};
\filldraw[black] (0.9*\y,-1.7*\y-6*\y)  node[anchor=west] {\tiny{$e(k)$}};

\filldraw[black] (0.1*\y,0.6*\y-6*\y)  node[anchor=west] {\tiny{$i$}};
\filldraw[black] (0.4*\y,0*\y-6*\y)  node[anchor=west] {\tiny{$j$}};
\filldraw[black] (-0.8*\y,0*\y-6*\y)  node[anchor=west] {\tiny{$n$}};

\filldraw[black] (2.8*\y,0.2*\y-6*\y)  node[anchor=west] {\normalsize{$= \frac{(-i C^{(3)}_{e \bar{a} \bar{i}})}{\infty-m^2_i} \frac{(-i C^{(3)}_{a \bar{b} \bar{n}})}{(p-p')^2-m^2_n} \frac{(-i C^{(3)}_{b \bar{e} \bar{j}})}{\infty-m^2_j} (-i C^{(3)}_{i j n})=0$}};

\end{tikzpicture}
\caption{Examples of surviving (first row) and cancelling (second row) diagrams in the limit $\theta_k \to \pm \infty$. 
The values of both diagrams do not change if we send $\theta_k \to + \infty$ or $\theta_k \to - \infty$.}
\label{Example_of_diagrams_in_large_rapidity_limit}
\end{center}
\end{figure}

Taking the large rapidity limit of~\eqref{Vabe_to_abe_on_shell_2} we obtain
\begin{multline}
\label{Vabe_to_abe_on_shell_limit_thetak_plus_minus_infty_difference}
\Bigl( \lim_{\theta_k \to + \infty} -\lim_{\theta_k \to - \infty} \Bigl) V^{(\text{on})}_{abe\to abe}=\\
\Bigl( i \frac{\partial}{\partial \theta_p}M^{(0)}_{ab \to ab} (p, p') -i \coth{\theta_{p p'}}   M^{(0)}_{ab \to ab} (p, p')\Bigl) \cdot \left( \frac{M^{(0)}_{ae \to ae} (\infty)}{m_a^2}-\frac{M^{(0)}_{be \to be} (\infty)}{m_b^2} \right)  \,.
\end{multline}
In~\eqref{Vabe_to_abe_on_shell_limit_thetak_plus_minus_infty_difference} we label by $M^{(0)}_{ij \to ij} (\infty)$ the amplitude associated with the scattering of two particles of types $i$ and $j$ in the limit $\theta_i - \theta_j=\infty$. 

We notice that $M^{(0)}_{ab \to ab} (p, p')$ contains simple poles at all imaginary values of $\theta_{p p'}$ corresponding to the propagation of bound states in $M^{(0)}_{ab \to ab} (p, p')$ while $\frac{\partial}{\partial \theta_p}M^{(0)}_{ab \to ab} (p, p')$ contains second-order poles at these locations. For this reason, the first factor on the r.h.s. of the equality in~\eqref{Vabe_to_abe_on_shell_limit_thetak_plus_minus_infty_difference} cannot be null as it is the sum of two functions with poles of different orders. The only possibility to have~\eqref{Vabe_to_abe_on_shell_limit_thetak_plus_minus_infty_difference} vanishing is that
$$
\frac{M^{(0)}_{b e \to b e} (\infty)}{m_b^2 }=\frac{M^{(0)}_{a e \to a e} (\infty)}{m_a^2} \,.
$$
Since this has to be true for any particles $a$, $b$ and $e$ then Property~\ref{Property_4_point_couplings} is proven.

Using~\eqref{values_of_tree_level_amplitudes_at_infty} then~\eqref{requirement_equal_limit_PlusMinusInfinite} is satisfied and it is possible to show that
\begin{equation}
\label{Vabe_to_abe_on_shell_limit_thetak_plus_minus_infty_result}
\begin{split}
& \lim_{\theta_k \to \pm \infty} V^{(\text{on})}_{abe\to abe}={\color{red} -i \ai \, m_e^2 M^{(0)}_{ab \to ab} (p, p' , p , p')}\\
&{\color{blue}+ i \Bigl( M^{(0)}_{ae \to ae} (\infty) \frac{\partial}{\partial p^2}M^{(0)}_{ab \to ab} (p, p' , p , p') \Bigl|_{p^2=m^2_a}}{\color{blue}+  \frac{\partial}{\partial p'^2} M^{(0)}_{ab \to ab} (p, p' , p , p')\Bigl|_{p'^2=m^2_b} M^{(0)}_{be \to be} (\infty) \Bigl)} \,.
\end{split}
\end{equation}
For the terms in the expression above we used the same colours as in~\eqref{Vabe_to_abe_on_shell_2} to indicate from which blocks of~\eqref{Vabe_to_abe_on_shell_2} they are generated.

\subsection{Absence of poles at finite rapidity}

In this section we check that $V^{(\text{on})}_{abe\to abe}$ does not have poles at finite values of $\theta_k \in \mathbb{R}$; this is fundamental if we want to integrate $V^{(\text{on})}_{abe\to abe}$, as a function of $\theta_k$, on the real line.

\begin{figure}
\begin{center}
\begin{tikzpicture}
\tikzmath{\y=1.2;}

\draw[directed] (0*\y,0*\y) -- (-1.3*\y,1*\y);
\draw[directed] (-1.3*\y,1*\y) -- (0*\y,2.3*\y);
\draw[directed] (0*\y,0*\y) -- (1.3*\y,1.3*\y);
\draw[directed] (1.3*\y,1.3*\y) -- (0*\y,2.3*\y);
\filldraw[black] (-1.1*\y,2*\y)  node[anchor=west] {\tiny{$a(p)$}};
\filldraw[black] (-1.2*\y,0.2*\y)  node[anchor=west] {\tiny{$e(k)$}};
\filldraw[black] (+0.7*\y,0.5*\y)  node[anchor=west] {\tiny{$a$}};
\filldraw[black] (+0.6*\y,1.9*\y)  node[anchor=west] {\tiny{$e$}};
\filldraw[black] (-1.45*\y,1*\y)  node[anchor=west] {\tiny{$\bullet$}};
\filldraw[black] (1.12*\y,1.3*\y)  node[anchor=west] {\tiny{$\bullet$}};

\draw[directed] (0*\y,0*\y-4*\y) -- (-1.3*\y,1*\y-4*\y);
\draw[directed] (-1.3*\y,1*\y-4*\y) -- (0*\y,2.3*\y-4*\y);
\draw[directed] (+1.3*\y,1*\y-4*\y) -- (0*\y,2.3*\y-4*\y);
\draw[directed] (0*\y,0*\y-4*\y) -- (+1.3*\y,1*\y-4*\y);
\filldraw[black] (-1.1*\y,2*\y-4*\y)  node[anchor=west] {\tiny{$a(p)$}};
\filldraw[black] (-1.2*\y,0.2*\y-4*\y)  node[anchor=west] {\tiny{$e(k)$}};
\filldraw[black] (+0.7*\y,0.5*\y-4*\y)  node[anchor=west] {\tiny{$e$}};
\filldraw[black] (+0.6*\y,1.9*\y-4*\y)  node[anchor=west] {\tiny{$a$}};
\filldraw[black] (-1.45*\y,1*\y-4*\y)  node[anchor=west] {\tiny{$\bullet$}};
\filldraw[black] (+1.1*\y,1*\y-4*\y)  node[anchor=west] {\tiny{$\bullet$}};

\draw[->] (1.5*\y,0.4*\y) -- (3.5*\y,-0.6*\y);
\filldraw[black] (2.1*\y,0.3*\y)  node[anchor=west] {\tiny{$\theta{p k} \to 0$}};
\draw[->] (1.5*\y,-2.4*\y) -- (3.5*\y,-1.4*\y);
\filldraw[black] (2.1*\y,-2.2*\y)  node[anchor=west] {\tiny{$\theta{p k} \to 0$}};


\draw[directed] (4.5*\y,-3*\y+0.2*\y) -- (4.5*\y,-1.36*\y+0.2*\y);
\filldraw[black] (4.32*\y,-1.45*\y+0.2*\y)  node[anchor=west] {\tiny{$\bullet$}};
\draw[directed] (4.5*\y,-1.36*\y+0.2*\y) -- (4.5*\y,-1.36*\y+1.84*\y+0.2*\y);
\filldraw[black] (4.5*\y,-0.5*\y+0.2*\y)  node[anchor=west] {\tiny{$a(p)$}};
\filldraw[black] (4.5*\y,-2.5*\y+0.2*\y)  node[anchor=west] {\tiny{$e(k)$}};

\end{tikzpicture}
\caption{Elastic (top left) and inelastic (bottom left) configurations of $M^{(0)}_{ae\to ae}(p,k)$ and their collinear limit.}
\label{Elastic_inelastic_configurations_collinear_limit}
\end{center}
\end{figure}
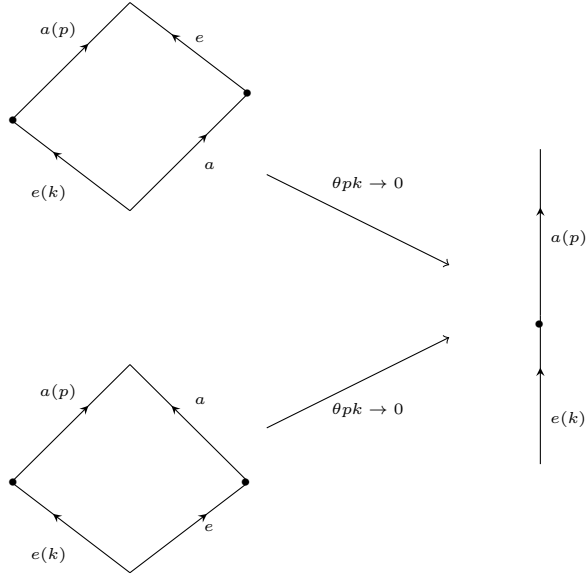
From~\eqref{Vabe_to_abe_on_shell_2} we notice that the only possible singularities of $V^{(\text{on})}_{abe\to abe}$ are located at $\theta_k=\theta_p$ and $\theta_k=\theta_{p'}$; indeed the singularities of the two-to-two tree-level amplitudes entering the r.h.s. of~\eqref{Vabe_to_abe_on_shell_2} are all located at imaginary values of $\theta_k$ since they have to correspond to the propagation of bound states. 
We study the potential singularity at $\theta_k = \theta_p$; the case $\theta_k = \theta_{p'}$ can be studied analogously. Let us take the limit $\theta_k \to \theta_p$ of the red terms in~\eqref{Vabe_to_abe_on_shell_2}: in this limit, the potential singular part of these terms is given by
\begin{equation}
\label{residue_at_thetak_pole_red_contributions}
\begin{split}
{\color{red} \frac{i}{2}\frac{M^{(0)}_{ae \to ae} (p, p)}{\sinh{\theta_{pk}}} \frac{\cosh \theta_{p p'}}{\sinh \theta_{p p'}} \Bigl( -  \frac{M^{(0)}_{b e \to b e} (p', p)}{ m_e^2}   +\frac{M^{(0)}_{ab \to ab} (p, p')}{ m_a^2}    \Bigl) \, .}
\end{split}
\end{equation}
If $e$ and $a$ are particles of different types (i.e. $e \ne a$) then $M^{(0)}_{ae \to ae} (p, p)=0$ and the residue at the pole is null. 
This is clear by the fact that we consider theories with purely elastic tree-level S-matrices for which $M^{(0)}_{ae \to ae}=0$ on the inelastic branch of the kinematics where the momenta of the outgoing particles are different from the momenta of the incoming particles (i.e. when $p^{(\text{in})}_a \ne p^{(\text{out})}_a$ and $p^{(\text{in})}_e \ne p^{(\text{out})}_e$). This kinematical configuration, analytically continued to imaginary rapidities, is depicted on the bottom left of Figure~\ref{Elastic_inelastic_configurations_collinear_limit}; the collinear limit $\theta_k \to \theta_p$ corresponds to the configuration at which the elastic and inelastic branch of the kinematics overlap. Since the amplitude is null on the inelastic branch then it needs to hold $M^{(0)}_{ae \to ae}(p,p)=0$.

In the particular case in which $e=a$ then $M^{(0)}_{ae \to ae}(p,p)$ can be nonzero since the scattering is elastic in both branches of the kinematics; in this case both quadrilateral on the l.h.s. of Figure~\ref{Elastic_inelastic_configurations_collinear_limit} become identical parallelograms and describe an elastic scattering process. 
However, even in this case, the singular part of the red contributions vanishes since the terms 
in parenthesis in~\eqref{residue_at_thetak_pole_red_contributions} sum to zero. A similar analysis can be carried out for the limit $\theta_k \to \theta_{p'}$. 
The potential singularities of the orange contributions in~\eqref{Vabe_to_abe_on_shell_2} cancel separately; this can easily be checked following the same approach used above. In the end, the quantity in~\eqref{Vabe_to_abe_on_shell_2} has no singularities for all values of $\theta_k \in \mathbb{R}$.

\subsection{Integrating on-shell amplitudes}
\label{Integration_of_Sum_e_Vabe_OnShell}

Since the quantity in~\eqref{Vabe_to_abe_on_shell_2}, as a function of $\theta_k$, has no poles on the real axis then we can perform its integration between $-\infty$ and $+\infty$.
To perform the integration first we remove from~\eqref{Vabe_to_abe_on_shell_2} the contributions which survive in the limit $\theta_k \to \infty$ otherwise the integral would be ill-defined: in loop computations, these contributions are generated by cuts of tadpole diagrams which must be cancelled by infinite counterterms.
After the infinite-rapidity contributions \eqref{Vabe_to_abe_on_shell_limit_thetak_plus_minus_infty_result} are removed the expression in~\eqref{Vabe_to_abe_on_shell_2} becomes
\begin{equation}
\label{Vabe_to_abe_on_shell_infty_removed}
\begin{split}
&\hat{V}^{(\text{on})}_{abe\to abe}= {\color{red} \frac{i}{2 m_e^2} \frac{\cosh{\theta_{p p'}}}{\sinh{\theta_{k p}} \sinh{\theta_{k p'}}} M^{(0)}_{ae \to ae} (p, k , p , k) M^{(0)}_{b e \to b e} (p', k , p' , k)}\\
&{\color{red}+\frac{i}{2 m_a^2} \frac{\cosh{\theta_{p' k}}}{\sinh{\theta_{p p'}} \sinh{\theta_{p k}}} M^{(0)}_{ae \to ae} (p, k , p , k) M^{(0)}_{ab \to ab} (p, p' , p , p')}\\
&{\color{red}+\frac{i}{2 m_b^2} \frac{\cosh{\theta_{p k}}}{\sinh{\theta_{p' p}} \sinh{\theta_{p' k}}} M^{(0)}_{ab \to ab} (p, p' , p , p') M^{(0)}_{b e \to b e} (p', k , p' , k) +i \ai \, m_e^2 M^{(0)}_{ab \to ab} (p, p' , p , p')}\\
&{\color{blue}+ i \frac{\partial}{\partial k^2} \Bigl( M^{(0)}_{ae \to ae} (p, k , p , k) M^{(0)}_{be \to be} (p', k , p' , k) \Bigl)\Bigl|_{k^2=m^2_e}}\\
&{\color{blue}+ i \frac{\partial}{\partial p^2} \Bigl( \hat{M}^{(0)}_{ae \to ae} (p, k , p , k) M^{(0)}_{ab \to ab} (p, p' , p , p') \Bigl)\Bigl|_{p^2=m^2_a}}\\
&{\color{blue}+ i \frac{\partial}{\partial p'^2} \Bigl( M^{(0)}_{ab \to ab} (p, p' , p , p') \hat{M}^{(0)}_{be \to be} (p', k , p' , k) \Bigl)\Bigl|_{p'^2=m^2_b}}\\
&{\color{orange}+\frac{i M^{(0)}_{ae \to ae} (p, k , p , k)}{2 m_e^2 \tanh{\theta_{pk}}} \frac{\partial}{\partial \theta_k}M^{(0)}_{be \to be} (p', k , p' , k)+\frac{iM^{(0)}_{be \to be} (p', k , p' , k)}{2 m_e^2 \tanh{\theta_{p' k}}} \frac{\partial}{\partial \theta_k}M^{(0)}_{ae \to ae} (p, k , p , k)}\\
&{\color{orange}+\frac{i M^{(0)}_{ae \to ae} (p, k , p , k)}{2 m_a^2 \tanh{\theta_{kp}}} \frac{\partial}{\partial \theta_p}M^{(0)}_{ab \to ab} (p, p' , p , p')+\frac{iM^{(0)}_{ab \to ab} (p, p' , p , p')}{2 m_a^2 \tanh{\theta_{p' p}}} \frac{\partial}{\partial \theta_p}M^{(0)}_{ae \to ae} (p, k , p , k)}\\
&{\color{orange}+\frac{i M^{(0)}_{ab \to ab} (p, p' , p , p')}{2 m_b^2 \tanh{\theta_{pp'}}} \frac{\partial}{\partial \theta_{p'}}M^{(0)}_{b e \to b e} (p' , k, p', k)+\frac{iM^{(0)}_{b e \to b e} (p', k , p' , k)}{2 m_b^2 \tanh{\theta_{k p'}}} \frac{\partial}{\partial \theta_{p'}}M^{(0)}_{ab \to ab} (p, p' , p , p') \, .} 
\end{split}
\end{equation}

In the following we assume $\theta_{p'}<\theta_{p}$. 
Moreover, we integrate over the line
\begin{equation}
\int_{\gamma} d \theta_k \equiv \int_{-\infty}^{\theta_{p'}-\epsilon} d \theta_k +\int_{\theta_{p'}+\epsilon}^{\theta_{p}-\epsilon} d \theta_k +\int_{\theta_{p}+\epsilon}^{+\infty} d \theta_k
\end{equation}
in such a way to avoid the point $\theta_{p'}$ and $\theta_p$. Since the integrand has no poles at these points then there is no difference between integrating on the real line or on $\gamma$.
The advantage of integrating on $\gamma$ is that the integral of each single term composing~\eqref{Vabe_to_abe_on_shell_infty_removed} is well defined (note indeed that the single terms in~\eqref{Vabe_to_abe_on_shell_infty_removed} can be singular at $\theta_k=\theta_{p'}$ and $\theta_k=\theta_{p}$ and only their sum is finite).
Let us define
\begin{equation}
\label{definition_of_Vab_to_ab_colored}
\hat{V}^{(*)}_{ab\to ab} \equiv \frac{1}{8\pi} \sum_{e=1}^r \int_\gamma d\theta_k \hat{V}^{(*)}_{abe\to abe} \, ,
\end{equation}
where $*$ can be `red', `blue' or `orange' and indicates the different coloured terms composing~\eqref{Vabe_to_abe_on_shell_infty_removed}.
With this definition the expression in~\eqref{definition_of_Vab_to_ab_term_i_123_onshell_limit} can be written as
\begin{equation}
\label{definition_of_Vab_to_ab_as_sum_of_colours}
\hat{V}^{(\text{on})}_{ab\to ab} = \hat{V}^{(\text{red})}_{ab\to ab}+\hat{V}^{(\text{blue})}_{ab\to ab}+\hat{V}^{(\text{orange})}_{ab\to ab} \, .
\end{equation}
In the following, we compute the three terms on the r.h.s. of~\eqref{definition_of_Vab_to_ab_as_sum_of_colours} separately.

\paragraph{Integrating the blue contributions.}
Let us start with the blue contributions in~\eqref{Vabe_to_abe_on_shell_infty_removed}, which separately do not contain poles on the real line.
Substituting these contributions into the integrand of~\eqref{definition_of_Vab_to_ab_colored} and using~\eqref{definition_necessary_for_mass_renormalization_first_time_ga_appears}, \eqref{coefficients_of_the_Sigma_expansion} and the first relation in~\eqref{definition_of_mass_renormalization_and_tab_to_diagonalise_the_mass_matrix} we obtain 
\begin{equation}
\label{eq_blue_Vab_appendix}
\begin{split}
\hat{V}^{(\text{blue})}_{ab\to ab}&= \bigl( \Sigma^{(1)}_{aa}+ \Sigma^{(1)}_{bb} \bigl) M^{(0)}_{ab}(p, p')\\
&-\delta m_a^2 \frac{\partial}{\partial p^2} M^{(0)}_{ab}(p, p' )\Bigl|_{p^2=m_a^2}-\delta m_b^2 \frac{\partial}{\partial p'^2} M^{(0)}_{ab}(p, p' )\Bigl|_{p'^2=m_b^2}\\
&+ \frac{i}{8 \pi} \sum_{e=1}^r \int_{-\infty}^{+\infty} d \theta_k \frac{\partial}{\partial k^2} \Bigl( M^{(0)}_{ae} (p, k) M^{(0)}_{b e} (p', k) \Bigl)\Bigl|_{k^2=m^2_e}\, .
\end{split}
\end{equation}

\paragraph{Integrating the red contributions.}

We move now to the red contributions.
Let us consider the $k$-dependent part in the second row of~\eqref{Vabe_to_abe_on_shell_infty_removed}. Its integral in $\theta_k$ can be written as
\begin{equation}
\begin{split}
 &\int_\gamma d\theta_k  \frac{\cosh{\theta_{p' k}}}{\sinh{\theta_{p k}}} M^{(0)}_{ae} (p, k)=\int_\gamma d\theta_k  \frac{\cosh{(\theta_{p' p}+\theta_{p k})}}{\sinh{\theta_{p k}}} M^{(0)}_{ae} (p, k)\\
 &=\int_\gamma d\theta_k  \cosh{\theta_{p' p}} \coth{\theta_{p k}} M^{(0)}_{ae} (p, k)+\int_\gamma d\theta_k  \sinh{\theta_{p' p}} M^{(0)}_{ae} (p, k) 
 \end{split}
\end{equation}
The first term in the last row of the expression above is odd in $\theta_{pk}$ and its integration in $\theta_{pk}$ between $-\infty$ and $+\infty$ is therefore null. 
However, we are integrating in $\theta_k$ and the integral with respect to this variable is non-zero since the integrand function does not tend to zero in the limit of large rapidity. After having performed a  change of integration variable we obtain
\begin{equation}
 +2 \theta_p \cosh{\theta_{p' p}} m_a^2 m_e^2 \ai+\sinh{\theta_{p' p}} \int_\gamma d\theta_k  M^{(0)}_{ae} (p, k) \,.
\end{equation}

Multiplying by $1/8 \pi$ and the remaining $k$-independent factors in the second row of~\eqref{Vabe_to_abe_on_shell_infty_removed}, and summing over the different types of particles $e$ we obtain
\begin{equation}
\begin{split}
&\frac{i \ai}{8 \pi} \theta_p \coth{\theta_{p p'}} M^{(0)}_{ab} (p, p') \, \sum_{e=1}^r m_e^2+\frac{1}{2}  \frac{\delta m_a^2}{m_a^2} M^{(0)}_{ab}(p, p')\\
&-\frac{i}{2} M^{(0)}_{ab} (p, p') \frac{1}{8 \pi} \ai \sum_{e=1}^r m_e^2 \int_{-\infty}^{+\infty} d\theta_k \,.
\end{split}
\end{equation}
Performing similar computations for the first and third rows of~\eqref{Vabe_to_abe_on_shell_infty_removed} we obtain
\begin{equation}
\label{eq_red_Vab_appendix}
\begin{split}
&\hat{V}^{(\text{red})}_{ab\to ab}= \frac{i \ai}{8 \pi} \theta_{p p'} \coth{\theta_{p p'}} M^{(0)}_{ab} (p, p') \, \sum_{e=1}^r m_e^2\\
&+\frac{1}{2} \Bigl( \frac{\delta m_a^2}{m_a^2} +\frac{\delta m_b^2}{m_b^2} \Bigl) M^{(0)}_{ab}(p, p' ) + \sum_{e=1}^r \int_{\gamma} d \theta_k \frac{i p \cdot p' }{  \pi}  S^{(0)}_{ae} (p, k ) S^{(0)}_{b e} (p', k ) \,.
\end{split}
\end{equation}

\paragraph{Integrating the orange contributions.}
The contribution to $\hat{V}^{(\text{orange})}_{ab\to ab}$ carried from the second and third rows in the orange part of~\eqref{Vabe_to_abe_on_shell_infty_removed} can be computed in a similar way to before and is given by
\begin{equation}
\label{appendix_easy_part_of_orange_contribution}
- \frac{i \ai}{8 \pi} \theta_{p p'}\frac{\partial}{\partial \theta_{p p'}} M^{(0)}_{ab} (p, p') \sum_{e=1}^r \, m_e^2 \,.
\end{equation}
Computing the integral of the first row of the orange part of~\eqref{Vabe_to_abe_on_shell_infty_removed} requires some more steps: noting that 
$$
\frac{\partial}{\partial \theta_k}M^{(0)}_{ae} (p, k)=\frac{\partial}{\partial \theta_k}\hat{M}^{(0)}_{ae} (p, k)=-\frac{\partial}{\partial \theta_p}\hat{M}^{(0)}_{ae} (p, k)
$$ 
and 
$$
\frac{\partial}{\partial \theta_k}M^{(0)}_{be} (p', k)=\frac{\partial}{\partial \theta_k} \hat{M}^{(0)}_{be} (p', k)=-\frac{\partial}{\partial \theta_{p'}}\hat{M}^{(0)}_{be} (p', k)
$$
then the contribution of this row to $\hat{V}^{(\text{orange})}_{ab\to ab}$ can be written as
\begin{equation}
\label{appendix_first_row_orange_integrated_first}
-\frac{i}{16 \pi} \sum_{e=1}^r \frac{1}{m_e^2} \biggl( \frac{\partial}{\partial \theta_{p'}} \int_\gamma d\theta_k  \frac{ M^{(0)}_{ae} (p, k)}{ \tanh{\theta_{pk}}} \hat{M}^{(0)}_{be} (p', k) + \frac{\partial}{\partial \theta_{p}} \int_\gamma d\theta_k \frac{M^{(0)}_{be} (p', k)}{ \tanh{\theta_{p' k}}} \hat{M}^{(0)}_{ae } (p, k) \biggl) \,.
\end{equation}
Both integrals in the expression above are convergent since both the integrand functions tend to zero for large rapidity $\theta_k$. Consequently, we can change the integration variable to $\theta_{kp}$ in the first integral and to $\theta_{k p'}$ in the second integral without caring about the effects of the boundaries of integration. 
Using the fact that the first integrand above is a function of $\theta_{kp}$ and $\theta_{p p'}$ (note that it is also a function of $\theta_{k p'}$ but this variable can be written as $\theta_{k p'}=\theta_{kp}+\theta_{p p'}$) its integration with respect to $\theta_{kp}$ generates a function of $\theta_{p p'}$ alone.
A similar argument can be repeated for the second integral. Due to this fact, we can replace
$$ 
\frac{\partial}{\partial \theta_{p'}} \to -\frac{\partial}{\partial \theta_{pp'}}
$$
in the first term in~\eqref{appendix_first_row_orange_integrated_first} and
$$ 
\frac{\partial}{\partial \theta_{p}} \to \frac{\partial}{\partial \theta_{pp'}}
$$
in the second term in~\eqref{appendix_first_row_orange_integrated_first}. We end up with
\begin{equation}
\label{appendix_first_row_orange_integrated_second}
\frac{i}{16 \pi} \sum_{e=1}^r \frac{1}{m_e^2} \frac{\partial}{\partial \theta_{pp'}} \biggl( \int_\gamma d\theta_k  \frac{ M^{(0)}_{ae} (p, k)}{ \tanh{\theta_{pk}}} \hat{M}^{(0)}_{be} (p', k) - \int_\gamma d\theta_k \frac{M^{(0)}_{be} (p', k)}{ \tanh{\theta_{p' k}}} \hat{M}^{(0)}_{ae} (p, k) \biggl) \,.
\end{equation}

Using
$$
\hat{M}^{(0)}_{ae} (p', k)=M^{(0)}_{ae} (p, k)- \ai m_a^2 m_e^2
$$
and
$$
\hat{M}^{(0)}_{be} (p', k)=M^{(0)}_{be} (p', k)- \ai m_b^2 m_e^2
$$
the expression in~\eqref{appendix_first_row_orange_integrated_second} becomes
\begin{equation}
\label{appendix_first_row_orange_integrated_third}
\begin{split}
&- \frac{i \bigl( \ai \bigl)^2}{8 \pi} m_a^2 m_b^2  \sum_{e=1}^r m_e^2 -\frac{i}{\pi} p \cdot p' \sum_{e=1}^r \int_\gamma d\theta_k S^{(0)}_{ae} (p, k ) S^{(0)}_{b e} (p', k ) \\
&-\frac{i}{\pi} m_a m_b \sinh{\theta_{pp'}}  \sum_{e=1}^r \frac{\partial}{\partial \theta_{p p'}} \int_\gamma d\theta_k S^{(0)}_{ae} (p, k ) S^{(0)}_{b e} (p', k ).
\end{split}
\end{equation}
The sum of~\eqref{appendix_easy_part_of_orange_contribution} and~\eqref{appendix_first_row_orange_integrated_third} is then
\begin{equation}
\label{eq_orange_Vab_appendix}
\begin{split}
\hat{V}^{(\text{orange})}_{ab\to ab}=&- \frac{i \ai}{8 \pi} \theta_{p p'}\frac{\partial}{\partial \theta_{p p'}} M^{(0)}_{ab} (p, p') \sum_{e=1}^r \, m_e^2- \frac{i \bigl( \ai \bigl)^2}{8 \pi} m_a^2 m_b^2  \sum_{e=1}^r m_e^2\\
&-\frac{i}{\pi} p \cdot p' \sum_{e=1}^r \int_\gamma d\theta_k S^{(0)}_{ae} (p, k ) S^{(0)}_{b e} (p', k ) \\
&-\frac{i}{\pi} m_a m_b \sinh{\theta_{pp'}}  \sum_{e=1}^r \frac{\partial}{\partial \theta_{p p'}} \int_\gamma d\theta_k S^{(0)}_{ae} (p, k ) S^{(0)}_{b e} (p', k ) \,.
\end{split}
\end{equation}

\paragraph{Combining all coloured contributions.}
Finally, summing~\eqref{eq_blue_Vab_appendix}, \eqref{eq_red_Vab_appendix} and~\eqref{eq_orange_Vab_appendix} we obtain
\begin{equation}
\label{appendix_total_Vab_to_ab_onshell}
\begin{split}
&\hat{V}^{(\text{on})}_{ab\to ab}=
\bigl( \Sigma^{(1)}_{aa}+ \Sigma^{(1)}_{bb} \bigl) M^{(0)}_{ab}(p, p')\\
&-\delta m_a^2 \frac{\partial}{\partial p^2} M^{(0)}_{ab}(p, p' )\Bigl|_{p^2=m_a^2}-\delta m_b^2 \frac{\partial}{\partial p'^2} M^{(0)}_{ab}(p, p' )\Bigl|_{p'^2=m_b^2}\\
&+ \frac{i}{8 \pi} \sum_{e=1}^r \int_{-\infty}^{+\infty} d \theta_k \frac{\partial}{\partial k^2} \Bigl( M^{(0)}_{ae} (p, k) M^{(0)}_{b e} (p', k) \Bigl)\Bigl|_{k^2=m^2_e}\\
&+\frac{i \ai}{8 \pi} \theta_{p p'} \coth{\theta_{p p'}} M^{(0)}_{ab} (p, p') \, \sum_{e=1}^r m_e^2\\
&+\frac{1}{2} \Bigl( \frac{\delta m_a^2}{m_a^2} +\frac{\delta m_b^2}{m_b^2} \Bigl) M^{(0)}_{ab}(p, p' )- \frac{i \ai}{8 \pi} \theta_{p p'}\frac{\partial}{\partial \theta_p} M^{(0)}_{ab} (p, p') \sum_{e=1}^r \, m_e^2\\
&- \frac{i \bigl( \ai \bigl)^2}{8 \pi} m_a^2 m_b^2  \sum_{e=1}^r m_e^2-\frac{i}{\pi} m_a m_b \sinh{\theta_{pp'}}  \sum_{e=1}^r \frac{\partial}{\partial \theta_{p p'}} \ \text{p.v.} \int^{+\infty}_{-\infty} d\theta_k \,  S^{(0)}_{ae} (p, k ) S^{(0)}_{b e} (p', k ) \,.
\end{split}
\end{equation}

\section{Orbit relations}
\label{app:orbit_relation}

In this appendix, we establish some orbit relations necessary in Section~\ref{sec:SL-affine-toda} for the derivation of the one-loop S-matrices of simply-laced affine Toda theories.
We start briefly reviewing some known results of Coxeter geometry; the reader is invited to look at~\cite{Dorey:2021hub,Corrigan:1994nd} for more extended discussions.

Let us consider a semisimple Lie algebra $\mathfrak{g}$ of rank $r$ associated with a simply-laced Dynkin diagram. 
Following the argument presented at pp. 158-161 of~\cite{carter1989simple} we can define an orthonormal basis $\{ z_s\}$ of the Coxeter element by introducing the vectors
\begin{equation}
\label{eq:def1-as}
\begin{split}
    a^\bullet_s &= \sum_{a \in \bullet} q^a_s \alpha_a , \\
    a^\circ_s &= \sum_{a \in \circ} q^a_s \alpha_a \,,
\end{split}
\end{equation}
where $q_s^a$ are components of the eigenvectors $\{q_s\}$ of the Cartan matrix. Then a suitable set of eigenvectors of $w$ is provided by
\begin{equation}
\label{eq:zs_as_functions_of_as}
\begin{split}
   z_s &= \sqrt{\frac{h}{2}} \frac{1}{|q_s| \sin (\theta_s)} \bigl( e^{i \theta_s} a^\bullet_s + a^\circ_s \bigl)\\
    z_{h-s} &= \sqrt{\frac{h}{2}} \frac{1}{|q_s| \sin (\theta_s)} \bigl( e^{-i \theta_s} a^\bullet_s + a^\circ_s \bigl)
\end{split}
\end{equation}
which satisfy
\begin{equation}
\label{eq:scalar-product-zs}
z^\dagger_{s}=z_{h-s} \quad \text{and} \quad (z_s , z^\dagger_{s'}) = h \delta_{s s'} \,.
\end{equation}
From this fact, it follows that the identity operator in $\mathbb{R}^r$ can be decomposed onto these eigenvectors as
\begin{equation}
\label{eq:completeness-zs}
    \sum_{s} \frac{z_{h-s} \otimes z_s}{h} = \mathbb{I} \,.
\end{equation}

It is possible to show that the projectors of the roots $\{\gamma_a\}$ defined in~\eqref{eq:roots_weigths_connection} on this basis are
\begin{equation}
\label{eq:projections-roots}
\begin{split}
&P_s(\gamma_a)= \sqrt{\frac{2}{h}}(z_{h-s}, \gamma_a) =  2 \frac{q^a_s}{|q_s|} \hspace{12mm} \text{if} \ a \in \circ ,\\
&P_s(\gamma_a)=\sqrt{\frac{2}{h}}(z_{h-s}, \gamma_a) = 2  \frac{q^a_s}{|q_s|}e^{-i \theta_s} \quad \text{if} \ a \in \bullet \,,
\end{split}
\end{equation}
with $\theta_s$ given in~\eqref{eq:eigenvectors_w}. From~\eqref{eq:projections-roots} we see that the black roots are rotated by an angle $-\theta_s$ with respect to the white roots: this is just a generalization of the angles $U_{\gamma_a}$ presented in Section \ref{sec:coxeter-geometry} to the spin $s$ eigenplanes of $w$.
The projections for the remaining roots of the system are obtained by rotating the projectors of $\{\gamma_a\}^r_{a=1}$ as
\begin{equation}
P_s(w^{l}\gamma_a) = e^{2 i l \theta_{s}} P_s(\gamma_a) \,.
\end{equation}
This simply follows from the fact that $\{ z_s \}$ are eigenvectors of the Coxeter element.

Since $\{\gamma_a\}^r_{a=1}$ are real vectors in $\mathbb{R}^r$ and $z^\dagger_{s}=z_{h-s}$ then it holds that
\begin{equation}
(z_{h-s}, \gamma_a)=(z_{s}, \gamma_a)^*
\end{equation}
from which 
\begin{equation}
\label{eq:reality-conditions}
\begin{split}
&q^a_s=q^a_{h-s} \hspace{8mm}\text{if} \ a  \in \circ \,,\\
&q^a_s=-q^a_{h-s} \hspace{5mm} \text{if} \ a  \in \bullet \,.
\end{split}
\end{equation}
The extra minus sign for $\bullet$ roots is due to the extra phase factor $e^{-i\theta_s}$ in \eqref{eq:projections-roots} and the identity $\theta_{h-s} = \pi-\theta_s$.

Having defined the projections we now establish some useful identities involving them. 

\begin{lemma}
\label{property_scalar_pr_asasprime}

The set $\{q_s\}$ of eigenvectors of the Cartan matrix used to define the eigenvectors of $w$ must satisfy
\begin{equation}
\label{eq:completeness-rel-qb}
\sum_{a \in \bullet} q^a_s q^a_{s'} = \frac{|q_s|^2}{2}  \bigl( \delta_{s' , s} - \delta_{s' , h-s}\bigl),
\end{equation}
\begin{equation}
\label{eq:completeness-rel-qw}
\sum_{a \in \circ} q^a_s q^a_{s'} = \frac{|q_s|^2}{2} \bigl( \delta_{s' , s} + \delta_{s' , h-s}\bigl).
\end{equation}
\end{lemma}

\begin{proof}
Inverting~\eqref{eq:zs_as_functions_of_as} we can write
\begin{equation}
\label{eq:def2-as}
\begin{split}
a^\bullet_s &= \frac{|q_s|}{i \sqrt{2h}} \bigl( z_s - z_{h-s} \bigl),\\
a^\circ_s &= \frac{i |q_s|}{\sqrt{2h}} \bigl( e^{-i \theta_s} z_s - e^{i \theta_s} z_{h-s} \bigl).
\end{split}
\end{equation}
Combining these relations with~\eqref{eq:scalar-product-zs} we obtain
\begin{equation}
\label{eq:sp1}
\begin{split}
&(a^\bullet_s, {a^\bullet_{s'}}^\dagger) = |q_s|^2 \bigl( \delta_{s' , s} - \delta_{s' , h-s}\bigl) \,,\\
&(a^\circ_s, {a^\circ_{s'}}^\dagger) = |q_s|^2 \bigl( \delta_{s' , s} + \delta_{s' , h-s}\bigl) \,.
\end{split}
\end{equation}
However, from the definition \eqref{eq:def1-as} we also have
\begin{equation}
\label{eq:sp2}
\begin{split}
&(a^\bullet_s, {a^\bullet_{s'}}^\dagger) = \sum_{a, b \in \bullet} q^a_s q^b_{s'} (\alpha_a , \alpha_b) = 2 \sum_{a \in \bullet} q^a_s q^a_{s'} \,,\\
&(a^\circ_s, {a^\circ_{s'}}^\dagger) = \sum_{a, b \in \circ} q^a_s q^b_{s'} (\alpha_a , \alpha_b) = 2 \sum_{a \in \circ} q^a_s q^a_{s'} \,,
\end{split}
\end{equation}
where we used the fact that the simple roots within a given set ($\bullet$ or $\circ$) are mutually orthogonal to each other and we set the root length to be $\sqrt{2}$. 
Thus combining \eqref{eq:sp1} and \eqref{eq:sp2} we obtain \eqref{eq:completeness-rel-qb} and~\eqref{eq:completeness-rel-qw}.
\end{proof}

\begin{lemma}
\label{theo1}
Given a semisimple Lie algebra $\mathfrak{g}$ associated with a simply-laced Dynkin diagram then $\forall$ $x$ and $y \in \mathbb{R}^r$ it holds that
\begin{equation}
\label{eq:completeness-relation}
\sum_{e=1}^r \sum^{h-1}_{p=0} (x, w^{-p} \gamma_e) (w^{-p} \gamma_e,y)= 2 h (x, y) \,.
\end{equation}
\end{lemma}
\begin{proof}
    Let LHS be the left-hand side of \eqref{eq:completeness-relation}. By inserting two completeness relations of the type~\eqref{eq:completeness-zs} and using the projections \eqref{eq:projections-roots} we can write
\begin{equation}
    \label{eq:LHS_lemma_completness}
    \begin{split}
        \text{LHS} & = \sum_{e=1}^r \sum^{h-1}_{p=0} \sum_{s,s'} \frac{2}{h} \frac{(x,z_{h-s})(z_{s'}, y)}{|q_{s'}||q_{h-s}|}\ q_{s'}^e q_{h-s}^e e^{\frac{i \pi}{h} u_{\gamma_e} (s'-s+1)}\ e^{\frac{2 i p}{h} (s - s')} \,,
    \end{split}
    \end{equation}
where $u_{\gamma_e}=-1$ if $e\in\bullet$ and $u_{\gamma_e}=0$ if $e \in \circ$, so to take into account the extra phase in the second line of~\eqref{eq:projections-roots}. 
Noting that
$$
\sum^{h-1}_{p=0} e^{\frac{2 i p}{h} (s - s')} = h \delta_{s s'} \,,
$$
we obtain
\begin{equation}
        \text{LHS}= \sum_{e=1}^r \sum_{s} 2 \frac{(x,z_{h-s})(z_{s}, y)}{|q_s|^2}\ q_{s}^e q_{s}^e = 2 \sum_{s} (x,z_{h-s})(z_{s}, y) = 2 h (x,y) \, .
\end{equation}
In the expression above we used the relations in~\eqref{eq:reality-conditions} to transform $q_{h-s}$ into $q_s$ and the completeness relation \eqref{eq:completeness-zs}.
\end{proof}

\begin{lemma}
\label{property_strange_relations}
$\forall$ $a\in \bullet$, $b\in \bullet$ and $l \in \mathbb{Z}$ it holds that
    \begin{equation}
    \label{eq:rel_white_white}
    -2 l (\gamma_b, w^{l} \gamma_a) + \sum^r_{e=1} \sum^{l-1}_{q=0} (\gamma_e, w^{l-q} \gamma_a) (\gamma_e, w^{-q} \gamma_b)=  ( (1+w^{-1})\lambda_a , w^{-l} \gamma_b) \,.
    \end{equation}
\end{lemma}

\begin{proof}
    For convenience let us define the left-hand side and the right-hand side of the relation by
    \begin{equation}
        \begin{split}
            \text{LHS} & \equiv -2 l (\gamma_b, w^{l} \gamma_a) + \sum^r_{e=1} \sum^{l-1}_{q=0} (\gamma_e, w^{l-q} \gamma_a) (\gamma_e, w^{-q} \gamma_b), \\
            \text{RHS} & = ( (1+w^{-1})\lambda_a , w^{-l} \gamma_b).
        \end{split}
    \end{equation}
   Let us start with the RHS, which is the simpler quantity to compute. We first write $\lambda_a$ as a function of $\gamma_a$ by inverting \eqref{eq:roots_weigths_connection} and then we use the completeness relation \eqref{eq:completeness-zs}. We then get:
\begin{equation}
\begin{split}
\text{RHS}&=  ( (1+w^{-1})\lambda_a , w^{-l} \gamma_b) = ( (1+w^{-1}) (1- w^{-1})^{-1}\gamma_a , w^{-l} \gamma_b)\\
&= \frac{1}{h} \sum_s (1+e^{-2 i \theta_s}) (1-e^{-2 i \theta_s})^{-1} (\gamma_a, z_{h-s}) (z_s, w^{-l} \gamma_b)= \frac{2}{i} \sum_s \cot (\theta_s) \frac{q^a_s q^b_s}{|q_s|^2} e^{2 i l \theta_s} \,.
\end{split}
\end{equation}
The LHS can be decomposed on the basis of eigenvectors of the Coxeter element in a similar way, leading to
\begin{multline}
    \text{LHS}=  -4 l \sum_s e^{2 i l \theta_s}\frac{q^a_s q^b_s}{|q_s|^2} \\
    + 4  \sum_{s, s'}  e^{2i l \theta_s} \frac{q^a_s q^b_{s'}}{|q_s|^2 |q_{s'}|^2} \sum_{q=0}^{l-1} e^{i (\theta_{s'} - \theta_{s}) (2q+1)} \left(\sum_{e \in \circ} q^e_s q^e_{s'}  + \sum_{e \in \bullet} q^e_s q^e_{s'} e^{i(\theta_{s} - \theta_{s'})} \right).
\end{multline}
Using the relations in~\eqref{eq:completeness-rel-qb} and~\eqref{eq:completeness-rel-qw} this expression becomes
\begin{multline}
\label{eq:intermediate_step_strange_rel}
    \text{LHS}=  -4 l \sum_s e^{2 i l \theta_s}\frac{q^a_s q^b_s}{|q_s|^2} \\
    + 2  \sum_{s, s'}  e^{2i l \theta_s} \frac{q^a_s q^b_{s'}}{ |q_{s'}|^2} \sum_{q=0}^{l-1} e^{i (\theta_{s'} - \theta_{s}) (2q+1)} \bigl[ 2 \delta_{s' , s} + \delta_{s', h-s} (e^{2i \theta_s} + 1) \bigl].
\end{multline}
The term proportional to $\delta_{s' s}$ in the second row of~\eqref{eq:intermediate_step_strange_rel} cancels the term in the first row and after using
\begin{equation}
\theta_{h-s} = \frac{\pi}{h} (h-s) = \pi- \theta_s
\end{equation}
and the relations in~\eqref{eq:reality-conditions} we obtain
\begin{equation}
\begin{split}
\text{LHS}&=   2  \sum_{s}  e^{2i l \theta_s} \frac{q^a_s q^b_{s}}{ |q_{s}|^2} (e^{-2i \theta_s} + 1) \sum_{q=0}^{l-1} e^{- 4 i q \theta_{s}}\\
&=2 \sum_{s}  \Bigl(e^{2 i l \theta_s} - e^{-2 i l \theta_s} \Bigl) \frac{q^a_s q^b_{s}}{ |q_{s}|^2} \frac{1}{1- e^{-2 i \theta_s}} .
 \end{split}
\end{equation}
We can now change $s \to h-s$ for the part of the sum containing $e^{-2 i l \theta_s}$. In this manner $e^{-2 i l \theta_s} \to e^{2 i l \theta_s}$ and using the relations in~\eqref{eq:reality-conditions} we obtain
\begin{equation}
\text{LHS}=2 \sum_{s} e^{2 i l \theta_s} \frac{q^a_s q^b_{s}}{ |q_{s}|^2} \Bigl(\frac{1}{1- e^{-2 i \theta_s}} - \frac{1}{1- e^{2 i \theta_s}} \Bigl) = \frac{2}{i} \sum_{s} e^{2 i l \theta_s} \frac{q^a_s q^b_{s}}{ |q_{s}|^2} \cot(\theta_s)
\end{equation}
which exactly matches the RHS.
\end{proof}

\begin{lemma}
    $\forall$ $a \in \circ$, $b \in \bullet$ and $l \in \mathbb{Z}$ it holds that
    \begin{multline}
    \label{eq:rel_white_black}
    -(2l+1) (\gamma_b, w^{l} \gamma_a) + \sum^r_{e=1} \sum^{l-1}_{q=0} (\gamma_e, w^{l-q} \gamma_a) (\gamma_e, w^{-q} \gamma_b)\\
    +\sum_{e \in \bullet} (\gamma_e, \gamma_a) (\gamma_e, w^{-l} \gamma_b)=( (1+w^{-1})\lambda_a , w^{-l} \gamma_b) \,.
    \end{multline}
\end{lemma}

\begin{proof}
 Let LHS and RHS be the left-hand side and right-hand side of \eqref{eq:rel_white_black}, as usual. Proceeding as in the proof of Lemma~\ref{property_strange_relations} it is easy to show that
\begin{equation}
\label{eq:RHS_black_white_rel_app}
\text{RHS}= \frac{2}{i} \sum_s \cot (\theta_s) \frac{q^a_s q^b_s }{|q_s|^2} e^{ i (2l+1) \theta_s} \,.
\end{equation}
Decomposing also the LHS on the basis of eigenvectors of the Coxeter element we obtain
\begin{equation}
\label{eq:int_step_black_white_app}
\begin{split}
    \text{LHS}=  &-2(2 l+1) \sum_s e^{i \theta_s (2l +1)}\frac{q^a_s q^b_s}{|q_s|^2} \\
    &+ 2  \sum_{s, s'}  e^{2i l \theta_s} \frac{q^a_s q^b_{s'}}{ |q_{s'}|^2} \sum_{q=0}^{l-1} e^{2q i (\theta_{s'} - \theta_{s})} e^{i \theta_{s'}} \Bigl( 2 \delta_{s' s} +\delta_{s', h-s} (e^{2i \theta_s} + 1) \Bigl) \\
    &+ 2  \sum_{s, s'}   \frac{q^a_s q^b_{s'}}{ |q_{s'}|^2}  e^{(2 l+1) i \theta_{s'} }  \bigl( \delta_{s' s} + \delta_{s' h-s} e^{2i \theta_s} \bigl)  .
    \end{split}
\end{equation}
As before, the contributions proportional to $\delta_{s' s}$ cancel the term in the first-row of~\eqref{eq:int_step_black_white_app} leading to
 \begin{equation}
\begin{split}
    \text{LHS}=  & 2  \sum_{s}  e^{i \theta_s(2 l -1) } \frac{q^a_s q^b_{s}}{ |q_{s}|^2} \bigl( e^{2 i \theta_s}+1 \bigl) \sum^{l-1}_{q=0} e^{-4 i q \theta_s} + 2  \sum_{s}  e^{-i \theta_s (2l-1) } \frac{q^a_s q^b_{s}}{ |q_{s}|^2}\\
    =& 2  \sum_{s}  e^{i \theta_s(2 l +1) } \frac{q^a_s q^b_{s}}{ |q_{s}|^2}  \frac{1-e^{-4 i l \theta_s}}{1-e^{-2 i \theta_s}}+ 2  \sum_{s}  e^{i \theta_s (2l-1) } \frac{q^a_s q^b_{s}}{ |q_{s}|^2} \,.
    \end{split}
\end{equation}
Using the trick of changing $s \to h-s$ when required to transform $e^{-2i l \theta_s} \to e^{2i l \theta_s}$ we obtain
\begin{equation}
     \text{LHS}= 2  \sum_{s}  e^{2 i l \theta_s} \frac{q^a_s q^b_{s}}{ |q_{s}|^2} \left( \frac{e^{i \theta_s}}{1-e^{-2i \theta_s}} - \frac{e^{-i \theta_s}}{1-e^{2i \theta_s}} + e^{-i \theta_s}\right) = \frac{2}{i}  \sum_{s}  e^{i (2l+1) \theta_s} \frac{q^a_s q^b_{s}}{ |q_{s}|^2}  \cot(\theta_s) \,,
\end{equation}
which exactly matches~\eqref{eq:RHS_black_white_rel_app}.
\end{proof}

Since under exchanging $\circ\leftrightarrow\bullet$ the Coxeter element transforms as $w \to w^{-1}$, then identities for $\{a \in \circ, b \in \circ\}$ and $\{a \in \bullet, b \in \circ\}$ can be obtained from the identities above by replacing $w$ with $w^{-1}$.

\bibliographystyle{JHEP}
\bibliography{refs}

\end{document}